%% file: random_planar_curves.tex
\documentclass[12pt,a4paper,notitlepage]{article}
\usepackage[utf8]{inputenc}
\usepackage[english]{babel}

\usepackage{bm,bbm}
\usepackage{amsbsy,amsmath,amsthm, amssymb,amsfonts,MnSymbol}
\usepackage{graphicx}
\usepackage{subfigure}
\usepackage{verbatim}
\usepackage{accents}
\usepackage{stmaryrd}

\usepackage[numbers]{natbib}

\usepackage[	pdfauthor={},%
			pdftitle={},%
			pdftex]{hyperref}
\hypersetup{	colorlinks,%
			citecolor=black,%
			filecolor=black,%
			linkcolor=black,%
			urlcolor=black%
			}

\newcommand{\Z}{\mathbb{Z}}
\newcommand{\N}{\mathbb{N}}
\newcommand{\C}{\mathbb{C}}
\newcommand{\R}{\mathbb{R}}

\newcommand{\half}{\mathbb{H}}
\newcommand{\disc}{\mathbb{D}}

\renewcommand{\P}{\mathbb{P}}

\newcommand{\E}{\mathbb{E}}

\newcommand{\ind}{\mathbbm{1}}

\newcommand{\F}{\mathcal{F}}
\newcommand{\B}{\mathcal{B}}

\newcommand{\OO}{\mathcal{O}}
\newcommand{\oo}{o}
\newcommand{\de}{\mathrm{d}}

\DeclareMathOperator{\dist}{dist}
\DeclareMathOperator{\diam}{diam}

\newcommand{\hcap}{\mathrm{cap}_\half}

\newcommand{\usigma}{{\underline{\sigma}}}

\newcommand{\eps}{\varepsilon}
\DeclareMathOperator{\imag}{Im}
\DeclareMathOperator{\real}{Re}
\renewcommand{\coloneq}{\mathrel{\mathop:}=}

\theoremstyle{plain}
\newtheorem{theorem}{Theorem}[section]
\newtheorem{lemma}[theorem]{Lemma}
\newtheorem{proposition}[theorem]{Proposition}
\newtheorem{corollary}[theorem]{Corollary}

\theoremstyle{definition}
\newtheorem{definition}[theorem]{Definition}
\newtheorem{condition}{Condition}

\newtheorem{conditionc}{Condition}

\theoremstyle{remark}
\newtheorem{remark}[theorem]{Remark}

\usepackage[hmargin=3cm,vmargin={2.5cm,2cm},includefoot]{geometry}

\newcommand{\enustyo}{%
\renewcommand{\theenumi}{\arabic{enumi}}
\renewcommand{\labelenumi}{\arabic{enumi}.}
}

\newcommand{\enustyii}{%
\renewcommand{\theenumi}{(\roman{enumi})}
\renewcommand{\labelenumi}{\theenumi}
}

\newcommand{\tl}{\mathbb{T}}
\newcommand{\hl}{\mathbb{T}'}

\newcommand{\len}{\Lambda}

\newcommand{\Prob}{\mathrm{Prob}}

\newcommand{\xs}{X_{\mathrm{simple}}}

\newcommand{\psd}{p_{\mathrm{sd}}}

\newcommand{\dd}{\mathrel{\mathop:}}

\newcommand{\elen}{m}

\def\br#1{\left(#1\right)}
\def\brs#1{\left\{#1\right\}}
\def\bra#1{\left|#1\right|}

\newcommand{\discr}[1]{#1^\#}
\newcommand{\discrui}[2]{#1^{\#,#2}}

\newcommand{\conti}[1]{#1^\medsquare}

\begin{document}

\title{Random curves, scaling limits and Loewner evolutions}
\author{Antti Kemppainen\footnote{Email: Antti.H.Kemppainen@helsinki.fi}
\and Stanislav Smirnov\footnote{Email: Stanislav.Smirnov@unige.ch}
}
\date{\today}

\maketitle

\begin{abstract}
In this paper, we provide a framework of estimates for describing 2D scaling
limits by Schramm's SLE curves. In particular, we show that a weak estimate
on the probability of an annulus crossing implies that a random curve arising
from a statistical mechanics model will have scaling limits and those will be
well-described by Loewner evolutions with random driving forces.
Interestingly, our proofs indicate that existence of a nondegenerate
observable with a conformally-invariant scaling limit seems sufficient to
deduce the required condition.

Our paper serves as an important step in establishing the convergence of Ising and
FK Ising interfaces to SLE curves, moreover, the setup is adapted to branching interface trees, 
conjecturally describing the full interface picture by a collection of branching SLEs.
\end{abstract}

\clearpage

\tableofcontents

\clearpage

\input{random_planar_curves_sec1-2.tex}

\input{random_planar_curves_sec3.tex}

\input{random_planar_curves_sec4.tex}

\appendix

\input{random_planar_curves_appendix.tex}

\section*{Acknowledgments}

A.K. was financially supported by Academy of Finland.
Both authors were supported 
by the Swiss NSF, NCCR SwissMAP, EU RTN CODY and ERC AG COMPASP.

%
\nocite{pommerenke-1992-}
\nocite{lawler-2005-}
\nocite{rohde-schramm-2005-}
\nocite{grimmett-2006-}


\input{random_planar_curves.bbl}
\end{document}

%% file: random_planar_curves_sec1-2.tex

\section{Introduction}

%

Oded Schramm's introduction of SLE as the only possible conformally invariant scaling limit of interfaces
has led to much progress 
in our understanding of 2D lattice models at criticality.
For several of them it was shown that interfaces
(domain wall boundaries) indeed
converge to Schramm's SLE curves as the lattice mesh tends to zero
\cite{smirnov-2001-,smirnov-2001-p,smirnov-2003-,lawler-schramm-werner-2004-,
schramm-sheffield-2005-,smirnov-2010-,camia-newman-2007-,schramm-sheffield-2009-}.

All the existing proofs start  by relating some observable to
a discrete harmonic or holomorphic function with appropriate boundary values and describing its scaling limit
in terms of its continuous counterpart.
Conformal invariance of the latter allowed then to construct
the scaling limit of the interface itself by sampling the observable as it is drawn.
The major technical problem in doing so is how to deduce 
from some weaker notions
the strong convergence of interfaces, i.e., the convergence in law with respect to the topology
induced by the uniform norm to the space of continuous curves, which are only defined up to reparametrizations.
So far two routes have been suggested:
first to prove the convergence of the driving process in Loewner characterization, 
and then improve it to convergence of curves, cf. \cite{lawler-schramm-werner-2004-};
or first establish some sort of precompactness for laws of discrete interfaces,
and then prove that any sub-sequential scaling limit is in fact
an SLE, cf. \cite{smirnov-2003-}.

We will lay framework for both  approaches, showing that a rather weak hypotheses
is sufficient to conclude that 
an interface has sub-sequential scaling limits, but also
that they can be described almost surely by Loewner evolutions.
We build upon an earlier work of Aizenman and Burchard \cite{aizenman-burchard-1999-},
but draw stronger conclusions from similar conditions, and
also reformulate them in several geometric as well as conformally invariant ways.

At the end we check this condition for a number of lattice models.
In particular, this paper serves as an important step in establishing the convergence of Ising and FK Ising interfaces 
\cite{chelkak-duminil-hongler-kemppainen-smirnov-2013-}.
Interestingly, our proofs indicate that existence of a non-degenerate observable with a conformally invariant scaling limit seems sufficient to deduce the required condition.
So we believe, that the main obstacle to establish convergence to SLE of interfaces in various
models lies in finding a (exactly or approximately) discrete holomorphic observable.
Our techniques also apply to interfaces in massive versions of lattice models, as in \cite{makarov-smirnov-2010-}.
In particular, the proofs for loop-erased random walk and harmonic explorer we include below can be modified
to their massive counterparts, as those have similar martingale observables \cite{makarov-smirnov-2010-}.

Moreover, our setup is adapted to branching interface trees, conjecturally converging to branching SLE$(\kappa,\kappa-6)$,
cf \cite{sheffield-2009-}.
We are preparing a follow-up \cite{kemppainen-smirnov-2013-}, 
which will exploit this in the context of the critical FK Ising model.
In the percolation case a construction was proposed in  \cite{camia-newman-2006-},
also using the Aizenman--Burchard work.

Another approach to a single interface was proposed by Sheffield and Sun \cite{sheffield-sun-2012-}.
They ask for milder condition on the curve, but require simultaneous convergence of the Loewner evolution driving force
when the curve is followed in two opposite directions towards generic targets.
The latter property is missing in many of the important situations we have in mind,
like convergence of the full interface tree.

The authors would like to thank the anonymous referees for their valuable comments
as well as Alexander Glazman and Hugo Duminil-Copin for reading through
and commenting on the preliminary versions of the paper. Their efforts to increase the
quality of this work are highly appreciated.

\subsection{The setup and the assumptions}

Our paper is concerned with sequences of random planar curves
and different conditions sufficient to establish their precompactness. 

We start with a probability measure $\P$ on the set $X(\C)$ of planar curves,
having in mind an interface (a domain wall boundary) in some lattice model of statistical physics
or a self-avoiding random trajectory on a lattice.
By a \emph{planar curve} we mean a continuous mapping $\gamma: [0,1] \to \C$.
The resulting space $X(\C)$ is endowed with the usual supremum metric with minimum taken over all reparameterizations,
which is therefore parameterization-independent, see Section~\ref{ssec: curves}.
Then we consider $X(\C)$ as a measurable space with Borel $\sigma$-algebra.
For any domain $V \subset \C$, let $\xs(V)$ be the set of Jordan curves $\gamma: [0,1] \to \overline{V}$
such that $\gamma(0,1) \subset V$. Note that the end points are allowed to lie on the boundary.

\begin{figure}[tbh]
\centering
\subfigure[Typical setup: a random curve is defined on a lattice approximation of $U$
and is connecting two boundary points $a$ and $b$.] 
{
	\label{sfig: domains and dmp a}
	\includegraphics[scale=.4]
{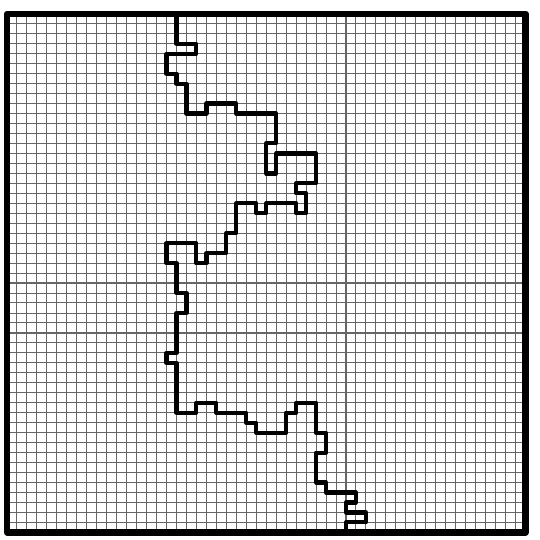}
} 
\hspace{0.6cm}
\subfigure[The same random curve after a conformal transformation to $\disc$ taking $a$ and $b$ to $-1$ and $+1$, 
respectively.]
{
	\label{sfig: domains and dmp b}
	\includegraphics[scale=.4]
{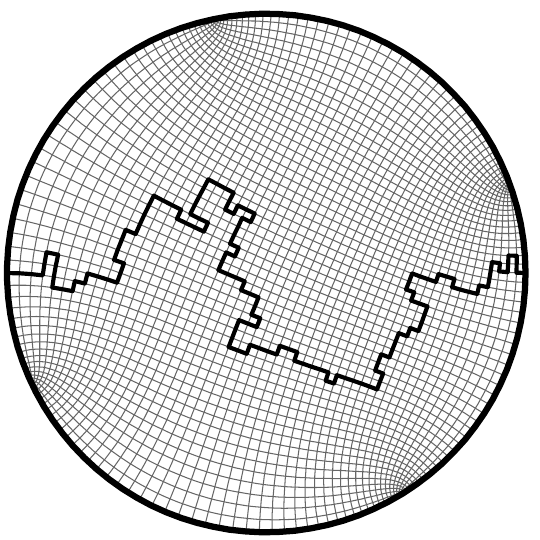}
}
\hspace{0.6cm}
\subfigure[Under the domain Markov property the curve conditioned on its beginning part has
the same law as the one in the domain with the initial segment removed.]
{
	\label{sfig: domains and dmp c}
	\includegraphics[scale=.4]
{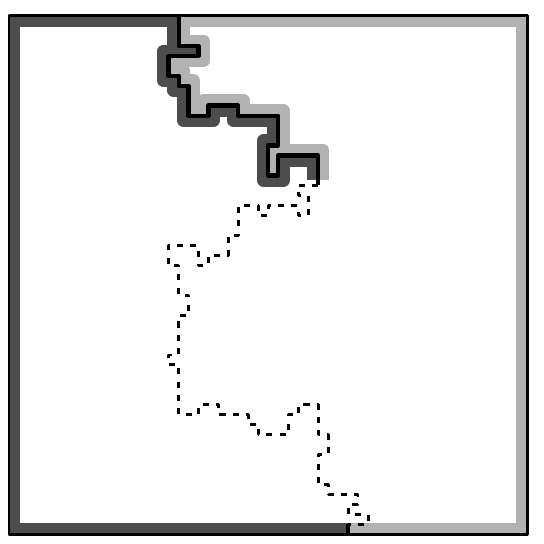}
}  
\caption{The assumptions of the main theorem are often 
easier to verify in the domain where the curve is originally defined (a) and the slit domains
appearing as we trace the curve (c).
Nevertheless, to set up the Loewner evolution we need to uniformize conformally to a fixed domain, e.g. the unit disc (b). 
Figure~(c) illustrates the domain Markov property under which it is possible to verify the simpler ``time $0$'' condition (presented in this section)
instead of its conditional versions (see Section~\ref{ssec: four conditions}).
}
\label{sfig: domains and dmp}
\end{figure}

Typically, the random curves we want to consider connect two boundary points $a,b \in \partial U$
in a simply connected domain $U$. 
Also it is possible to assume that the random curve is (almost surely) simple, because
the curve is usually defined on a lattice with small but finite lattice mesh without ``transversal'' self-intersections.
Therefore, by slightly perturbing the lattice and the curve it is possible to remove self-intersections. 
The main theorem of this paper involves the
Loewner equation, and consequently the curves have to be either simple or non-self-traversing, i.e.,
curves that are limits of sequences of simple curves. 

While we work with different domains $U$, we still prefer to restate our conclusions for a fixed domain.
Thus we encode the domain $U$ and the curve end points $a,b \in \partial U$ 
by a conformal transformation $\phi$ from $U$ onto the unit disc $\disc= \{ z\in \C \;:\; |z|<1 \}$.
The domain $U=U(\phi)$ is then the domain of definition of $\phi$ and the points $a$ and $b$
are preimages $\phi^{-1}(-1)$ and $\phi^{-1}(1)$, respectively, if necessary define these in the sense of prime ends.

Because of the above reasons the first fundamental object in our study is 
a \emph{pair} $(\phi,\P)$ where $\phi$ is a \emph{conformal map}
and $\P$ is a \emph{probability measure on curves} with the following restrictions: Given
$\phi$ we define the domain $U=U(\phi)$ to be the domain of definition of $\phi$ and we require that
$\phi$ is a conformal map from $U$ onto the unit disc $\disc$.
Therefore $U$ is a simply connected domain other than $\C$.
We require also that $\P$ is supported on (a closed subset of)
\begin{equation} \label{eq: curves on u}
\left\{ \gamma \in \xs(U) \;:\; 
   \begin{gathered}
   \text{the beginning and end point of} \\ 
   \text{$\phi (\gamma)$ are $-1$ and $+1$, respectively}   
   \end{gathered} \right\}.
\end{equation}
The second fundamental object in our study is 
\emph{some collection} $\Sigma$ of pairs $(\phi,\P)$ satisfying the above restrictions.

Because the spaces involved are metrizable, when discussing convergence
we may always think of $\Sigma$ as a sequence $( (\phi_n,\P_n) )_{n \in \N}$. 
In applications, we often have in mind a sequence of interfaces for the same
lattice model but with varying lattice mesh 
$\delta_n \searrow 0$: then each $\P_n$ is supported on curves defined on the $\delta_n$-mesh lattice.
The main reason for working with the more abstract family compared to a sequence is to simplify the notation.
If the set in \eqref{eq: curves on u} is non-empty, which is assumed, then there are in fact plenty of such curves,
see Corollary~2.17 in \cite{pommerenke-1992-}.

We uniformize by a disk $\disc$ to work with a bounded domain.
As we show later in the paper, our conditions are conformally invariant, so the choice of a particular uniformization domain is not important.

For any $0<r<R$ and any point $z_0 \in \C$, denote the annulus of radii $r$ and $R$ 
centered at $z_0$ by $A(z_0,r,R)$:
\begin{equation}
A(z_0,r,R) = \{ z \in \C \;:\; r < |z - z_0| < R \} .
\end{equation}
We call the quantity $(1/2\pi) \, \log \left( R/r\right)$ the modulus of the annulus 
$A(z_0,r,R)$.
The following definition makes speaking about crossing of annuli precise.

\begin{definition}\label{def: crossing}
For a curve $\gamma:[T_0,T_1] \to \C$ and an annulus $A=A(z_0,r,R)$, $\gamma$ is said to 
be a \emph{crossing} of the annulus $A$ if both $\gamma(T_0)$ and $\gamma(T_1)$ lie outside $A$ 
and they are in the different components of $\C \setminus A$. A curve
$\gamma$ is said to \emph{make a crossing} of the annulus $A$ if
there is a subcurve which is a crossing of $A$. 
A \emph{minimal crossing} of the annulus $A$ is a crossing which doesn't have genuine subcrossings.
\end{definition}

\begin{figure}[tbh!]
\centering
\subfigure[Unforced crossing: the component of the annulus is not
disconnecting $a$ and $b$. It is possible that the curve avoids the set. In this picture $A^u$
is the entire half-annulus.] 
{
	\label{sfig: crossings a}
	\includegraphics[scale=.75]
{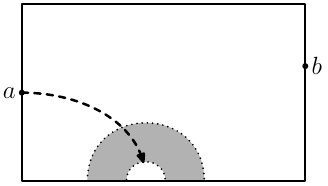}
} 
\hspace{0.4cm}
\subfigure[Forced crossing: the component of the annulus
disconnects $a$ and $b$ and does it in the way, that every curve connecting $a$ and $b$
has to cross the annulus at least once. In this picture $A^u$ is empty.]
{
	\label{sfig: crossings b}
	\includegraphics[scale=.75]
{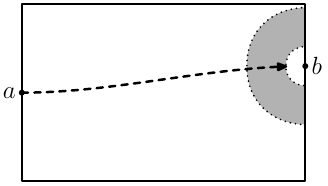}
}
\hspace{0.4cm}
\subfigure[There is an ambiguous case which resembles more either one of the previous two cases depending
on the geometry. In this case the component of the annulus separates $a$ and $b$, but there are some
curves from $a$ to $b$ in $U$ which don't cross the annulus. In this picture $A^u$ is the small top part
of the half-annulus.] 
{
	\label{sfig: crossings c}
	\includegraphics[scale=.75]
{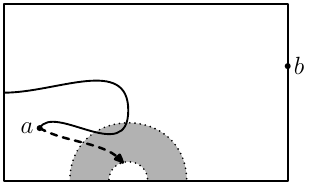}
}  
\subfigure[Unforced crossing of a topological quadrilateral (quad): the quad is not disconnecting 
$a$ and $b$.] 
{
	\label{sfig: crossings d}
	\includegraphics[scale=.75]
{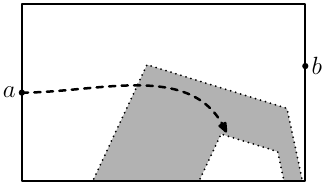}
} 
\hspace{0.4cm}
\subfigure[The quads we consider have two of their sides on the boundary and two in the interior of the domain.
We usually denote the sides by $S_0, S_2 \subset  U$ and $S_1,S_3 \subset \partial U$. The set $V$ is the interior
of the quad. There exists a unique number $L>0$ (and a unique conformal map) such that the quad can be mapped 
conformally onto the rectangle $(0,L) \times (0,1)$ so that the sides of the quad get mapped to the sides of the rectangle
and $S_0$ gets mapped onto $\{0\} \times{} {[0,1]}$.]
{
	\label{sfig: crossings e}
	\includegraphics[scale=.75]
{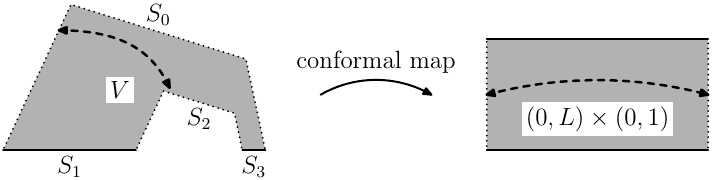}
}
\caption{The general idea of Condition~\ref{cond: annulus} is that an event of an unforced crossing 
has uniformly positive probability to fail. In the figures~\protect\subref{sfig: crossings a}--\protect\subref{sfig: crossings c}
the solid line is the boundary of the domain,
the dotted lines are the boundaries of the annulus and the dashed lines refer to the crossing event we are considering.
The conformally invariant version of this is Condition~\ref{cond: conf} which is formulated using topological quadrilaterals.
The figure \protect\subref{sfig: crossings d} corresponds to the figure~\protect\subref{sfig: crossings a}
in this latter setting. 
} 
\label{fig: crossings}
\end{figure}

We cannot require that crossing any fixed annulus has a small probability under $\P$:
indeed, annuli centered at $a$ or at $b$ have to be crossed at least once. 
For that reason we introduce the following definition
for a fixed simply connected domain $U$ and an annulus $A=A(z_0,r,R)$ which is allowed to vary.
If $\partial B(z_0,r) \cap \partial U = \emptyset$ define $A^u= \emptyset$, otherwise
\begin{equation}\label{eq: definition Au}
A^u = \left\{ z \in U \cap A \,:\, 
   \begin{gathered}
   \text{the connected component of $z$ in $U \cap A$} \\
   \text{doesn't disconnect $a$ from $b$ in $U$}
   \end{gathered}
   \right\} .
\end{equation}
This reflects the idea explained in Figure~\ref{fig: crossings}.

The main theorem is proven under a set of equivalent conditions. In this section, two simplified versions
are presented. They are so called \emph{time $0$ conditions} which imply 
the stronger \emph{conditional} versions
if our random curves satisfy the domain Markov property, cf. Figure~\ref{sfig: domains and dmp c}. 
It should be noted that even in physically interesting situations the latter might fail,
so the conditions presented in the section~\ref{ssec: four conditions} 
should be taken as the true assumptions of the main theorem.

\begin{condition}\label{cond: g const time zero}
The family $\Sigma$ is said to satisfy a 
\emph{geometric bound on an unforced crossing (at time zero)}
if there exists $C >1$ such that
for any $(\phi,\P) \in \Sigma$
and for any annulus $A=A(z_0,r,R)$ with $0 < C \, r \leq R$, 
\begin{equation}\label{ie: cond g1}
\P \left( \text{$\gamma$ makes a crossing of $A$ which is contained in $\overline{A^u}$} \, \right) \leq \frac{1}{2}  . 
\end{equation}
\end{condition}

We stress already at this point that the constant $1/2$ on the right-hand side of \eqref{ie: cond g1}
or in similar bounds
is arbitrary and could be replaced by any other constant strictly less than one. We will demonstrate this  
in Corollary~\ref{cor: equiv conditions onehalf}.

A \emph{topological quadrilateral} $Q=(V; S_k, k=0,1,2,3)$ consists a domain $V$ which is homeomorphic to
a square in a way that the boundary arcs $S_k$, $k=0,1,2,3$, are in counterclockwise order
and correspond to the four edges of the square. 
There exists a unique positive $L$ and a conformal map from $Q$ 
onto a rectangle $[0,L]\times[0,1]$ mapping $S_k$ to the four edges
of the rectangle with image of $S_0$ being $\{0\}\times[0,1]$.
The number $L$ is called the \emph{modulus} of (or the \emph{extremal length the curve family} joining the opposite sides of) 
$Q$ and we will denote it by $\elen(Q)$. 

We often consider a topological quadrilateral $Q=(V; S_k, k=0,1,2,3)$ which is lying on the boundary in the sense that
$S_1 \cup S_3 \subset \partial U$ while $S_0 \cup S_2 \subset U$ 
--- this idea corresponds to the condition imposed when we defined $A^u$. 
For this type of topological quadrilateral we say that a curve $\gamma: [T_0,T_1] \to \C$ \emph{crosses} $Q$
in the domain $U$ if 
there is a subinterval $[t_0,t_1] \subset [T_0,T_1]$ such that $\gamma (t_0,t_1) \subset V$, but
$\gamma [t_0,t_1]$ intersects both $S_0$ and $S_2$.
The other notions of Definition~\ref{def: crossing} are extended to the topological quadrilaterals
in the same way.
The following is the first conformally invariant version of our conditions, formulated in terms of
topological quadrilaterals.

\begin{conditionc}\label{cond: c const time zero}
The family $\Sigma$ is said to satisfy a 
\emph{conformal bound on an unforced crossing (at time zero)}
if there exists $M >0$ such that
for any $(\phi,\P) \in \Sigma$
and for any topological quadrilateral $Q$ with $V(Q) \subset U$, $S_1 \cup S_3 \subset \partial U$ and
$\elen(Q) \geq M$
\begin{equation}
\P \left( \text{$\gamma$ makes a crossing of $Q$} \right) \leq \frac{1}{2}  . 
\end{equation}
\end{conditionc}

\begin{remark}
In percolation type models of statistical physics including the random cluster models, 
this type of crossing events are the most central objects of study.
\end{remark}

\begin{remark}
Notice that depending on the point of view, either one of the conditions can appear stronger than
the other one.
In Condition~\ref{cond: g const time zero} we require 
that the bound holds for all annuli with large modulus and
simultaneously for \emph{all} components of $A^u$, 
whereas in Condition~\ref{cond: c const time zero} the bound holds 
for all topological quadrilaterals with large modulus and
for its \emph{single} (only) component. 
On the other hand, the set of topological quadrilaterals
is bigger than the set of topological quadrilaterals $Q$ whose boundary arcs $S_0(Q)$ and $S_2(Q)$
are subsets of different boundary components of some annulus and $V(Q)$ is subset of that annulus. The latter set is 
the set of shapes relevant in Condition~\ref{cond: g const time zero}, at least naively speaking.
\end{remark}

\subsection{Main theorem}

The main results of this article will be on the tightness of certain families of probability measures
and on the convergence of probability measures in the weak sense.
Hence let's first recall the following definitions.

\begin{definition}
Let $Y$ be a metric space and $\B_Y$ its Borel $\sigma$-algebra. 

If $\Sigma_0$ is a collection of probability measures on $(Y,\B_Y)$, then
a random variable $f : Y \to \R$ is said to be \emph{tight} or \emph{stochastically bounded} 
in $\Sigma_0$ if and only if
for each $\eps>0$ there is $M>0$ such that $\P( |f| \leq M ) \geq 1-\eps$ for all $\P \in \Sigma_0$.

A collection $\Sigma_0$ of probability measures on $(Y,\B_Y)$ is said to be 
\emph{tight} if for each $\eps>0$ there exists a compact
set $K \subset Y$ so that $\P(K) \geq 1 - \eps$ for any $\P \in \Sigma_0$. 
\end{definition}

For the background in the weak convergence of probability measures the reader should see for example \cite{billingsley-1999-}.
Prohorov's theorem states that a family of probability measures is relatively 
compact if it is tight, see Theorem~5.1 in \cite{billingsley-1999-}. 
Moreover, in a separable and complete metric space relative compactness and tightness are equivalent.

Denote by $\phi \P$ the pushforward of $\P$ by $\phi$ defined by
\begin{equation}
( \phi \P )(A) = \P( \phi^{-1}(A) )
\end{equation}
for any measurable $A \subset \xs(\disc)$.
In other words $\phi\P$ is the law of the random curve $\phi(\gamma)$.
Given a family $\Sigma$ as above, define the family of pushforward measures 
\begin{equation}\label{eq: def sigma disc}
\Sigma_\disc = \left\{ \phi \P \;:\; (\phi,\P)\in \Sigma  \right\} .
\end{equation}
The family $\Sigma_\disc$ consist of measures  on the curves $\xs(\disc)$ connecting $-1$ to $1$.

Fix a conformal map
\begin{equation}\label{eq: def Phi}
\Phi(z) = i \frac{z+1}{1-z} 
\end{equation}
which takes $\disc$ onto the \emph{upper half-plane} $\half = \{ z \in \C \,:\, \imag z>0\}$. 
Note that if $\gamma$ is distributed according to $\P \in \Sigma_\disc$,
then $\tilde{\gamma} = \Phi(\gamma)$ is a simple curve in the upper half-plane
slightly extending the definition of $\xs(\half)$, namely,
$\tilde{\gamma}$ is simple 
with $\tilde{\gamma}(0) = 0 \in \R$, $\tilde{\gamma}((0,1)) \subset \half$
and $|\gamma(t)| \to \infty$ as $t \to 1$.
Therefore by the results of Appendix~\ref{ssec: sle}, if $\tilde{\gamma}$ is parametrized with the half-plane capacity,
then it has a continuous driving term $W=W_\gamma:\R_+ \to \R$. 
As a convention the driving term or process of a curve or a random curve in $\disc$
means the driving term or process in $\half$ after the transformation $\Phi$
and using the half-plane capacity parametrization.

The following theorem and its reformulation, Proposition~\ref{prop: reformulation of main theorem},
are the main results of this paper.
Note that the following theorem concerns with $\Sigma_\disc$. 
The proof will be presented in Section~\ref{sec: proof main}. It is divided into three independent steps
each in its own subsection and the actual proof is then presented in Section~\ref{ssec: proof main}. 
See Section~\ref{ssec: four conditions}
for the exact assumptions of the theorem, namely, Condition~\ref{def: b unf crossing}.
It should be noted that when the random curve has the domain Markov property, 
which is schematically defined in Figure~\ref{sfig: domains and dmp c},
Condition~\ref{cond: g const time zero} implies Condition~\ref{def: b unf crossing},
which is merely a conditional version of Condition~\ref{cond: g const time zero}.

\begin{theorem} \label{thm: main}
If the family $\Sigma$ of probability measures satisfies Condition~\ref{def: b unf crossing}, then
the family $\Sigma_\disc$ is tight
and therefore relatively compact in the topology of the weak convergence of probability measures on $(X,\B_X)$.
Furthermore if $\P_n \in \Sigma_\disc$ is converging weakly and the limit is denoted by $\P^*$ then
the following statements hold $\P^*$ almost surely
\begin{enumerate} \enustyii
\item \label{enui: main thm a} the point $1$ is not a double point, 
i.e., $\gamma(t) = 1$ only if $t=1$, \label{ei: main b simple}
\item \label{enui: main thm b} there exists $\beta>0$ such that $\gamma$ has 
a H\"older continuous parametrization for the H\"older exponent $\beta$,
\item \label{enui: main thm c} the tip $\gamma(t)$ of the curve lies on the boundary of the connected component of 
$\disc \setminus \gamma[0,t]$ containing $1$ (having the point $1$ on its boundary), for all $t$, \label{ei: main tip}
\item \label{enui: main thm d} if $\hat{K}_t$ is the hull of $\Phi(\gamma[0,t])$, 
then the capacity $\hcap (\hat{K}_t) \to \infty$ as $t \to 1$\label{ei: main hcap infty} 
\item \label{enui: main thm e} for any parametrization of $\gamma$  the capacity $t \mapsto \hcap (\hat{K}_t)$
   is strictly increasing and if $(K_t)_{t \in \R_+}$ is $(\hat{K}_t)_{t \in [0,1)}$ reparametrized with capacity, then
   the corresponding $g_t$ satisfies the Loewner equation with
   a driving process $(W_t)_{t \in \R_+}$ 
   which is H\"older continuous for any exponent $\alpha < 1/2$. \label{ei: main loewner}
\end{enumerate}
Furthermore, there exist constants $\eps>0$ and $C>0$ such that 
\begin{equation}\label{ie: main thm integrability}
\E^* \left[\exp\left(\eps \max_{s \in [0,t]} |W_s|/\sqrt{t} \right)\right] \leq C
\end{equation}
for any $t>0$. Here $\E^* $ denotes the expected value with respect to $\P^*$.
\end{theorem}

\begin{remark}
Note that the claims \ref{ei: main b simple}--\ref{ei: main hcap infty} don't depend on the parameterization.
\end{remark}

The following corollary clarifies the relation between the convergence of random curves and the convergence of their driving processes.
For instance, it shows that if the driving processes of Loewner chains satisfying Condition~\ref{def: b unf crossing} converge,
also the limiting Loewner chain is generated by a curve.
In the statement of the result we assume that $\half$ is endowed with a bounded metric, for instance, the one
inherited from the Riemann sphere. Another possibility is to map $\half$ onto a bounded domain such as $\disc$.

\begin{corollary}\label{cor: driving to curve}
Suppose that $(W^{(n)})_{n \in \N}$ is a sequence of driving processes of random Loewner chains
that are generated by simple random curves $(\gamma^{(n)})_{n \in \N}$ in $\half$, satisfying Condition~\ref{def: b unf crossing}. 
Suppose also that $(\gamma^{(n)})_{n \in \N}$ are parametrized by capacity. Then
\begin{itemize}
\item $(W^{(n)})_{n \in \N}$ is tight in the metrizable space of continuous functions on $[0,\infty)$ with the topology
of uniform convergence on the compact subsets of $[0,\infty)$.
\item $(\gamma^{(n)})_{n \in \N}$ is tight in the space of curves $X$.
\item $(\gamma^{(n)})_{n \in \N}$ is tight in the metrizable space of continuous functions on $[0,\infty)$ with the topology
of uniform convergence on the compact subsets of $[0,\infty)$.
\end{itemize}
Moreover, if the sequence converges in any of the topologies above it also converges in the two other topologies
and the limits agree in the sense that the limiting random curve is
driven by the limiting driving process.
\end{corollary}

The space $C([0,\infty))$ is metrizable, since a metric on it is given, for example, by
\begin{equation*}
\de(f,g) = \sum_{n = 0}^\infty 2^{-n} \min\{1,\sup\{ |f(t) - g(t)| \,:\, t \in [0,2^n] \}\} .
\end{equation*}
It is understood that $a=\gamma^{(n)}(0)$ and $b=\infty$ in the definition of $A^u$.

For the next corollary let's define the space of \emph{open curves} by identifying in the set of continuous maps
$\gamma:(0,1) \to \C$ different parametrizations. The topology will be given by the convergence on the compact subsets
of $(0,1)$. See also Section~\ref{ssec: proof corollaries}.
It is necessary to consider open curves since in rough domains nothing guarantees that there are
curves starting from a given boundary point or prime end.

We say that $(U_n,a_n,b_n)$, $n \in \N$, converges to $(U,a,b)$ in the \emph{Carath\'eodory sense} 
if there exist conformal and onto mappings $\psi_n : \disc \to U_n$ and $\psi: \disc \to U$
such that they satisfy 
$\psi_n(-1)=a_n$, $\psi_n(+1)=b_n$, $\psi(-1)=a$ and $\psi(+1)=b$ (possibly defined as prime ends)
and such that $\psi_n$ converges to $\psi$ uniformly in the compact subsets of $\disc$ as $n \to \infty$.
Note that this limit is not necessarily unique as a sequence $(U_n,a_n,b_n)$ can converge to different limits for different
choices of $\psi_n$. However if know that $(U_n,a_n,b_n)$, $n \in \N$, converges to $(U,a,b)$, then
$\psi(0) \in U_n$ for large enough $n$ and $U_n$ converges to $U$ 
in the usual sense of Carath\'eodory kernel convergence with respect to the point $\psi(0)$. For the definition see 
Section~1.4 of \cite{pommerenke-1992-}.

The next corollary shows that if we have a converging sequence of random curves in the sense of Theorem~\ref{thm: main}
and if they are supported on domains which converge in the Carath\'eodory sense, then the limiting random curve
is supported on the limiting domain. Note that the Carath\'eodory kernel convergence allows 
that there are deep fjords in $U_n$
which are ``cut off'' as $n \to \infty$. One can interpret the following corollary to state that with high probability 
the random curves don't enter any of these fjords. This is a desired property of the convergence.

\begin{corollary}\label{cor: convergence in general domains}
Suppose that the sequence $(U_n,a_n,b_n)$ converges to $(U^*, a^*, b^*)$ in the Carath\'eodory sense  
(here $a^*$, $b^*$ are possibly defined as prime ends)
and suppose that
$(\phi_n)_{n \geq 0}$ are conformal maps such that $U_n = U(\phi_n), a_n = a(\phi_n), b_n = b(\phi_n)$
and $\lim \phi_n = \phi^*$ for which $U^* = U(\phi^*), a^* = a(\phi^*), b = b(\phi^*)$.
Let $\hat{U} = U^* \setminus (V_a \cup V_b)$ where $V_a$ and $V_b$ are some neighborhoods of $a$ and $b$, respectively, 
and set $\hat{U}_n = \phi_n^{-1} \circ \phi (\hat{U})$. If $(\phi_n,\P_n)_{n \geq 0}$ satisfy Condition~\ref{def: b unf crossing}
and $\gamma^{(n)}$ has the law $\P_n$, then $\gamma^{(n)}$ restricted to $\hat{U}_n$ has a weakly converging subsequence in the topology
of $X$, the laws for different $\hat{U}$ are consistent so that it is possible to define a random curve $\gamma$ on the open interval
$(0,1)$ such that the limit for $\gamma^{(n)}$ restricted to $\hat{U}_n$ is $\gamma$ restricted to the closure of $\hat{U}$.
In particular, almost surely the limit of $\gamma^{(n)}$ is supported on open curves of $U^*$ and doesn't enter 
$(\limsup U_n) \setminus \overline{U}^*$.
\end{corollary}

Here we define $\limsup A_n$ for a sequence of sets $A_n \subset \C$ to be the set 
\begin{equation*}
\left\{x \in \C \,:\, \exists \text{ increasing } n_k\in \N \text{ and }  
  x_k \in A_{n_k} \text{ sequences s.t. } \lim_{k \to \infty} x_k = x\right\} 
 = \bigcap_{m=1}^\infty \overline{\bigcup_{n=m}^\infty A_n}  \;.
\end{equation*}

\subsection{The principal application of the main theorem}\label{ssec: main application}

The results of this paper (Corollary~\ref{cor: convergence in general domains} 
and Proposition~\ref{prop: fkIsing}) together with \cite{smirnov-2010-}, \cite{chelkak-smirnov-2009-} and
\cite{chelkak-duminil-copin-hongler-2013-}
are used in \cite{chelkak-duminil-hongler-kemppainen-smirnov-2013-} to establish 
the following (strong) convergence result for the
Ising model interfaces.
For the exact setting consult Section~\ref{ssec: fk model}.

\begin{theorem}[%
Chelkak--Duminil-Copin--Hongler--Kemppainen--Smirnov \cite{chelkak-duminil-hongler-kemppainen-smirnov-2013-}]
\label{thm: cd-chks}
Let $U$ be a bounded simply connected domain with two distinct boundary points $a,b$
(possibly defined as prime ends).
\begin{itemize}
\item (Convergence of spin Ising interfaces) Consider the interface $\gamma_\delta$ 
in the critical spin Ising model with Dobrushin boundary conditions on $(U_\delta,a_\delta,b_\delta)$. 
The law of $\gamma_\delta$ converges weakly, as $\delta \to 0$, 
to the chordal Schramm-Loewner Evolution SLE$(\kappa)$ running from $a$ to $b$ in $U$ with $\kappa=3$.
\item (Convergence of FK Ising interfaces).
Consider the interface $\gamma_\delta$ in the critical FK Ising model with Dobrushin boundary conditions 
on $(U_\delta,a_\delta,b_\delta)$. The law of $\gamma_\delta$ converges weakly, as $\delta \to 0$, to the chordal Schramm-Loewner Evolution SLE$(\kappa)$ running from $a$ to $b$ in $U$ with $\kappa=16/3$.
\end{itemize}
\end{theorem}

The above result is based on a standard approach  for proving convergence. First we
show \emph{precompactness} of the sequence so that it has subsequential limits. 
Then we show that those limits are independent of the subsequence (\emph{uniqueness}).
It follows that the whole sequence converges to this unique limit.
The results of the present article are sufficient to cover the entire precompactness part, but this work
also gives some required tools for the uniqueness part.

The uniqueness part is based on finding an \emph{observable} which has a well-behaved scaling limit.
A typical observable is a solution of a discrete boundary value problem, e.g., the observable
could be a discrete harmonic function with prescribed boundary values and 
defined on the same or related graph as the interface. There needs to be a strong connection
between the observable and the interface so that the observable is a martingale with respect to
the information generated by the growing curve.

Unfortunately, the observables satisfying all the required properties 
have so far been found in only a few cases.

In the article \cite{chelkak-duminil-hongler-kemppainen-smirnov-2013-}
Condition~\ref{def: b unf crossing} is verified for the spin Ising model 
using the results of \cite{chelkak-duminil-copin-hongler-2013-}.
In Section~\ref{ssec: spin Ising} below, we give its alternative derivation 
using only the observable results of \cite{chelkak-smirnov-2009-}, thus giving a new proof
of Theorem~\ref{thm: cd-chks}, independent of \cite{chelkak-duminil-copin-hongler-2013-} and using only
\cite{chelkak-smirnov-2009-} and the ``martingale characterization'' from \cite{chelkak-duminil-hongler-kemppainen-smirnov-2013-}.

Moreover, our proof indicates that in general,
a non-degenerate martingale observable should suffice
to verify Condition~\ref{def: b unf crossing}. Another known example of such an approach
is its verification for the harmonic observable
which we sketch in Section~\ref{ssec: he}.

\subsection{An application to the continuity of SLE}

This section is devoted to an application of Theorem~\ref{thm: main}.

Consider SLE$(\kappa)$, $\kappa \in [0,8)$, for different values of $\kappa$. For an introduction to Schramm--Loewner evolution see
Appendix~\ref{ssec: sle} below and \cite{lawler-2005-}. 
The driving processes of the different SLEs can be given in the same probability space in the obvious
way by using the same standard Brownian motion for all of them. A natural question is to ask whether or not SLE is as a random curve
continuous in the parameter $\kappa$. See also \cite{johansson-rohde-wong-2012-}, where it is proved that SLE is continuous in $\kappa$
for small and large $\kappa$ in the sense of almost sure convergence of the curves when the driving processes are
coupled in the way given above. We will prove the following theorem using Corollary~\ref{cor: driving to curve}.

\begin{theorem}\label{thm: sle continuity}
Let $\gamma^{[\kappa]}(t)$, $t \in [0,\infty)$, be SLE$(\kappa)$ parametrized by capacity.
Suppose that $\kappa \in [0,8)$ and $\kappa_n \to \kappa$ as $n \to \infty$. 
Then as $n \to \infty$, the law of $\gamma^{[\kappa_n]}$ converges weakly to the law of $\gamma^{[\kappa]}$
in the topology of uniform convergence on the compact subsets of $[0,\infty)$. 
\end{theorem}

We'll present the proof here since it is independent of the rest of the paper except that it relies on
Corollary~\ref{cor: driving to curve}, Proposition~\ref{prop: equiv conditions} (equivalence of geometric and conformal conditions)
and Remark~\ref{rem: dmp} (on the domain Markov property). The reader can choose to read these parts before reading this proof.

Notice that SLE$(\kappa)$ is not simple when $\kappa>4$. Therefore we need to slightly extend the setting
of this paper to be able to use it in the proof of Theorem~\ref{thm: sle continuity}.
The assumption that the random curves are simple is used essentially only to guarantee that
they are Loewner chains with continuous driving processes. Also that assumption makes it less
cumbersome to talk about the tip of the curve and whether or not some set separates the tip and the target points
from each other, but this is not a problem in the general case either, since we can always use conformal
mappings and resolve the question in some Jordan domain. As a consequence, 
no extra difficulties arise and we can work with SLE$(\kappa)$
as if they were simple curves.

\begin{proof}
Let $\kappa_0 \in [0,8)$.
First we verify that the family consisting of SLE$(\kappa)$s on $\disc$, say, where $\kappa$ runs over the interval $[0,\kappa_0]$,
satisfies Condition~\ref{def: b unf crossing}.
Since SLE$_\kappa$ has the conformal domain Markov property, it is enough to verify Condition~\ref{cond: c const time zero}. 
More specifically, it is enough
to show that there exists $M>0$ such that if $Q=(V,S_0,S_1,S_2,S_3)$ is a topological quadrilateral with $\elen(Q) \geq M$
such that  $V \subset \half$, $S_k \subset \R_+ \dd= [0,\infty)$
for $k=1,3$ and $S_2$ separates $S_0$ from $\infty$ in $\half$, then
\begin{equation}\label{ie: sle quad crossing}
\P( \text{SLE$(\kappa)$ intersects } S_0 ) \leq \frac{1}{2}
\end{equation}
for any $\kappa \in [0,\kappa_0]$.

Suppose that $M>0$ is large and $Q$ satisfies $\elen(Q) \geq M$.
Let $Q'=(V';S_0',S_2')$ 
be the doubly connected domain where $V'$ is the interior of the closure of $V \cup V^*$, $V^*$ is the mirror image of $V$
with respect to the real axis, and $S_0'$ and $S_2'$ are the inner and outer boundary of $V'$, respectively.
Then the modulus (or extremal length) of $Q'$, which is defined as the extremal length of the curve family connecting $S_0'$ and $S_2'$ in $V'$
(for the definition see Chapter~4 of \cite{ahlfors-1973-}), 
is given by $\elen(Q')=\elen(Q)/2$.

Let $x = \min ( \R \cap S_0' )>0$ and $r = \max \{ |z-x| \,:\, z \in S_0'\} >0$. Then
$Q'$ is a doubly connected domain which separates $x$ and a point on $\{z \,:\, |z-x|=r\}$ from $\{0,\infty\}$.
By Theorem~4.7 of \cite{ahlfors-1973-}, of all the doubly connected domains with this property, the complement of $(-\infty,0] \cup [x,x+r]$
has the largest modulus. By the equation 4.21 of \cite{ahlfors-1973-}, 
\begin{equation}
\exp( 2 \pi \, \elen(Q')) \leq 16 \left( \frac{x}{r} + 1 \right)
\end{equation}
which implies that $r \leq \rho x$ where
\begin{equation}
\rho = \left( \frac{1}{16} \exp( \pi M ) - 1 \right)^{-1} 
\end{equation}
which can be as small as we like by choosing $M$ large.

If SLE$(\kappa)$ crosses $Q$ then it necessarily intersects $\overline{B(x,r)}$. By the scale invariance of SLE$(\kappa)$
\begin{equation}
\P( \text{SLE$(\kappa)$ intersects } S_0 ) \leq \P\left( \text{SLE$(\kappa)$ intersects } \overline{B(1,\rho)} \right).
\end{equation}
Now by standard arguments \cite{rohde-schramm-2005-}, the right hand side can be made less than $1/2$
for $\kappa \in [0,\kappa_0]$ and $0<\rho \leq \rho_0$ where $\rho_0>0$ is suitably chosen constant.

Denote the driving process of $\gamma^{[\kappa]}$ by $W^{[\kappa]}$.
If $\kappa_n \to \kappa \in [0,8)$, then obviously $W^{[\kappa_{n}]}$ converges weakly
to $W^{[\kappa]}$. Hence by Corollary~\ref{cor: driving to curve} also $\gamma^{[\kappa_n]}$ converges weakly
to some $\tilde{\gamma}$ whose driving process is distributed as $W^{[\kappa]}$. That is,
$\gamma^{[\kappa_n]}$ converges weakly to $\gamma^{[\kappa]}$ as $n \to \infty$ provided that $\kappa_n \to \kappa$ as $n \to \infty$.
\end{proof}

\subsection{Structure of this paper}

\begin{figure}[tbh]
\centering
\subfigure[When the radius of the inner circle goes to zero, the dashed line is no longer visible from a faraway reference point.
If such an event has positive probability for the limiting measure, then the Loewner equation doesn't describe the whole curve.] 
{
	\label{sfig: lca}
	\includegraphics[scale=1]
{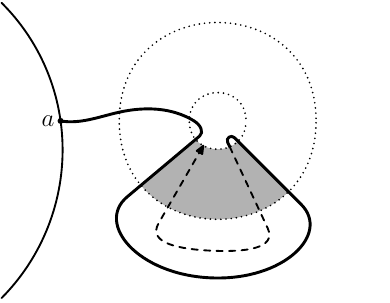}
} 
\hspace{0.4cm}
\subfigure[Longitudinal crossing of an arbitrarily thin tube of fixed length along the curve or the boundary violates the local growth
needed for the continuity of the Loewner driving term.]
{
	\label{sfig: lcb}
	\includegraphics[scale=1]
{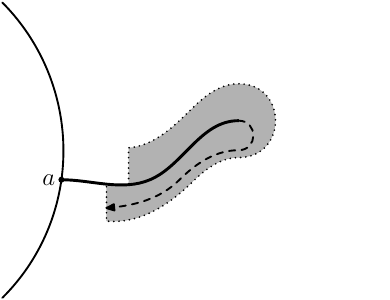}
} 
\caption{In the proof of Theorem~\ref{thm: main}, the regularity of random curves is established by establishing a probability upper bound
on multiple crossings and excluding two unwanted scenarios presented in this figure.} \label{fig: loewner curve}
\end{figure}

In Section~\ref{sec: cond}, the general setup of this paper is presented. 
Four conditions are stated and shown to be equivalent. Any one of them
can be taken as the main assumption for Theorem~\ref{thm: main}.

The proof of Theorem~\ref{thm: main} is presented in Section~\ref{sec: proof main}.
The proof consists of three parts: the first one is the existence of regular parametrizations of the random curves and
the second and third steps are described in Figure~\ref{fig: loewner curve}. 
The relevant condition is verified for a list of random curves arising
from statistical mechanics models
in Section~\ref{sec: checking condition}.


\section{The space of curves and equivalence of conditions} \label{sec: cond}

%

\subsection{The space of curves and conditions}

\subsubsection{The space of curves} \label{ssec: space of curves} \label{ssec: curves}

We follow the setup of Aizenman and Burchard's paper \cite{aizenman-burchard-1999-}:
planar curves are continuous mappings from $[0,1]$ to $\C$ modulo reparameterizations. Let
\begin{equation*}
C' = \left\{ f \in C \left( [0,1],\C \right) \,:\, 
  \begin{gathered}
  \text{either $f$ is not constant on any subinterval of $[0,1]$} \\
  \text{or $f$ is constant on $[0,1]$} 
  \end{gathered}
  \right\}.
\end{equation*}
It is also possible to work with the whole space $C \left( [0,1],\C \right)$, but the next definition is easier for $C'$.
Define an equivalence relation $\sim$ in $C'$ so that $f_1 \sim f_2$
if they are related by an increasing homeomorphism $\psi:[0,1] \to [0,1]$ with $f_2 = f_1 \circ \psi$. 
The reader can check that this defines an equivalence relation.
The mapping $f_1 \circ \psi$ is said to be a \emph{reparameterization}  of $f_1$ or that $f_1$ is \emph{reparameterized}
by $\psi$. 

Note that these parameterizations are, in a sense, arbitrary and are in general different from the
Loewner parameterization which we are going to construct.

Denote the equivalence class of $f$ by $[f]$. The set of all equivalence classes
\begin{equation*}
X = \left\{ [f] : f \in C' \right\}
\end{equation*}
is called the \emph{space of curves}. Make $X$ a metric space by setting
\begin{equation}\label{eq: curve distance metric}
\de_X ([f],[g]) = \inf \{ \|f_0- g_0\|_\infty : f_0 \in [f], g_0 \in [g] \} .
\end{equation}
It is easy to see that this is a metric, see e.g. \cite{aizenman-burchard-1999-}.
The space $X$ with the metric $\de_X$ is complete and separable
reflecting the same properties of $C \left( [0,1],\C \right)$.
And for the same reason as $C \left( [0,1],\C \right)$ is not compact neither is $X$.

Define two subspaces, the space $X_{\mathrm{simple}}$ of \emph{simple curves} and 
the space $X_0$ of curves with no self-crossings by
\begin{align*}
X_{\mathrm{simple}} &= \left\{ [f] : f \in C', f \textrm{ injective} \right\} \\
X_0 &= \overline{ X_{\mathrm{simple}} }
\end{align*}
Note that $X_0 \subsetneq X$ since there exists $\gamma_0 \in X \setminus X_{\mathrm{simple}}$ with positive
distance to $X_{\mathrm{simple}}$. For example, such is the broken line passing through points $-1,1,i$ and $-i$
which has a double point which is stable under small perturbations.

What do the curves in $X_0$ look like? Roughly speaking,
they may touch themselves and have multiple points, but they can have no ``transversal'' self-intersections.
For example, the broken line through points $-1,1,i,0,-1+i$, also has a double point at $0$,
but it can be removed by small perturbations. Also, every passage through the double point separates its neighborhood
into two components, and every other passage is contained in (the closure) of one of those.
See also Figure~\ref{fig: non-self-crossing}.

\begin{figure}[tbh]
\centering
	\includegraphics[scale=1]
{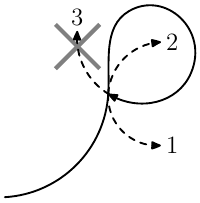}
\caption{In this example the options 1 and 2 are possible so that the resulting curve in the class $X_0$. If the curve continues
  along $3$ it doesn't lie in $X_0$, namely, there is no sequence of simple curves converging to that curve.
    }%
    \label{fig: non-self-crossing}
\end{figure}

Given a domain $U \subset \C$ define $X(U)$ as the closure of 
$\{ [f] : f \in C', f[0,1] \subset U \}$ 
in $(X,\de_X)$.
Define also $X_0(U)$ as the closure of the set of simple curves in $X(U)$. 
The notation $X_{\mathrm{simple}}(U)$ we
reserve for 
\begin{equation*}
X_{\mathrm{simple}}(U) = 
  \left\{ [f] : f \in C', f\big((0,1)\big) \subset U, f \textrm{ injective} \right\},
\end{equation*}
so the end points of such curves may lie on the boundary.
Note that the closure of $\xs(U)$ is still $X_0(U)$.

Use also notation $\xs(U,a,b)$ for curves in $\xs(U)$ whose end points are $\gamma(0)=a$ and $\gamma(1)=b$.
We will quite often consider some reference sets as
$\xs(\disc,-1,+1)$ and $\xs(\half,0,\infty)$ where the latter can be understood
by extending the above definition to curves defined on the Riemann sphere, say.

We will often use the letter $\gamma$ to denote elements of $X$, i.e. a curve modulo reparameterizations. 
Note that topological properties of the curve (such as its endpoints or passages through annuli
or its \emph{locus} $\gamma[0,1]$)
as well as metric ones (such as dimension or length) are independent of parameterization.
When we want to put emphasis on the locus, we will be speaking about \emph{Jordan} curves or arcs,
usually parameterized by the open unit interval $(0,1)$.

Denote by $\Prob(X)$ the space of probability
measures on $X$ equipped with the Borel $\sigma$-algebra $\B_X$
and the weak-$*$ topology induced by continuous functions
(which we will call weak for simplicity).
Suppose that $\P_n$ is a sequence of measures in $\Prob(X)$.

If for each $n$, $\P_n$ is supported on a closed subset of $\xs$
(which for discrete curves can be assumed without loss of generality)
and if $\P_n$ converges weakly to a probability measure $\P$,
then $1 = \limsup_n \P_n (X_0) \leq \P(X_0)$ 
by general properties of the weak convergence of probability measures \cite{billingsley-1999-}.
Therefore $\P$ is supported on $X_0$ but in general it doesn't have to be supported on $\xs$.

\subsubsection{Comment on the probability structure}

Suppose $\P$ is supported on $D \subset X(\C)$ which is a closed subset of $\xs(\C)$. 
Consider some measurable map $\chi : D \to C([0,\infty),\C)$ so that $\chi(\gamma)$ is a parametrization of $\gamma$.
If necessary $\chi$ can be continued to $D^c$ by setting $\chi =0$ there. 

Let $\pi_t$ be the natural projection from $C([0,\infty),\C)$ to $C([0,t],\C)$.
Define a $\sigma$-algebra
\begin{equation*}
\F_t^{\chi,0} = \sigma ( \pi_s \circ \chi, 0 \leq s \leq  t )~,
\end{equation*}
and make it right continuous by setting $\F_t^\chi = \bigcap_{s>t} \F_s^{\chi,0}$.

For a moment denote by $(\tau,\hat{\tau})$ for given $\gamma,\hat{\gamma} \in D$ the maximal pair of times
such that $\chi(\gamma)|_{[0,\tau]}$ is equal to $\chi(\hat{\gamma})|_{[0,\hat{\tau}]}$ in $X$, that is, equal modulo a reparameterization.
We call $\chi$ a \emph{good parametrization} of the curve family $D$, if for each $\gamma,\hat{\gamma} \in D$, 
$\tau=\hat{\tau}$ and $\chi(\gamma,t)=\chi(\hat{\gamma},t)$ for all $0 \leq t \leq \tau$.

Each reparameterization from a good parametrization to another can be represented as stopping times $T_u$, $u \geq 0$.
From this it follows that the set of stopping times is the same for every good parametrization. We will use simply the
notation $\gamma[0,t]$ to denote the $\sigma$-algebra $\F_t^\chi$. The choice of a good parametrization $\chi$ is 
immaterial since all the events we will consider are essentially reparameterization invariant. But to ease the notation
it is useful to always have some parametrization in mind.

Often there is a natural choice for the parametrization. For example, if we are considering paths on a lattice, then
the probability measure is supported on polygonal curves. In particular, the curves are piecewise smooth and it is possible
to use the arc length parametrization, i.e. $|\gamma'(t)|=1$. One of the results in this article is that given the hypothesis,
which is described next, it is possible to use the capacity parametrization of the Loewner equation.
Both the arc length and the capacity are good parameterizations.

The following lemma is implied by the above definitions.

\begin{lemma}
If $A \subset \C$ is a non-empty, closed set,
then $\tau_A = \inf\{ t \geq 0 \,:\, \chi(\gamma,t) \in A \}$ is a stopping time.
\end{lemma}

\begin{remark}
The stopping times we need in the proof of the main theorem are always explicitly of this type.
\end{remark}

\subsubsection{Four equivalent conditions}\label{ssec: four conditions}

Recall the general setup: we are given a collection $(\phi,\P) \in \Sigma$ where the conformal map $\phi$
contains also the information about the domain $(U,a,b)=(U(\phi),a(\phi),b(\phi))$ and $\P$ is a probability
measure on $\xs(U,a,b)$. Furthermore, we assume that each $\gamma$,
which is distributed according to $\P$,
has some suitable parametrization.

For given domain $U$ and for given simple (random) curve $\gamma$ on $U$, we always define
$U_\tau = U \setminus \gamma[0,\tau]$ for each (random) time $\tau$. We call $U_\tau$ as the domain at time
$\tau$.

\begin{definition}
For a fixed domain $(U,a,b)$ and for fixed simple (random) curve in $U$ starting from $a$,  
define for any annulus $A = A(z_0,r,R)$ and for any (random) time $\tau \in [0,1]$, 
$A^u_\tau = \emptyset$ if $\partial B(z_0,r) \cap \partial U_\tau = \emptyset$ and
\begin{equation}
A^u_\tau = \left\{ z \in U_\tau \cap A \,:\, 
   \begin{gathered}
   \text{the connected component of $z$ in $U_\tau \cap A$} \\
   \text{doesn't disconnect $\gamma(\tau)$ from $b$ in $U_\tau$}
   \end{gathered}
   \right\}
\end{equation}
otherwise. A connected set $C$ disconnects $\gamma(\tau)$ from $b$ if it disconnects some neighborhood of $\gamma(\tau)$
from some neighborhood of $b$ in $U_\tau$. If $\gamma[\tau,1]$ contains a crossing of $A$ which is contained in $A_\tau^u$,
we say that $\gamma$ \emph{makes an unforced crossing} of $A$ in $U_\tau$ (or an unforced crossing of $A$ observed at time $\tau$). 
The set $A_\tau^u$ is said to be \emph{avoidable} at time $\tau$.
\end{definition}

\begin{remark}
Neighborhoods are needed here only to incorporate the fact that $\gamma(t)$ and $b$ are boundary points.
\end{remark}

The first two of the four
equivalent conditions are geometric,
asking an unforced crossing of an annulus to be unlikely
uniformly in terms of the modulus.

\begin{condition} \label{def: b unf crossing}\label{cond: annulus}
The family $\Sigma$
is said to satisfy a \emph{geometric bound on an unforced crossing}
if there exists $C >1$ such that
for any $(\phi,\P) \in \Sigma$,
for any stopping time $0 \leq \tau \leq 1$ and for any annulus $A=A(z_0,r,R)$ where $0 < C \, r \leq R$, 
\begin{equation}
\P \big( \big.
      \gamma[\tau,1] \textrm{ makes a crossing of } A
      \text{ which is contained in } \overline{A^u_\tau} 
   \,\big|\, \gamma[0,\tau] \big)  < \frac{1}{2}  . 
\end{equation}
\end{condition}

\begin{condition} \label{def: pb unf crossing}\label{cond: annulus exp}\label{cond: geom power-law} 
The family $\Sigma$ 
is said to satisfy a \emph{geometric power-law bound on an unforced crossing}
if there exist $K >0$ and $\Delta>0$ such that 
for any $(\phi,\P) \in \Sigma$,
for any stopping time $0 \leq \tau \leq 1$ and for any annulus $A=A(z_0,r,R)$ where $0 < r \leq R$,
\begin{equation}
\P \big( \big.
      \gamma[\tau,1] \textrm{ makes a crossing of } A
      \text{ which is contained in } \overline{A^u_\tau}
   \,\big|\, \gamma[0,\tau] \big) \leq K \left( \frac{r}{R} \right)^\Delta  . 
\end{equation}
\end{condition}

Let $Q\subset U_t$ be a topological quadrilateral,
i.e. an image of the square $(0,1)^2$ under a homeomorphism $\psi$.
Define the ``sides'' $\partial_0Q$, $\partial_1Q$, $\partial_2Q$, $\partial_3Q$,
as the ``images'' of 
\begin{equation*}
\brs{0}\times(0,1), \quad (0,1)\times\brs{0}, \quad \brs{1}\times(0,1), \quad (0,1)\times\brs{1} 
\end{equation*}
under $\psi$.
For example, we set
$$\partial_0Q:=\lim_{\epsilon\to0}\mathrm{Clos}\br{\psi\br{\; (0,\epsilon)\times(0,1) \;}}.$$
We consider $Q$ such that two opposite sides
$\partial_1Q$ and $\partial_3Q$  are contained in $\partial U_t$. 
A \emph{crossing} of $Q$ is a curve in $U_t$ connecting two opposite  sides $\partial_0Q$ and $\partial_2Q$.
The latter without loss of generality (just perturb slightly)
we assume to be smooth curves of finite length inside $U_t$.
Call $Q$ \emph{avoidable} if it doesn't disconnect $\gamma(t)$ and $b$ inside $U_t$.

\begin{conditionc}
\label{cond:quad}\label{cond: conf}
The family $\Sigma$
is said to satisfy a 
\emph{conformal bound on an unforced crossing}
if there exists a constant $M >0$ such that 
for any $(\phi,\P) \in \Sigma$, for any stopping time $0 \leq \tau \leq 1$ and 
any avoidable quadrilateral $Q$ of $U_\tau$, such that the modulus $m(Q)$ is larger than $M$
\begin{equation}
\P \big( \big.
      \gamma[\tau,1] \textrm{ crosses } Q
   \,\big|\, \gamma[0,\tau] \big) \leq \frac{1}{2} .
\end{equation}
\end{conditionc}

\begin{remark}
In the condition above,
the quadrilateral $Q$ depends on $\gamma[0,\tau]$,
but this does not matter, as we consider \emph{all} such quadrilaterals.
A possible dependence on $\gamma[0,\tau]$ ambiguity can be addressed by mapping $U_t$ to a reference domain and choosing quadrilaterals there.
See also Remark~\ref{rem all any}.
\end{remark}

\begin{conditionc} \label{cond:quadexp}\label{cond: conf power-law}
The family $\Sigma$
is said to satisfy a \emph{conformal power-law bound on an unforced crossing}
if there exist constants $K$ and $\epsilon$ such that 
for any $(\phi,\P) \in \Sigma$, for any stopping time $0 \leq \tau \leq 1$ and 
any avoidable quadrilateral $Q$ of $U_\tau$ 
\begin{equation}
\P \big( \big.
      \gamma[\tau,1] \textrm{ crosses } Q
   \,\big|\, \gamma[0,\tau] \big) \leq K\,\exp(-\epsilon \, m(Q)) .
\end{equation}
\end{conditionc}

\begin{proposition} \label{prop: equiv conditions}
The four conditions \ref{cond: annulus}, \ref{cond: annulus exp}, \ref{cond:quad} and \ref{cond:quadexp}
 are equivalent and conformally invariant.
\end{proposition}
This proposition is proved below in Section~\ref{ssec: equiv cond}.
Equivalence of conditions immediately implies the following

\begin{corollary}\label{cor: equiv conditions onehalf}
The constant $1/2$ in Conditions~\ref{def: b unf crossing} and \ref{cond:quad} can be replaced by any other from $(0,1)$.
\end{corollary}

\subsubsection{Remarks concerning the conditions}

\begin{remark}
Conditions~\ref{cond: annulus} and \ref{cond: annulus exp} could be described as being \emph{geometric}
since they involve crossing of fixed shape.
Conditions~\ref{cond:quad} and \ref{cond:quadexp} are \emph{conformally invariant} because they are formulated
using the modulus, i.e., the extremal length which is a conformally invariant quantity.
The conformal invariance in Proposition~\ref{prop: equiv conditions} means for example, that if 
Condition~\ref{cond: annulus} holds with a constant $C>1$ for $(\phi,\P)$ defined in $U$ 
and if $\psi:U \to U'$ is conformal and onto, then Condition~\ref{cond: annulus}
holds for $(\phi \circ \psi^{-1}, \psi \P)$ with a constant $C'>1$ which depends only on the constant $C$
but not on $(\phi,\P)$ or $\psi$.
\end{remark}

\begin{remark}\label{rem: dmp}
To formulate the \emph{domain Markov property} with an appropriate set of stopping times,
let's suppose that 
$\Sigma$ is a collection of pairs $(\phi_n^{U,a,b},\P_n^{U,a,b})$
where $n \in \N$ refers to the lattice mesh $\delta_n$ which tends to zero as $n$ tends to infinity,
$U$ is a simply connected domain whose boundary is a discrete curve (broken line) on the lattice
with mesh $\delta_n$ and $a$ and $b$ are lattice points on the boundary of the domain and
as usual $\phi_n^{U,a,b}$ is a conformal map taking $U,a,b$ onto $\disc,-1,1$.
If for any stopping time $\tau$, such that $\gamma(\tau)$ is almost surely a lattice point,
it  holds that
\begin{equation*}
\P^{U,a,b}_n \left(\left. \gamma|_{[\tau,1]} \in \cdot \,\right|\, \gamma|_{[0,\tau]} \right) 
= \P^{U \setminus \gamma[0,\tau],\gamma(\tau),b}_n ,
\end{equation*}
then the random curve or $\Sigma$ is said to have the domain Markov property.
This property could be formulated
more generally so that if $\P$ is a probability measure such that $(\phi,\P) \in \Sigma$ for some $\phi$,
then for any stopping time $\tau$,
$\P \left(\left. \gamma|_{[\tau,1]} \in \cdot \,\right|\, \gamma|_{[0,\tau]} \right)$
is equal to some probability measure $\P'$ such that $(\phi',\P') \in \Sigma$
for some $\phi'$.

When the domain Markov property holds,
the ``time zero conditions'' \ref{cond: g const time zero} and \ref{cond: c const time zero} 
are sufficient for Conditions~\ref{cond: annulus} and \ref{cond:quad}, respectively.
\end{remark}

\begin{remark}\label{rem all any}
Our conditions impose an estimate on conditional probability, which is hence satisfied almost surely.
By taking a countable dense set  of round annuli (or of topological rectangles),
we see that it does not matter whether we require the estimate to hold separately for \emph{any} given annulus almost surely;
or to hold almost surely for \emph{every} annulus. The same argument applies to topological rectangles.
\end{remark}

\begin{remark}
Suppose now that the random curve $\gamma$ is an interface in a statistical physics model with two possible states
at each site, say, \emph{blue} and \emph{red}. In that case $U$ will be a simply connected domain formed by entire faces
of some lattice, say, hexagonal lattice, $a,b \in \partial U$ are boundary points, the faces next to
the arc $ab$ are colored blue and next to the arc $ba$ red and $\gamma$ is the interface between the
blue cluster of $ab$ (connected set of blue faces) and the red cluster of $ba$.

In this case under \emph{positive association} (e.g. observing blue faces somewhere increases the probability
of observing blue sites elsewhere) the sufficient condition implying Condition~\ref{def: b unf crossing}
is uniform upper bound for the probability of the \emph{crossing event of  
an annular sector} with alternating boundary conditions (red--blue--red--blue) on the four boundary arcs 
(circular--radial--circular--radial) by blue faces.
For more detail, see Section~\ref{sssec: cond for fk ising}.
\end{remark}

\subsection{Equivalence of the geometric and conformal conditions} \label{ssec: equiv cond}

In this section we prove Proposition~\ref{prop: equiv conditions}
about equivalence of geometric and conformal conditions.
We start with recalling the notion of Beurling's extremal length
and then proceed to the proof.
Note that since Condition \ref{cond:quad} is conformally invariant,
conformal invariance of other conditions immediately follows.

Suppose that a curve family $\Gamma \subset X$ consists of curves that are regular enough for the purposes below.
A non-negative Borel function $\rho$ on $\C$ is called \emph{admissible} if
\begin{equation} \label{ie: admissible rho}
\int_\gamma \rho \; \de \len \; \geq 1
\end{equation}
for each $\gamma \in \Gamma$. Here $\de \len$ is the arc-length measure.

The \emph{extremal length} of a curve family $\Gamma \subset X$ is defined as
\begin{equation} \label{eq: conf module}
\elen(\Gamma) = \frac{1}{\inf_{\rho}\int \rho^2 \; \de A}
\end{equation}
where the infimum is taken over all the admissible functions $\rho$.
Here $\de A$ is the area measure (Lebesgue measure on $\C$).
The quantity inside the infimum is called the \emph{$\rho$-area} and 
the quantity on the left-hand side of the inequality~\eqref{ie: admissible rho} is called the \emph{$\rho$-length} of $\gamma$.

The extremal length is conformally invariant. The \emph{modulus} $\elen(Q)$ of a topological quadrilateral $Q=(V,S_0,S_1,S_2,S_3)$
can be defined as the extremal length of the curve family connecting the sides $S_0$ and $S_2$ within $V$. By
conformal invariance this definition of the modulus agrees with the one given in the introduction, for instance, in
Figure~\ref{sfig: crossings e}.
Similarly, the modulus of an annulus, which was also given above, is equal to the extremal length of the curve family connecting
the two boundary circles of the annulus.

The following basic estimate is easy to obtain.

\begin{lemma}\label{lm: elen annulus}
Let $A=A(z_0,r_1,r_2)$, $0 < r_1 < r_2$, be an annulus. Suppose that $\Gamma$ is a curve family with the property
that each curve $\gamma \in \Gamma$ contains a crossing of $A$. Then,
\begin{equation}
\elen (\Gamma) \geq \frac{1}{2\pi} \log \left( \frac{r_2}{r_1} \right)
\end{equation}
and therefore
\begin{equation}
r_1 \geq r_2 \cdot \exp \left( - 2\pi\elen (\Gamma) \right) .
\end{equation}
\end{lemma}

\begin{proof}
Let $\widehat{\Gamma}$ be the family of curves connecting the two boundary circles of $A$. 
If $\rho$ is admissible for $\widehat{\Gamma}$ then it is also admissible for $\Gamma$. Hence,
$\elen(\Gamma) \geq \elen(\widehat{\Gamma}) = (2\pi)^{-1} \, \log( r_2 / r_1 )$.
\end{proof}

Next we present an integral estimate for the extremal length which will be
essential in the proof below. The first formulation of this lemma is classical
and the second form is the one that we use.

\begin{lemma}[Integral estimates of the extremal length]\label{lem: elen integral}
Let $a<b$, let $\Omega$ be a domain and let $C_a$ and $C_b$ be two subsets of $\overline{\Omega}$.
Let $\Gamma$ be the curve family connecting $C_a$ to $C_b$ inside $\Omega$.
For each $x \in (a,b)$ let $I_x$ be a set separating $C_a$ and $C_b$ in $\Omega$.

\begin{itemize}
\item
Suppose that $C_a \subset \overline{\Omega} \cap \{z\in \C \,:\, \real z < a\}$,
$C_b \subset \overline{\Omega} \cap \{z\in \C \,:\, \real z > b\} $ and
$I_x \subset \Omega \cap \{z\in \C \,:\, \real z = x\}$  for each $x$.
Suppose also that the mapping $x \mapsto \Lambda(I_x)$
is measurable where $\Lambda$ is the length measure.
The extremal length $\elen(\Gamma)$ satisfies
\begin{equation*}
\elen(\Gamma) \geq \int_a^b \frac{\de x}{\Lambda(I_x)} . 
\end{equation*}
\item Let $z_0 \in \C$ and
suppose that $C_a \subset \overline{\Omega} \cap \{z\in \C \,:\, |z-z_0| < e^a\}$,
$C_b \subset \overline{\Omega} \cap \{z\in \C \,:\, |z-z_0| > e^b\} $ and
$I_x \subset \Omega \cap \{z\in \C \,:\, |z-z_0| = e^x\}$  for each $x$.
Suppose also that the mapping $x \mapsto \theta(I_x)$
is measurable where $\theta$ is the arc length measure defined in radians for any subset of a circle
of the form $\partial B(z_0,e^x)$.
The extremal length $\elen(\Gamma)$ satisfies
\begin{equation*}
\elen(\Gamma) \geq \int_a^b \frac{\de x}{\theta(I_x)} . 
\end{equation*}
\end{itemize}
\end{lemma}

\begin{proof}
Let $l=\int_a^b \frac{\de x}{\Lambda(I_x)}$.
The first claim follows if we choose the particular function
$\rho(z) = l^{-1} \,\ind_{a<\real z <b} \,\Lambda(I_{\real z})^{-1}$ to give an upper bound
for the infimum in \eqref{eq: conf module}.
The second claim follows then by conformal invariance of the extremal length.
\end{proof}

We now proceed to showing the equivalence of four conditions by establishing the following implications:

\paragraph{\ref{def: b unf crossing}$\Leftrightarrow$\ref{def: pb unf crossing}} 
Condition~\ref{def: b unf crossing} directly follows from \ref{def: pb unf crossing}
by setting $C:=(2K)^{1/\Delta}$.

In the opposite direction, an unforced crossing of the annulus $A(z_0,r,R)$
implies consecutive unforced crossings of the concentric annuli $A_j:=A(z_0,C^{j-1}r,C^{j}r)$,
with $j\in\{1,\dots,n\}$, $n \dd=\lfloor \log (R/r)/\log C \rfloor$,
which have conditional (on the past) probabilities of at most $1/2$ by Condition~\ref{def: b unf crossing}. 
Trace the curve $\gamma$ denoting by $\tau_j$ the ends of unforced crossings of $A_{j-1}$'s
(with $\tau_1=\tau$), and estimating
\begin{align*}
\P &\big( \big.
      \gamma[\tau,1] \textrm{ crosses } A^u_\tau
   \,\big|\, \gamma[0,\tau] \big) \leq 
\prod_{j=1}^n~\P \big( \big.
      \gamma[\tau_j,1] \textrm{ crosses } (A_j)^u_{\tau_j}
   \,\big|\, \gamma[0,\tau_j] \big)\\
&\leq \left( \frac{1}{2} \right)^n
\leq \left( \frac{1}{2} \right)^{(\log (R/r)/\log C)-1}
= 2 \left( \frac{r}{R} \right)^{\log 2/\log C}.
\end{align*}
We infer condition~\ref{def: pb unf crossing} with $K:=2$ and $\Delta:=\log 2/\log C$.

\paragraph{\ref{cond:quad}$\Leftrightarrow$\ref{cond:quadexp}} 
This equivalence is proved similarly to the equivalence of the geometric conditions.
The only difference is that instead of cutting an annulus into concentric ones of moduli $C$,
we start with an avoidable quadrilateral $Q$, and cut from it $n=[m(Q)/M]$
quadrilaterals $Q_1,\dots,Q_n$ of modulus $M$.
If $Q$ is mapped by a conformal map $\phi$
onto the rectangle $\{z:\,0<\real z<m(Q),\,0<\imag z<1\}$,
we can set $Q_j:=\phi^{-1}\{z:\,(j-1)M<\real z<jM,\,0<\imag z<1\}$.
Then as we trace $\gamma$, all $Q_j$'s are avoidable for its consecutive pieces.

\paragraph{\ref{def: b unf crossing}$\Rightarrow$\ref{cond:quad}}
We show that Condition \ref{def: b unf crossing} with constant $C$ 
implies Condition \ref{cond:quad} with $M = 4 (C+1)^2$.

Let $m\ge M$ be the modulus of $Q$,
i.e. the extremal length $\elen(\Gamma)$ of the family $\Gamma$ of curves
joining $\partial_0Q$ to $\partial_2Q$ inside $Q$. 
Let $\Gamma^*$ be the dual family of curves
joining $\partial_1Q$ to $\partial_3Q$ inside $Q$, then
$\elen(\Gamma)=1/\elen(\Gamma^*)$. 

Denote by $d_1$ the distance  between
$\partial_1Q$ and $\partial_3Q$
in the inner Euclidean metric of $Q$, and let $\gamma^*$ be
a curve of length $\le2d_1$ joining $\partial_1Q$ to $\partial_3Q$
inside $Q$.
Observe that any crossing $\gamma$  
of $Q$ 
contains a subcurve which an element of $\Gamma$ and therefore it
has diameter $d\ge2C d_1$.
Indeed, working with the extremal length of the family $\Gamma^*$,
take a metric $\rho$ equal to $1$ in the $d_1$-neighborhood of $\gamma$.
Then its area integral $\iint \rho^2$ is at most $(d+2 d_1)^2$.
But every curve from $\Gamma^*$ intersects $\gamma$ and runs through this
neighborhood for the length of at least $d_1$, thus having $\rho$-length
at least $d_1$.
Therefore $1/m=\elen(\Gamma^*)\ge(d_1)^2 / (d + 2 d_1)^2$,
so we conclude that $m \le (2 + d/d_1)^2$
and hence
\begin{equation}
d\ge(\sqrt{m}-2)d_1\ge\br{2(C+1)-2}d_1=2 Cd_1.
\label{eq:diam}
\end{equation}
 
Now take an annulus $A$ centered at the middle point of $\gamma^*$ with inner radius
$d_1$ and outer radius $R:=C d_1$.
It is sufficient to prove that every crossing of $Q$ contains an unforced crossing of $A$.

Assume on the contrary that $\gamma$ is a curve crossing $Q$ but not $A$.
Clearly $\gamma$ has to intersect $\gamma^*$, say at $w$.
But $\gamma^*$ is entirely contained inside the inner circle of $A$. 
On the other hand by \eqref{eq:diam}
the diameter of $\gamma$ is bigger than $2R$.
Thus $\gamma$ intersects both boundary circles of $A$, and
we deduce  Condition \ref{cond:quad}.

\paragraph{\ref{cond:quadexp}$\Rightarrow$\ref{def: b unf crossing}}
Now we will show that Condition  \ref{cond:quadexp}
with constants $K$ and $\epsilon$
(equivalent to Condition \ref{cond:quad}) implies
Condition \ref{def: b unf crossing} with constant $C=\br{2Ke^2}^{2\pi/\epsilon}$.

We have to show that probability of an unforced crossing of a fixed 
annulus $A=A(z_0,r,Cr)$ is at most $1/2$.
Without loss of generality assume that we work with the crossings from 
the inner circle to the outer one.

For $x\in [0,\log C]$ denote by ${\cal I}^x$
the (at most countable)
set of arcs $I^x$ which  compose $\Omega\cap\partial B(z_0,re^x)$.
By $|I|$ we will denote the length of the arc $I$ measured in radians
(regardless of the circle radius).
Given two arcs $I^x$ and $I^y$ with $y<x$, we will write
$I^y\prec I^x$ if any curve $\gamma$ intersecting $I^x$
has to intersect $I^y$ first, and can do 
so  without intersecting any other arc from ${\cal I}^y$ afterwards.
We denote by $I^y(I^x)$ the unique arc $I^y\in{\cal I}^y$
such that $I^y\prec I^x$.

\begin{figure}[tbh]
\centering
	\includegraphics[scale=.9]
{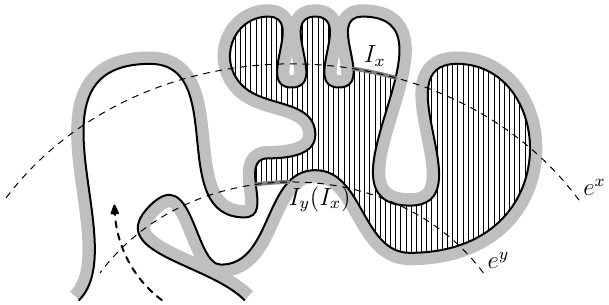}
\caption{This figure shows an example of the last arc $I^y(I^x)$ that a path to $I^x$ has to intersect
($I^y(I^x)$ and $I^x$ are the gray lines)
and the corresponding topological quadrilateral (the region shaded with vertical lines). 
}
\end{figure}

By $Q(I^x)$ we denote the topological quadrilateral
which is cut from $\Omega$ by the arcs $I^x$ and $I^0(I^x)$.
Denote
$$
\ell\br{I^x}=\ell_0^x\br{I^x}:=\int_0^x\frac1{\bra{I^y(I^x)}}dy.
$$
By the second integral estimate of Lemma~\ref{lem: elen integral},
\begin{equation}\label{eq:beurling}
m(Q(I^x))\ge \ell\br{I^x}.
\end{equation}
Note that if $\gamma$ crosses $A$ and intersects $I^x$, then
it makes an unforced crossing of $Q(I^x)$, so we conclude that
by Condition~\ref{cond:quadexp} the probability of
crossing $A$ and intersecting $I^x$ is majorated by
\begin{equation}\label{eq:quadbound}
 K\,\exp\br{-\epsilon \ell\br{I^x}}.
\end{equation}
Denote also $\bra{{\cal I}^x}:=\sum\bra{I^x}$
and $\ell\br{{\cal I}^x}:=\int_0^x\frac1{\bra{{\cal I}^y}}dy$

We call a collection of arcs $\brs{I_j}$
(possibly corresponding to different $x$'s) \emph{separating}, 
if every unforced crossing $\gamma$ intersects one of those.
To deduce Condition~\ref{def: b unf crossing}, by \eqref{eq:quadbound} 
it is enough to find a separating collection
of arcs such that
\begin{equation}\label{eq:2K}
\sum_j \exp\br{-\epsilon\,\ell\br{I_j}} < \frac1{2K}.
\end{equation}
Note that for every $x$ the total length $\bra{{\cal I}^x}\le2\pi$,
and so by our choice of constant $C$ we have
$$\ell\br{{\cal I}^{\log C}}\ge\frac{\log C}{2\pi}\ge\frac2\epsilon,$$ 
as well as
$$\exp\br{2-\epsilon\ell\br{{\cal I}^{\log C}}}\le\exp\br{2-\epsilon\frac{\log C}{2\pi}}
\le\exp\br{2-\log\br{2Ke^2}}=\frac1{2K}.$$ 

Therefore it is 
enough to establish 
that for any $w \in [0,\log C]$ with $\ell\br{{\cal I}^{w}}\ge\frac2\epsilon$
there exist arcs
$I_j$ separating ${\cal I}^w$ with the following estimate:
\begin{equation}\label{eq:abstract}
\sum_j \exp\br{-\epsilon\,\ell\br{I_j}} \le \exp\br{2-\epsilon\ell\br{{\cal I}^w}}.
\end{equation}
We will do this in an abstract setting 
for families of arcs.
Besides properties mentioned above, we note that
for any two arcs $I$ and $J$ the arcs
$I^x(I)$ and $I^x(J)$ either coincide or are disjoint.
Also without loss of generality any arc $I$ we consider
satisfies $I\prec J$ for some $J\in{\cal I}^w$.

By a limiting argument it is enough to prove \eqref{eq:abstract}
for ${\cal I}^w$ of finite cardinality $n$, and we will do this by induction in $n$.

If $n=1$, then we take the only arc $J$ in ${\cal I}^{w}$ as the separating one,
and the estimate \eqref{eq:abstract} readily follows:
$$\exp\br{-\epsilon\,\ell\br{J}}
 = \exp\br{-\epsilon\,\ell\br{{\cal I}^w}}
 < \exp\br{2-\epsilon\,\ell\br{{\cal I}^w}}.$$

Suppose $n>1$.
Denote by $v$ the minimal number such that ${\cal I}^v$ contains more than one arc.

If
$$\ell_v^w\br{{\cal I}^w}:=\int_v^w\frac1{\bra{{\cal I}^y}}dy < \frac2\epsilon,$$
then we take the only arc $J$ in ${\cal I}^{v-\delta}$ as the separating one.
The required estimate \eqref{eq:abstract} then holds if $\delta$ is small enough:
\begin{align*}
\exp\br{-\epsilon\,\ell\br{J}}&
= \exp\br{-\epsilon\,\ell_0^w\br{{\cal I}}+\epsilon\,\ell_{v-\delta}^w\br{{\cal I}}}  \\
&\le \exp\br{-\epsilon\,\ell\br{{\cal I}_0^w}+\epsilon\frac2\epsilon}
= \exp\br{-\epsilon\,\ell\br{{\cal I}}+2}.
\end{align*}

\begin{figure}[tbh]
\centering
\subfigure[] 
{
	\label{sfig: el integral estimate a}
	\includegraphics[scale=.65]
{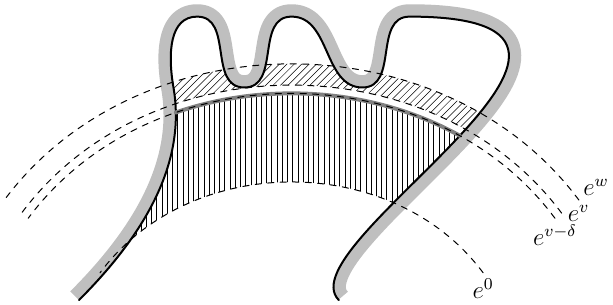}
} 
\hspace{0.2cm}
\subfigure[]
{
	\label{sfig: el integral estimate b}
	\includegraphics[scale=.65]
{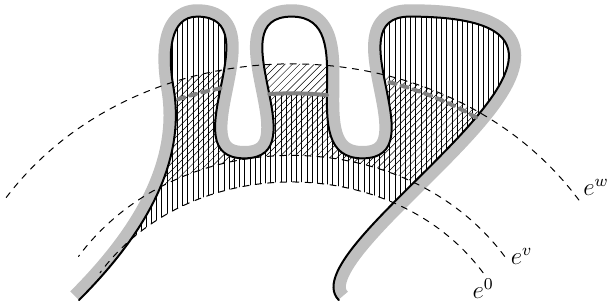}
}
\caption{In this figure, the shading with diagonal lines represents the integral $\ell_v^w\br{{\cal I}^w}$
which in the subfigure \ref{sfig: el integral estimate a} is small 
and in \ref{sfig: el integral estimate b} is big. In the first case 
the arc of the circle of radius $e^{v - \delta}$ gives the arc with desired properties. In the
second case we use the induction hypothesis to find a set of arcs of circles with radii in the range $[e^v, e^w]$.
These arcs are here illustrated by gray lines and
one of the topological quadrilaterals cut by an arc are illustrated in both subfigures by
vertical shading.}
\label{fig: el integral estimate}
\end{figure}

Now assume  that, on the contrary,
$$\ell_v^w\br{{\cal I}^w} \ge \frac2\epsilon.$$
Suppose ${\cal I}^v$ is composed of the arcs $J_k$.
For each $k$ denote by ${\cal I}_k^x$ the collection of arcs $I\in{\cal I}^x$
such that $J_k\prec I$.
Since
\begin{equation}\label{eq:below}
\ell_v^w\br{{\cal I}_k^w}\ge\ell_v^w\br{{\cal I}^w}\ge\frac2\epsilon,
\end{equation}
we can apply the induction assumption to each of those collections
${\cal I}_k^w$ on the interval $x\in[v,w]$,
obtaining a set of separating arcs $\brs{I_{j,k}}_j$ 
such that 
\begin{equation}
\sum_j \exp\br{-\epsilon\,\ell_v\br{I_{j,k}}} \le \exp\br{2-\epsilon\ell_v^w\br{{\cal I}_k}}.
\end{equation}
Then the desired estimate follows from
\begin{align}
\sum_{j,k} \exp\br{-\epsilon\,\ell\br{I_{j,k}}} 
\le& \exp\br{-\epsilon\,\ell_0^v\br{{\cal I}^v}}~\sum_{k} \sum_{j} \exp\br{-\epsilon\,\ell_v\br{I_{j,k}}}\notag\\
\le& \exp\br{-\epsilon\,\ell_0^v\br{{\cal I}^v}}~\sum_{k} \exp\br{2-\epsilon\ell_v^w\br{{\cal I}_k^w}}\notag\\
\overset{*}
{\le}& \exp\br{-\epsilon\,\ell_0^v\br{{\cal I}^v}}~\exp\br{2-\epsilon\ell_v^w\br{{\cal I}^w}}\label{eq:star}\\
=& \exp\br{2-\epsilon\ell\br{{\cal I}^w}},\notag
\end{align}
assuming we have the inequality (\ref{eq:star}$*$) above.
To prove it we first observe that for $x\in[v,w]$,
$$\sum_k\bra{{\cal I}_k^x}=\bra{{\cal I}^x}.$$
Using Jensen's inequality for the probability measure
$$
\frac{dy}{ \bra{{\cal I}^y} \ell_v^w\br{{\cal I}} },$$
and the convex function $x^{-1}$, we write
\begin{align*}
\ell_v^w\br{{\cal I}_k}&=\int_v^w\frac1{\bra{{\cal I}_k^y}}dy
=\int_v^w\br{\frac{\bra{{\cal I}_k^y}}{\bra{{\cal I}^y}
\ell_v^w\br{{\cal I}}}}^{-1} 
\frac{dy}{\bra{{\cal I}^y}\ell_v^w\br{{\cal I}}}\\
&\ge\br{\int_v^w\frac{\bra{{\cal I}_k^y}}{\bra{{\cal I}^y}
\ell_v^w\br{{\cal I}}} 
\frac{dy}{\bra{{\cal I}^y}\ell_v^w\br{{\cal I}}}}^{-1}
=\br{\int_v^w\frac{\bra{{\cal I}_k^y}dy}
{\bra{{\cal I}^y}^2\ell_v^w\br{{\cal I}}^2}}^{-1}.
\end{align*}
Thus
\begin{align}
\notag\sum_k\frac1{\ell_v^w\br{{\cal I}_k}}&
\le\sum_k\br{\int_v^w\frac{\bra{{\cal I}_k^y}dy}
{\bra{{\cal I}^y}^2\ell_v^w\br{{\cal I}}^2}}={\int_v^w\frac{\sum_k\bra{{\cal I}_k^y}dy}
{\bra{{\cal I}^y}^2\ell_v^w\br{{\cal I}}^2}}\\
&={\int_v^w\frac{dy}
{\bra{{\cal I}^y}}\,\frac1{\ell_v^w\br{{\cal I}}^2}}
={\ell_v^w\br{{\cal I}}
\frac1{\ell_v^w\br{{\cal I}}^2}}
=\frac1{\ell_v^w\br{{\cal I}}}.\label{eq:jensen}
\end{align}
An easy differentiation shows that the function $F(x):=\exp\br{-\epsilon/x}$ vanishes at $0$, is increasing and convex on the interval $[0,\epsilon/2]$, and so is sublinear there.
Observing that the numbers $1/\ell_v^w\br{{\cal I}_k}$ as well as their sum
belong to this interval by \eqref{eq:below} and \eqref{eq:jensen}, we can write
\begin{align*}
\sum_{k} \exp\br{-\epsilon\ell_v^w\br{{\cal I}_k}}
&=\sum_{k}F\br{1/\ell_v^w\br{{\cal I}_k}}\le F\br{\sum_{k}1/\ell_v^w\br{{\cal I}_k}}\\
&\le F\br{1/\ell_v^w\br{{\cal I}}}
=\exp\br{-\epsilon\ell_v^w\br{{\cal I}}},
\end{align*}
thus proving the inequality (\ref{eq:star}$*$) and the desired implication.

This completes the circle of implications, thus proving Proposition~\ref{prop: equiv conditions}.

%% file: random_planar_curves_sec3.tex

%
%

\section{Proof of the main theorem} \label{sec: proof main}

In this section, we present the proof of Theorem~\ref{thm: main}. 
As a general strategy, we find an increasing sequence of events $E_n \subset \xs(\disc)$ such that
\begin{equation*}
\lim_{n \to \infty} \inf_{\P \in \Sigma_\disc} \P(E_n) = 1
\end{equation*}
and the curves in $E_n$ have some good properties which among other things guarantee that the closure of $E_n$
is contained in the class of Loewner chains.

The structure of this section is as follows.
To use the main lemma (Lemma~\ref{lm: main lemma with convergence} in appendix, which constructs the Loewner chain) 
we need to verify its
three assumptions.
In Section~\ref{ssec: aizburch}, it is shown that with high probability the curves will have parametrizations
with uniform modulus of continuity. Similarly the results
in Section~\ref{ssec: cont driving} guarantee that the driving processes in the capacity parametrization have
uniform modulus of continuity with high probability.
In Section~\ref{ssec: no six arms}, a uniform result on the visibility of the tip $\gamma(t)$ is proven
giving the uniform modulus of continuity of the functions $F$ of Lemma~\ref{lm: main lemma with convergence}. 
Finally in the end of this section we prove the main theorem and its corollaries.

A tool which makes many of the proofs easier is the fact that we can use always the most suitable form of the
equivalent conditions. Especially, by the results of Section~\ref{ssec: equiv cond}
if Condition~\ref{def: b unf crossing} can be verified in the original domain then 
Condition~\ref{def: b unf crossing} (or any equivalent condition) holds in any reference domain
where we choose to map the random curve as long as the map is conformal. Furthermore, Condition~\ref{def: b unf crossing}
holds after we observe the curve up to a fixed time or a random time and then erase the observed initial part
by conformally mapping the complement back to reference domain.

\subsection{Reformulation of the main theorem}

In this section we reformulate the main result so that its proof amounts to verifying
four (more or less) independent properties, which are slightly technical to formulate.
The basic definitions are the following, see Sections~\ref{ssec: aizburch}, \ref{ssec: cont driving}
and \ref{ssec: no six arms} for more details.
Assume that $\rho_n$ is a decreasing sequence such that $\rho_n \searrow 0$ as $n \to \infty$,
that $\alpha,\alpha',T, R$ are positive numbers 
and that $\psi: [0,\infty) \to [0,\infty)$ is continuous and strictly increasing function with $\psi(0) = 0$. 
Define the following random variables
\begin{align}
N_0 &= \sup\{ n\geq 2 \,:\,
  \gamma \text{ intersects $\partial B(1,\rho_{n-1})$ after intersecting $\partial B(1,\rho_n)$} \} \\
C_{1,\alpha} &= \inf\left\{ C>0 \,:\, 
      \begin{gathered}
	  \gamma \text{ can be parametrized s.t.}\\
	  |\gamma(s) - \gamma(t)| \leq C |t-s|^\alpha \quad \forall (t,s) \in [0,1]^2
      \end{gathered}
      \right\}\\
C_{2,\alpha',T} &= \inf\left\{ C>0 \,:\, 
      \begin{gathered}
	  |W_\gamma(s) - W_\gamma(t)| \leq C |t-s|^{\alpha'}  \quad \forall (t,s) \in [0,T]^2 
      \end{gathered}
      \right\}\\
C_{3,\psi,T,R} &= \inf\left\{ C>0 \,:\, 
      \begin{gathered}
	  |F_\gamma(t,y) - \hat{\gamma}(t)| \leq C \psi(y)
	  \quad \forall (t,y) \in [0,T] \times [0,R]
      \end{gathered}
      \right\}
\end{align}
where $\hat{\gamma} = \Phi(\gamma)$ and 
\begin{equation}\label{eq: def hyperbolic geodesic to tip}
F_\gamma(t,y) = g_t^{-1} ( W_\gamma(t) + i\,y) 
\end{equation}
which can be called a hyperbolic geodesic ending to the tip of the curve.

We will prove the next proposition in 
Sections~\ref{ssec: aizburch}, \ref{ssec: cont driving}
and \ref{ssec: no six arms}. Theorem~\ref{thm: main} follows from the proposition (including the results of the next three
subsections) and Lemma~\ref{lm: main lemma with convergence}.

\begin{proposition}\label{prop: reformulation of main theorem}
If $\Sigma$ satisfies Condition~\ref{def: b unf crossing} and $\Sigma_\disc$ is as in \eqref{eq: def sigma disc}, 
then the following statements hold
\begin{itemize}
\item The random curves $\gamma$, whose laws form the collection $\Sigma_\disc$,
are \emph{transient uniformly}  in the following sense: 
there exists a sequence $\rho_n$ such that the random variable $N_0$ is tight in $\Sigma_\disc$.
\item The family of measures $\Sigma_\disc$ is \emph{tight} in $X$:
There exists $\alpha>0$ such that $C_{1,\alpha}$ is a tight random variable in $\Sigma_\disc$.
\item The family of measures $\Sigma_\disc$ is \emph{tight} in the sense of driving process convergence:
There exists $\alpha'>0$ such that $C_{2,\alpha',T}$ is a tight random variable in $\Sigma_\disc$
for each $T>0$.
\item 
There exists $\psi$ such that $C_{3,\psi,T,R}$ is a tight random variable in $\Sigma_\disc$
for each $T>0$, $R>0$.
\end{itemize}
\end{proposition}

\subsection{Extracting weakly convergent subsequences of probability measures on curves} \label{ssec: aizburch}

In this subsection, we first review the results of \cite{aizenman-burchard-1999-} and then we verify
their assumption (which they call hypothesis H1) given that Condition~\ref{def: b unf crossing} holds.
At some point in the course of the proof, we observe that it is nicer to work with a smooth domain
such as $\disc$, hence justifying the effort needed to prove the equivalence of the conditions.

Aizenman and Burchard \cite{aizenman-burchard-1999-} made the following assumption on a collection of 
probability measures on the space of curves. They called it Hypothesis H1 and for us it is
Condition~\ref{def: pb multiple crossing}.

\begin{condition} \label{def: pb multiple crossing}
A collection of measures $\Sigma_0$ on $X(\C)$ is said to satisfy a \emph{power-law bound on multiple crossings} 
if for each $n$, there are constants
$\Delta_n \geq 0$, $K_n > 0$ such that
\begin{equation}
\P \big( \gamma \textrm{ makes $n$ crossings of } A(z_0,r,R) \big) \leq K_n \left( \frac{r}{R} \right)^{\Delta_n} 
   \label{eq: pb multiple crossing}
\end{equation}
for any annulus $A(z_0,r,R)$ and for each $\P \in \Sigma_0$ and that satisfy $\Delta_n \to \infty$ as $n \to \infty$.
\end{condition}

\begin{remark}
The sequence $(\Delta_n)$ can trivially be chosen to be non-decreasing. Hence it is actually enough to check that
$\Delta_{n_j} \to \infty$ along a subsequence $n_j \to \infty$. 
\end{remark}

Based on this assumption Aizenman and Burchard proved the following result, see Theorem~1.1 and Theorem 2.3 in \cite{aizenman-burchard-1999-}.

\begin{theorem}[Aizenman--Burchard \cite{aizenman-burchard-1999-}]\label{thm: a-b}
Assume that a collection of measures $\Sigma_0$ on $X(\C)$ satisfies Condition~\ref{def: pb multiple crossing} and
that $\gamma$ is uniformly bounded, i.e., 
there exists $R>0$ such that $\P(\gamma \subset B(0,R))=1$ for all $\P \in \Sigma_0$
Then the following statements hold.

\begin{enumerate}
\item The family of $\Sigma_0$ is tight and hence any sequence in $\Sigma_0$ contains a weakly convergent
subsequence.
\item There exists exponents $\alpha>0$ and $\beta>0$ such that following random variables are tight on $\Sigma_0$ 
\begin{align}
Z_\alpha (\gamma) 
  &= \sup \big\{ M(\gamma,l) \cdot l^{\alpha} \; : \; 0<l<1 \big\} \\
\hat{Z}_\beta(\gamma) 
  &= \inf_{\hat{\gamma}} \, \sup \, \big\{ w(\hat{\gamma},\delta) \cdot \delta^{-\beta} 
             \; : \; 0<\delta<1 \big\}       \label{eq: umc}
\end{align}
where we use the following definitions.
The random variable $M(\gamma,l)$ is the minimum of the numbers $n$ such that there exists 
a partition $0 = t_0 < t_1 < \ldots < t_n = 1$ of the time interval $[0,1]$ such that
$\diam ( \gamma[ t_{k-1}, t_k ] ) \leq l$ for any $k = 1, 2,\ldots,n$. 
The random variable $w(\hat{\gamma},\delta)$ is the modulus of continuity of the parametrization $\hat{\gamma}$
of $\gamma$, that is,
\begin{equation}
w(\hat{\gamma},\delta) = \max \{ \, |\hat{\gamma}(t) - \hat{\gamma}(s)| \;:\; (s,t) \in [0,1]^2 \textrm{ s.t. } |s-t| \leq \delta \} .
\end{equation}
The infimum in $\hat{Z}_\beta(\gamma) $ is over all parametrizations $\hat{\gamma}$ of $\gamma$.
\end{enumerate}
\end{theorem}

\begin{remark}
A bound of the type $Z_\alpha(\gamma) \leq K$ for some $K>0$ and $\alpha > 0$ was called \emph{tortuosity bound} in 
\cite{aizenman-burchard-1999-} and similarly bound for $Z_\beta(\gamma) \leq K$ for some $K>0$ and $\beta > 0$
is \emph{modulus of continuity bound}. Existence of one type of bound implies existence of the other bound,
which might hint how Condition~\ref{def: pb multiple crossing} is sufficient assumption for this result. 
\end{remark}

\begin{remark}
The compact subsets $K \subset X$ were characterized in Lemma~4.1 in \cite{aizenman-burchard-1999-}.
A closed set $K \subset X$ is compact if and only if there exists a function $\psi:(0,1] \to (0,1]$ such that
\begin{equation*}
M(\gamma,l) \leq \frac{1}{\psi(l)}
\end{equation*}
for any $\gamma \in K$ and for any $0< l \leq 1$. 
And this is equivalent to the existence of parametrization which allows a uniform bound on the modulus of continuity.
\end{remark}

We will use the remainder  of this section to show that Condition~\ref{def: pb unf crossing}
implies Condition~\ref{def: pb multiple crossing} and hence the results of Theorem~\ref{thm: a-b}.
Notice that we assume Condition~\ref{def: pb unf crossing} in the original domain while
Condition~\ref{def: pb multiple crossing} is 
shown to hold in
a smooth and bounded reference domain which we choose to be $\disc$.

\begin{proposition} \label{prop: cond implies aizbur}
If $\Sigma$ satisfies Condition~\ref{def: pb unf crossing}, then $\Sigma_\disc$ satisfies Condition~\ref{def: pb multiple crossing}.
Hence then also the conclusions of Theorem~\ref{thm: a-b} hold.
\end{proposition}

Let $D_t = \disc \setminus \gamma(0,t]$. 
Let $\tilde{C}>1$.
For an annulus $A=A(z_0,r,\tilde{C}^3 r)$ define  three concentric subannuli $A_k = A(z_0, \tilde{C}^{k-1} r, \tilde{C}^k r)$, $k=1,2,3$.
Define the \emph{index}  $I(A,D_t) \in \{0,1,2,\ldots\}$ 
of $\gamma$ at time $t$ with respect to $A$ to be 
the minimal number of crossings of $A_2$ made by $\tilde{\gamma}$ where $\tilde{\gamma}$ runs over the set of 
all possible futures of $\gamma[0,t]$
\begin{equation*}
\{\tilde{\gamma} \in \xs(D_t) \,:\, \tilde{\gamma} \textrm{ connects } \gamma(t) \textrm{ to } b \} .
\end{equation*}

Consider a sequence of stopping times $\tau_0 = 0$ and 
\begin{equation*}
\tau_{k+1} = \inf\{ t > \tau_{k} : \gamma [\tau_k,t] \textrm{ crosses } A \}
\end{equation*}
where $k=0,1,2,\ldots$ Define also $\sigma_0=0$ and
\begin{equation*}
\sigma_{k+1} = \inf\{ t > \sigma_{k} : \gamma [\sigma_k,t] \textrm{ crosses } A_2 \} .
\end{equation*}

Since $\gamma(\tau_k)$ and $\gamma(\tau_{k+1})$ lie in the different components of $\C \setminus A$,
the curve $\gamma [\tau_k,\tau_{k+1}]$ has to cross $A_2$ an odd number of times. 
Hence there are odd number of $l$ such that $\tau_k < \sigma_{l+1} < \tau_{k+1}$.
For each $l$, $\gamma[ \sigma_l, \sigma_{l+1}]$ crosses $A_2$ exactly once and therefore the index changes by $\pm 1$.
From this it follows that 
\begin{equation*}
I(A,D_{\tau_{k+1}}) = I(A,D_{\tau_k}) + 2n -1
\end{equation*}
with $n \in \Z$.

\begin{lemma} \label{lm: index unforced crossings}
Let $A=A(z_0,r,R)$ be an annulus and let $A_k, k=1,2,3$ be its subannuli as above.
\begin{enumerate}\enustyii
\item\label{enui: iuc i}
If $A$ is not on $\partial \disc$, i.e. $\overline{B(z_0,r)} \cap \partial \disc = \emptyset$,
then on the event $\tau<1$, $I(A,D_\tau)=1$, where $\tau$ is the hitting time of $\overline{B(z_0,r)}$.
\item\label{enui: iuc ii}
If $A$ is on $\partial D_s$, i.e. $\overline{B(z_0,r)} \cap \partial D_s \neq \emptyset$, and
the index increases from $I$ to $I + 2n - 1$, $n \geq 1$, during a minimal crossing $\gamma[s,t]$ of $A$
then the total number of unforced crossings of the annuli $A_k$, 
$k = 1,2,3$, made by $\gamma[s,t]$ has to be at least $2n-1$.
\end{enumerate}
\end{lemma}

\begin{proof}
The statement \ref{enui: iuc i} can be easily verified  
since the point $+1$ can be reached from $\gamma(\tau)$ in $D_s$ while making only one crossing by 
following a path close to the boundary of $D_s$.

Suppose now that $A$ is on $\partial D_s$.
Let $m \leq m'$ be such that 
\begin{equation*}
\sigma_{m-1} < s < \sigma_{m} \textrm{ and } \sigma_{m'} < t < \sigma_{m'+1} .
\end{equation*}
As we observed above if we set
$y_l \dd= I(A,D_{\sigma_{l}}) - I(A,D_{\sigma_{l-1}})$
then these changes in the index take values $y_l \in \{-1,1\}$ and they sum up to
\begin{equation*}
\sum_{l=m}^{m'} y_l = 2n -1, 
\end{equation*}
that is, to the total change of the index during $[s,t]$.

We claim that the following two statements hold:
\begin{itemize}
\item If $y_{m'} = 1$ then the last crossing $\gamma[\sigma_{m'},t]$ 
of a component of $A_1$ or $A_3$ has to be unforced 
as observed at time $\sigma_{m'}$.
\item If $y_{l} = 1 = y_{l+1}$ then the latter crossing  
$\gamma[ \sigma_{l}, \sigma_{l+1}]$ is an unforced crossing of $A_2$ 
as observed at time $\sigma_{l}$.
\end{itemize}
To prove these claims, let $h=m'$ or $h=l$ (depending on the claim, respectively) and suppose that $y_h=1$.
Let $C_0 \subset \partial A_2 \cap D_{\sigma_h}$ be the boundary arc of $A_2$ which has the property that
any curve from $\gamma(\sigma_h)$ to $+1$ in $D_{\sigma_h}$ has to intersect $C_0$ and $C_0$ is not separated from
$\gamma(\sigma_h)$ by any other such arc. Let $V_0$ be the component of $D_{\sigma_h}\setminus C_0$ which contains 
(a neighborhood of) $+1$
and let $V_1$ to be the component of $A_2 \cap D_{\sigma_h}$ which has $\gamma(\sigma_h)$ and $C_0$ on its boundary.
See Figure~\ref{fig: crossings a2}.
Then $\gamma(\sigma_h)$ can be connected to $+1$ in $V_0 \cup C_0 \cup V_1$.
Because we assumed that $y_k=1$, $\gamma(\sigma_h)$ and $C_0$ are in the different circular boundary arcs of the annulus $A_2$
and thus
it is clear that if $V_1 \subset A_2$ was crossed next, then the index would decrease by one.
Hence the next crossing in both of the scenarios has to be in the complement of $V_0 \cup C_0 \cup V_1$.
Since $\gamma(\sigma_h)$ can be connected to $+1$ in $V_0 \cup C_0 \cup V_1$, this crossing is unforced
as observed at time $\sigma_h$. Thus the claims hold.

\begin{figure}[htb]
\begin{center}
{
\includegraphics[scale=.9]{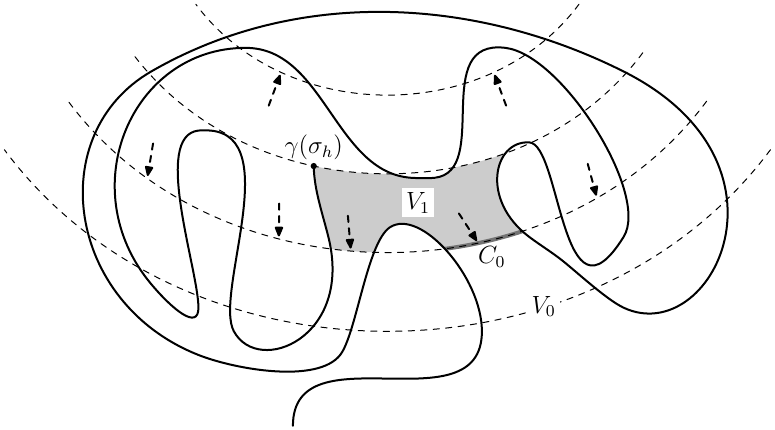}
}
\caption{
A sector of three concentric annuli with an initial segment of the interface. The boundaries of the annuli
are the dashed circular arcs (which are only partly shown in the figure).
If the index increases between times $\sigma_{h-1}$ and $\sigma_h$, then $\gamma(\sigma_h)$ and $C_0$
are in different circular arcs . If the next crossing was in $V_1$, it would decrease the index.
Hence in the scenarios in the proof of Lemma~\ref{lm: index unforced crossings}, the crossings corresponding to
arrows on grey background are not allowed. Therefore a crossing which corresponds to one of the arrows
on white background is going to happen next and those crossings are unforced.
}
\label{fig: crossings a2}
\end{center}
\end{figure}

The rest of the proof is divided in two cases depending on $y_{m'} \in \{-1, 1\}$.
If $y_{m'} = -1$, then
\begin{equation*}
\max_{j = m,\ldots,m'} \sum_{l=m}^{j} y_l \geq 2n  .
\end{equation*}
Therefore there has to be at least $2n-1$ pairs $(l,l+1)$ so that $y_{l} = 1 = y_{l+1}$. This can be easily proven
by induction. Hence the statement~\ref{enui: iuc ii} holds in this case by the second property we proved above.

If $y_{m'} = 1$, then there are at least $2n-2$ pairs $(l,l+1)$ so that $y_{l} = 1 = y_{l+1}$ by the same argument as in the
previous case.
In addition to this the last crossing $\gamma[\sigma_{m'},t]$ is unforced crossing of 
$A_1$ or $A_3$ by the first property we proved above. 
Hence the statement~\ref{enui: iuc ii} holds
also in this case.
\end{proof}

Now we are ready to give the proof of the main result of this section. Notice that here we need that the domain is smooth
otherwise the number $n_0$ below wouldn't be bounded. There are of course many ways to bypass this: 
for instance, if we
want the measures to be supported on H\"older curves (including the end points on the boundary), then we need to assume
that minimal number of crossings of annuli $A(z_0,r,R)$ centered at $z_0 = a$ or $z_0 = b$ grows at most
as a power of $r$ as $r \to 0$.

\begin{proof}[Proof of Proposition~\ref{prop: cond implies aizbur}]
We will prove the first claim that if $\Sigma$ satisfies Condition~\ref{def: pb unf crossing}, 
then $\Sigma_\disc$ satisfies Condition~\ref{def: pb multiple crossing}. 
The rest of the proposition follows then from
the results of
\cite{aizenman-burchard-1999-} which we formulated above in Theorem~\ref{thm: a-b}.

First of all, we can concentrate on the case that the variables $z_0,r,R$ are bounded. 
We can assume that $z_0 \in B(0,3/2), r<1/2, R<1$. In the complementary case either the left-hand side of
\eqref{eq: pb multiple crossing} is zero by the fact that there are no crossing of the annulus that stay inside
the unit disc or the ratio $r/R$ is uniformly bounded away from zero. In the latter case the constant $K_n$ can be
chosen so that the right-hand side of \eqref{eq: pb multiple crossing} is greater than one and \eqref{eq: pb multiple crossing}
is satisfied trivially.

Denote as usual $A=A(z_0,r,R)$.
By the fact that $R<1$, at most one of the points $\pm 1$ is in $A$. If either $\pm 1$ is in $A$, denote the distance from that point
to $z_0$ by $\rho$. Then $r < \rho < R$ and a trivial inequality shows that 
\begin{equation*}
\max \left\{ \frac{\rho}{r}, \frac{R}{\rho} \right\} \geq \sqrt{ \frac{R}{r} } .
\end{equation*}
Hence for each annulus, it is possible to choose a smaller annulus inside it
so that the points $\pm 1$ are away from that annulus and 
the ratio of the radii is still at least square root of the original one.
If we are able to show existence of the constants $K_n$ and $\Delta_n$ for annuli $A$ such that $\{-1,1\} \cap A = \emptyset$ then
constants $\hat{K}_n = K_n$ and $\hat{\Delta}_n = \Delta_n / 2$ can be used for a general annulus.

Let $A$ be such that $\{-1,1\} \cap A = \emptyset$ 
and set $n_0$ and $\tau$ in the following way:
if $\overline{B(z_0,r)}$ intersects the boundary, let $n_0=1$ when $\overline{B(z_0,r)}$ contains $-1$ or $1$
and $n_0=0$ otherwise and let $\tau=0$.
If $\overline{B(z_0,r)}$ doesn't intersect the boundary, let $n_0=2$ and let 
$\tau= \inf\{ t \in [0,1] \,:\, \gamma(t) \in \overline{B(z_0,r)}\}$.

By Lemma~\ref{lm: index unforced crossings}, 
if there is a crossing of $A$ that increases the index, there are unforced crossings of the annuli $A_k$, $k=1,2,3$. 
We can apply this result after time $\tau$.
If the curve doesn't make any unforced crossings of the annuli $A_k$, $k=1,2,3$, then there are at most
$n_0$ crossings of $A$.
This argument generalizes so that
if there are $n > n_0$ crossings of $A$, we apply Condition~\ref{def: pb unf crossing}
$(n-n_0)/2$ times in the annuli $A_k$, $k=1,2,3$, to get the bound 
\begin{equation*}
\P \left( \gamma \textrm{ makes $n$ crossings of } A(z_0,r,R)  \right)
     \leq K^{ \frac{n-n_0}{2}  } \,\cdot\, \left( \frac{r}{R} \right)^{  \frac{\Delta}{6} ( n-n_0 ) }
\end{equation*}
for any $\P \in \Sigma_\disc$. 
Hence the proposition holds for $\Delta_n = \Delta \cdot (n-2)/12$.
\end{proof}

\subsection{Continuity of driving process and finite exponential moment} \label{ssec: cont driving}

Let $\Phi: \disc \to \half$ be a conformal mapping such that $\Phi(-1)=0$ and $\Phi(1)=\infty$. To make the choice
unique, it is also possible to fix $\Phi(z) = \frac{2i}{1-z} + \OO(1)$ as $z \to 1$, i.e.
\begin{equation}
\Phi(z) = i \frac{1+z}{1-z} .
\end{equation}
Denote by $\Phi_t = \Phi \circ g_t$. We often shorten the notation by writing $\Phi \gamma = \Phi(\gamma)$.

Denote by $W(\,\cdot\,,\Phi \gamma)$ the driving process of $\Phi \gamma$ in the capacity parametrization.
Our primary interest is to estimate the tails of the distribution of the increments of the driving process.
Let's first study what kind of events are those when $|W(t,\Phi \gamma) -W(s,\Phi \gamma)|$ is large. 
Suppose that $u$ and $L$ are positive real numbers such that $u/L$ is small. Consider a hull $K$ that is a subset of a rectangle
$R_{L,u}=[-L,L] \times [0,u]$. If $K \cap [L,L+iu]\neq \emptyset$ then
for any $z$ in this set, 
$0.9 \, L \leq g_K(z) \leq 1.1 \, L$ as proved below in Lemma \ref{lm: exitpointest}. On the other hand
if $K \cap \, [-L +iu,L+iu] \neq \emptyset$ then $\hcap (K) \geq \frac{1}{4} u^2$. 
This is proved in Lemma \ref{lm: hcap}.

Based on this observation the following inequality holds 
\begin{equation} \label{ie: tail of driving process}
\P \left( \, \left|W \left(\frac{1}{4} u^2, \Phi\gamma \right)  \right| \geq 2 L \, \right) 
  \leq \P \big( \, \real [ (\Phi \gamma)( \tau_{R_{L,u}} ) ] = \pm L 
  \, \big) ,
\end{equation}
where $\tau_{R_{L,u}} = \inf \{ \, t \in [0,1] \; : \; \Phi \gamma (t) \in \half \cap \partial R_{L,u} \, \}$.
Therefore we study the event that the curve exits a rectangular neighborhood of the origin in the upper half-plane
through the sides of the rectangle. Notice also that the capacity $t = u^2/4$ corresponds to the height $2 \sqrt{t} = u$
of the rectangle
in the inequality~\eqref{ie: tail of driving process}. This is ultimately the source for the
exponent $\alpha < 1/2$  and for the term
$\sqrt{t}$ in \eqref{ie: main thm integrability} in the main theorem  (Theorem~\ref{thm: main}).
Figure~\ref{fig: crossings cont driving} illustrates both this correspondence and the proof of the next proposition.

\begin{figure}[htb]
\begin{center}
\includegraphics[scale=1.1]{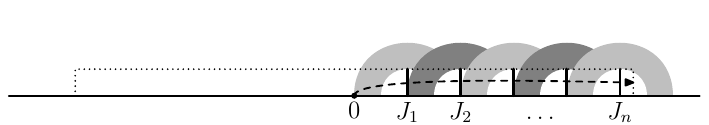}
\caption{If $\sup \{ |W_u - W_s| \; : \; u \in [s,t] \} \geq L$, then the curve $u \mapsto g_s(\gamma (u)) - W_s$, 
$s \leq u \leq t$, exits
the rectangle $[-L,L] \times [0,2\sqrt{t-s}]$ from one of the sides $\{ \pm L \} \times [0,2\sqrt{t-s}]$. Especially
the curve has to intersect all the vertical lines and make an unforced crossing of each of the annuli centered
at the base points of those lines.}
\label{fig: crossings cont driving}
\end{center}
\end{figure}

\begin{proposition}\label{prop: basic H rectangle estimate}
If Condition~\ref{def: b unf crossing} holds, then there are constants $K>0$ and $c>0$ so that
\begin{equation} \label{ie: rect crossing}
\P \big( \, \real [ (\Phi \gamma)( \tau_{R_{L,u}} ) ] = \pm L \, \big) \leq K e^{-c \frac{L}{u}} 
\end{equation}
for any $0<u<L$.
\end{proposition}

\begin{proof}
If Condition~\ref{def: b unf crossing} holds then it also holds in $\half$ by the results of 
Section~\ref{ssec: equiv cond}.
Let $C>1$ be the constant of Condition~\ref{def: b unf crossing} in $\half$.

By symmetry, it is enough to consider the event $E$ that
$\Phi \gamma$ exits the rectangle $R_{L,u}$ from the right-hand side $\{ L \} \times [0,u]$. 
Let $n = \lfloor L/(C u) \rfloor$. Consider the lines $J_k = \{ C u \cdot k \} \times [0,u]$, $k=1,2,\ldots,n$.
On the event $E$, each of the lines $J_k$ are hit before $\tau_{R_{L,u}}$ and the hitting times are ordered
\begin{equation*}
0 < \tau_{J_1} < \tau_{J_2} < \ldots < \tau_{J_n} \leq \tau_{R_{L,u}} <1
\end{equation*}
See Figure~\ref{fig: crossings cont driving}.

Let $x_k = C u \cdot k$ which is the base point of $J_k$.
On the event $E$ the annulus $A(x_1,u,Cu)$ is crossed and after each $\tau_{J_k}$ the annulus $A(x_{k+1},u,Cu)$
is crossed. Hence Condition~\ref{def: b unf crossing} can be applied with the stopping times 
$0, \tau_{J_1}, \ldots , \tau_{J_{n-1}}$ and the annuli $A(x_1,u,Cu)$, $A(x_2,u,Cu), \ldots$, $A(x_n,u,Cu)$. This gives the upper bound
$2^{-n}$ for the probability of $E$. Hence the inequality~\eqref{ie: rect crossing} follows with suitable constants depending only
on $C$. 
\end{proof}

We can now apply the above bounds~\eqref{ie: tail of driving process} and \eqref{ie: rect crossing}
to show the next proposition which can be interpreted in the following way.
The first statement shows the uniform transience of the curves (uniform over $\P \in \Sigma_\disc$)
in the same sense as in Proposition~\ref{prop: reformulation of main theorem}.
The second statement is a sufficient technical statement for 
the H\"older continuity of the driving processes and is used
in the proof of Theorem~\ref{thm: holder exponent driving}.
The third statement is needed for the exponential integrability of the driving process
in Theorem~\ref{thm: main}.

\begin{proposition} \label{pr: contdrv} \label{thm: holder cont}
Let $\upsilon(t)=\hcap ( \Phi\gamma[0,t] )/2$ for any $t \in [0,1)$ and define 
$\upsilon(1) = \lim_{t \to 1} \upsilon(t) \in (0,\infty]$.
If Condition~\ref{def: b unf crossing} holds, then 
\begin{enumerate}\enustyo
\item \label{enui: pr contdrv 1}
For all $\P \in \Sigma_\disc$, $\P(\upsilon(1) = \infty)=1$. There exists a sequence $b_n \in \R$ such that
\begin{equation}\label{ie: prob capacity goes to infinity}
\P\left( \sup_{0 \leq t \leq n } |W_t(\hat{\gamma})| \leq b_n \right) \geq 1 - \frac{1}{n}
\end{equation}
for any $\P \in \Sigma_\disc$.
\item \label{enui: prop contdrv 2}
Fix $T>0$ and $0<\alpha< \frac{1}{2}$. 
Let $X' \subset \xs(\disc)$ be the set of simple curves such that $\upsilon(1) > T$. Define
\begin{equation}
G_n = \left\{ \gamma \in X' : \sup_{j 2^{-n} \leq u \leq (j+1) 2^{-n}}
            |W_u(\hat{\gamma}) - W_{j 2^{-n}}(\hat{\gamma})| \leq 2^{-\alpha n} 
            \right\} .
\end{equation}
Then for large enough $n \geq n_0(\alpha,T,K,c)$
\begin{align*}
\P \big(  G_n \big)
   \geq 1 - 2^{- n} 
\end{align*}
for any $\P \in \Sigma_\disc$.
\item \label{enui: pr contdrv 3}
There exists constants $\eps>0$ and $C>0$ such that 
\begin{equation}
\E_\P \left[\exp\left(\eps \max_{s \in [0,t]} |W_s(\hat{\gamma})|/\sqrt{t} \right)\right] \leq C
\end{equation}
for any $t>0$ and for any $\P \in \Sigma_\disc$. Here $\E_\P$ is the expected value
with respect to $\P$.
\end{enumerate}
\end{proposition}

\begin{proof}
Notice first that in the inequality~\eqref{ie: tail of driving process}
we can replace $|W (u^2/4, \Phi\gamma )|$ on the left by $\max_{0 \leq s \leq u^2/4} |W (u^2/4, \Phi\gamma )|$.
This stronger version follows from the very same observation.

1. Let $b_n = (4/c) \, \sqrt{n} \log (K \, n)$. Then by \eqref{ie: tail of driving process} and
\eqref{ie: rect crossing}
\begin{equation}
\P\left( \sup_{0 \leq t \leq n } |W_t(\hat{\gamma})| > b_n \right) 
   \leq K \exp \left( - c \frac{b_n}{4 \sqrt{n}} \right) = \frac{1}{n} .
\end{equation}
In particular, $\P(\upsilon(1)=\infty)=1$.

2. Estimate the probability of the complement of $G_n$ by the following sum
\begin{align*}
\P \left( G_n^c \right)
  &\leq \sum_{j=0}^{2^n} 
     \P \left( \max_{u \in [T(j-1) 2^{-n},T j 2^{-n}]}|W_u - W_{T(j-1) 2^{-n}} | > 2^{- \alpha n } \right) \\
  &\leq K 2^{n} e^{-(c/4) T^{-1/2} 2^{(1/2-\alpha) n} } \leq 2^{-n}
\end{align*}
for $n$ large enough depending on $\alpha, T, K, c$.

3. Fix $t>0$ and denote the random variable $\max_{s \in [0,t]} |W_s(\hat{\gamma})|$ by $Z$.
Let $\eps>0$, which we fix in a moment, and $\phi(x) = \exp(\eps x /\sqrt{t})$.
Then by an equality following from Fubini's theorem and by the bound~\eqref{ie: rect crossing}
\begin{align*}
\E_\P ( \phi(Z) ) &= \phi(0) + \int_0^\infty \phi'(x) \P( Z \geq x ) \,\de x 
  \leq 1+ \frac{K \eps}{\sqrt{t}} \int_0^\infty \exp\left( \left(\eps  - \frac{c}{4} \right)\frac{x}{\sqrt{t}} \right) \de x \\
  &= 1+K \eps \int_0^\infty \exp\left( \left(\eps  - \frac{c}{4} \right)  y \right) \de y
\end{align*}
where $K,c$ are as in \eqref{prop: basic H rectangle estimate}. Choose $\eps < \frac{c}{4}$. Then the constant on the right
is finite. It is also independent of $t$ and $\P$ as claimed.
\end{proof}

Finally we reformulate the above somewhat technical results into the following cleaner theorem
(implied by the previous proposition as explained above)
on the H\"older continuity of the driving processes.
The theorem follows from the statement \ref{enui: prop contdrv 2} of Proposition~\ref{pr: contdrv} above and
Lemma~7.1.6 and the proof of Theorem~7.1.5 in \cite{durrett-1996-}.

\begin{theorem} \label{thm: holder exponent driving}
If Condition~\ref{def: b unf crossing} holds, then for each $\P \in \Sigma_\disc$ the curve $\gamma$ is
a Loewner chain which has $\alpha$-H\"older continuous driving process $\P$-almost surely for any $0 < \alpha < 1/2$
and the $\alpha$-H\"older norm of the driving process restricted to $[0,T]$ for $T>0$ is stochastically bounded.
\end{theorem}

\subsection{Continuity of the hyperbolic geodesic to the tip} \label{ssec: no six arms}

In the proof of the main theorem, we are going to apply Lemma~\ref{lm: main lemma with convergence}
of the appendix. Therefore we repeat here the following definition:
for a simple curve $\gamma$ in $\half$, let $(g_t)_{t \in \R_+}$ and $(W(t))_{t \in \R_+}$ be its
Loewner chain and driving function. Then we define the \emph{hyperbolic geodesic from $\infty$ to the tip $\gamma(t)$} 
as $F: \R_+ \times \R_+ \to \overline{\half}$ by
\begin{equation*}
F(t,y) = g_t^{-1} (W(t) + i\,y) .
\end{equation*}
The corresponding geodesic in $\disc$ for the curve $\Phi^{-1} \gamma$ is
\begin{equation}
F_{\disc}(t,y) = \Phi^{-1} \circ F(t,y) .
\end{equation}

Consider now the collection $\Sigma_\disc$ 
and the random curve $\gamma$ in $\xs(\disc,-1,+1)$. Define $F$ and $F_\disc$ as above for the curves $\Phi \gamma$ 
and $\gamma$, respectively. 
For $\rho>0$, let $\tau_\rho$ 
be the hitting time of $B(1,\rho)$, i.e., $\tau_\rho$ is the smallest $t$ such that $|\gamma(t)-1| \leq \rho$.
The following is the main result of this subsection.

\begin{theorem}\label{thm: hyp geodesic tip}
Suppose that $\Sigma$ satisfies Condition~\ref{def: b unf crossing}.
There exists a continuous increasing function $\psi:\R_+ \to \R_+$ such that $\psi(0)=0$ and
for any $\rho>0$ and $\eps>0$ there exists $\delta > 0$ such that
\begin{equation}
\P \left( 
  \begin{gathered}
  \sup_{t \in [0,\tau_\rho]} \left|F(t,y') - F(t,y)\right| \leq \psi(|y-y'|) \\
  \forall y,y' \in [0,L] \text{ s.t. } |y-y'|\leq \delta 
  \end{gathered}  
  \right)
  \geq 1 - \eps 
\end{equation}
for each $\P \in \Sigma_\disc$.
\end{theorem}

The proof is postponed after an auxiliary result, which is interesting in its own right. Namely, the next proposition
gives a ``superuniversal'' arms exponent, i.e., the property is uniform for basically all models of statistical physics: 
under Condition~\ref{def: b unf crossing} a certain event involving
six crossings of an annulus has small probability to occur \emph{anywhere}. Therefore the corresponding
six arms exponent, if it exists, has value always greater than $2$.
To see this, suppose that the probability of this six arms event in a single annulus $A(z_0,r,R)$ tends to zero
as $r^\alpha$ when $r \to 0$. Then we can sum over the lattice $r \Z^2$ and all annuli of the form $A(z,2r,R/4)$
where $z$ is a lattice point and get upper and lower bounds of the form $r^{\alpha -2}$ for seeing this six arms event in
anywhere (in any annulus of the form $A(z_0,r,R)$). Hence if this goes to zero, we must have $\alpha>2$.

Let $D_t = \disc \setminus \gamma(0,t]$
and define the following event event $E(r,R) = E_\rho(r,R)$ on $\xs(\disc)$:
Define $E(r,R)$ as the event that there exists $(s,t) \in [0,\tau]^2$ with $s<t$ such that
\begin{itemize}
\item 
$\diam \big(  \, \gamma[s,t] \, \big) \geq R$ and
\item  
there exists 
a crosscut $C$, $\diam( C ) \leq r$, that separates $\gamma(s,t]$ from $B(1,\rho)$ in $\disc \setminus \gamma(0,s]$.
\end{itemize}
Denote the set of such pairs $(s,t)$ by $\mathcal{T}(r,R)$.

Let's first demonstrate that the event $E(r,R)$ implies a certain six arms event 
(four arms if it occurs near the boundary)
occurring somewhere in $\disc$ --- the converse statement is also true, although we don't need it here.
If $C$ is as in the definition of $E(r,R)$, then for $r < \min\{\rho,R\}/2$
at least one of the end points of $C$ has to lie on $\gamma(0,s]$. Let $T(C) \leq s $
be the largest time such that $\gamma(T(C)) \in \overline{C}$.
Then also $(T(C),t) \in \mathcal{T}(r,R)$ and we easily see that $\gamma[T(C),t]$ makes a crossing
of $(A(\gamma(T(C)),r,R/2))^u_{T(C)}$ and is therefore unforced.
Moreover, $\gamma[0,T(C)]$ contains at least three crossings of $A(\gamma(T(C)),r,R/2)$,
when $\gamma(T(C))$ is sufficiently far from the boundary, or one crossing, when
$\gamma(T(C))$ is close to the boundary. Otherwise the above crossing couldn't be unforced. See also
Figure~\ref{sfig: lca}. Finally, after $t$ the curve $\gamma$ has to still make at least two crossings to reach
the target point $+1$.
Adding these numbers together, we conclude that on the event $E(r,R)$
there is $z_0 \in \disc$ such that $A(z_0,r,R/2)$ contains at least six crossings when $|z_0|<1-r$ or
four crossings when $|z_0|\geq 1- r$ and at least one of the crossings is unforced.

Now we know that $E(r,R)$ is a proper subevent of the full six arms event. By the next result
its probability is small.

\begin{proposition}
If $\Sigma_\disc$ satisfies Condition~\ref{def: b unf crossing}, then as $r\to 0$
\begin{equation*}
\sup \{ \P \left(  E (r,R) \right) \,:\, \P \in \Sigma_\disc \} = \oo(1) .
\end{equation*}
\end{proposition}

\begin{remark}
Since $\P \left(  E (r,R) \right)$ is decreasing in $R$, the bound is uniform for $R \geq R_0 > 0$.
\end{remark}

The idea of the proof is the following: divide the curve $\gamma$ into $N$ arcs
\begin{equation}\label{eq: def gamme cut to pieces}
J_k = \gamma[\sigma_{k-1}, \sigma_k]
\end{equation}
$0=\sigma_0 < \sigma_1 < \ldots < \sigma_N=1$ such that $\diam(J_k) \leq R/4$, $k=1,2,\ldots,N$.
Let $J_0 = \partial \disc$. For the event $E(r,R)$, firstly there has to exist a fjord of depth $R$ 
with a mouth formed by some pair $(J_j,J_k)$, $j<k$, and the number of such pairs is less than $N^2$.
Secondly, there has to be a piece of the curve which enters the fjord, hence resulting in an unforced crossing. 
Hence (given $N^2$) the probability that $E(r,R)$ occurs is less than $const. \cdot N^2 (r/R)^\Delta$.

\begin{proof}
Suppose that $0 < r < R/20$. We will specify more carefully in the end of the proof how small $r$ is for given $R$.

It is useful to do this by defining $\sigma_k$ as stopping times by setting $\sigma_k = 0$, $k \leq 0$, and then recursively
\begin{equation*}
\sigma_k =  \sup \left\{ t \in [\sigma_{k-1},1] : \diam \big( \gamma[\sigma_{k-1},t] \big) < \frac{R}{4} \right\} .
\end{equation*}
Let $J_k$, $k>0$, be as in \eqref{eq: def gamme cut to pieces} and let $J_0 = \partial \disc$.
Observe that if the curve is divided into pieces that have diameter at most $R/4-\eps$, $\eps>0$, then
none of these pieces can contain more than one of the $\gamma(\sigma_k)$.
Therefore $N \leq \inf_{\eps>0} M(\gamma,R/4 - \eps) \leq M(\gamma,R/8)$ 
where $M$ is as in Theorem~\ref{thm: a-b}. By that theorem
$N$ is stochastically bounded, which we will use below.

Define also stopping times
\begin{equation}
\tau_{j,k} = \inf \{ \, t \in  [\sigma_{k-1}, \sigma_k] \;:\; \dist ( \gamma(t), J_j ) \leq 2 r \} .
\end{equation}
for $0 \leq j<k$. If the set is empty, let's define the infimum to be equal to $1$.

Suppose that the event $E(r,R)$ occurs. Take a crosscut $C$ and a pair of times $0 \leq s <t \leq 1$
as in the definition of $E(r,R)$. 

Let $V \subset D_s$ be the connected component $D_s \setminus C$ which is disconnected from $+1$ by $C$ in $D_s$.
Let $j< k$ be such that the end points of $C$ are on $J_j$ and $J_k$. Then it holds that 
the stopping time $\tau \dd= \tau_{j,k}<1$ and we can set $z_1 = \gamma(\tau_{j,k})$.
Let $z_2$ be any point on $J_j$ such that $|z_1 - z_2 | = 2r$. 

Let $C' = [z_1,z_2] \dd= \{ \lambda \, z_1 +(1-\lambda) \, z_2 \,:\, \lambda \in [0,1]\}$
and
\begin{equation*}
V' = \{ z \in V \,:\, z \text{ is disconnected from $+1$ by $C'$ in $D_\tau$} \}
\end{equation*}
and let $D' = D_s \setminus V$.

We claim that the event of an unforced crossing of $(A(z_1,2r,R/2))^u_{\tau}$ occurs.

To prove this, notice first that $\partial D' = \partial \disc \cup \gamma[0,t_C] \cup C$ where
$t_C \in [0,1]$ is the unique time such that $\{\gamma(t_C)\} = \overline{C} \cap J_k$, i.e., the point $\gamma(t_C)$ is the
end point of $C$ which lies on $J_k$. 
Therefore $\partial D' \subset (\partial \disc \cup \gamma[0,\tau]) \cup (J_k \cup C)$.
Hence
$(J_k \cup C)$ separates the set $V$ from $+1$ in $D_\tau$.
Since $V \setminus V'$ is separated from $V'$ by $C'$ in $D_\tau$, we see that
$V \setminus V'$ is a subset of the union of the bounded components of $\C \setminus (J_j \cup J_k \cup C \cup C')$.
Consequently $V \setminus V' \subset B(z_1,R/4+ 3r)$.

Now since we known that $V \setminus V' \subset B(z_1,R/4+ 3r)$, 
$\gamma[s,t] \subset V$ and $\gamma[s,t]$ is connected, we can find $[s',t'] \subset (\tau,t]$
such that $\gamma[s',t']$ is a subset of $\overline{V'}$ and it crosses $A(z_1,2r,R/2)$.
Hence we have shown that $\gamma[s,t]$ contains an unforced crossing of $A_{j,k} \dd= A(z_1,2r, R/2)$ as observed at time $\tau=\tau_{j,k}$.
Consequently if we define $E_{j,k} = E_{j,k}(r,R)$ as
\begin{equation}
E_{j,k} = \left\{ \gamma \in \xs(\disc,-1,+1) \,:\, 
  \begin{gathered}
  \gamma[\tau_{j,k},1] \text{ contains a crossing of } A_{j,k} \\
  \text{which is contained in } (A_{j,k})^u_{\tau_{j,k}}
  \end{gathered}  
  \right\}
\end{equation}
we have shown that $E(r,R) \subset \bigcup_{j=0}^\infty \, \bigcup_{k = j+1 }^\infty \, E_{j,k}$.

Let $\eps>0$ and choose $m \in \N$ such that $\P(N > m) \leq \eps/2$ for all $\P \in \Sigma_\disc$.
Now
\begin{align}
\P(E(r,R)) \leq & \P(N > m) + \P\left[ \; \bigcup_{0 \leq j<k} \{N \leq m\} \cap E_{j,k} \right] \nonumber \\
  \leq & \frac{\eps}{2} + \P\left[ \; \bigcup_{0 \leq j<k \leq m} \{N \leq m\} \cap E_{j,k} \right] \nonumber \\
  \leq & \frac{\eps}{2} + \sum_{0 \leq j<k \leq m} \P\left[  \{N \leq m\} \cap E_{j,k} \right] \nonumber \\
  \leq & \frac{\eps}{2} + K m^2 \left( \frac{r}{R} \right)^\Delta  \leq \eps
\end{align}
when $r$ is smaller than $r_0>0$ which depends on $R$ and $\eps$. Here we used the facts that
$\{N \leq m\} \cap E_{j,k}=\emptyset$ when $k > m$ and that 
$\P[  \{N \leq m\} \cap E_{j,k} ] \leq \P[ E_{j,k} ]$.
\end{proof}

\begin{proof}[Proof of Theorem~\ref{thm: hyp geodesic tip}]
In this proof, we work on the unit disc. Fix $\rho>0$ and let $\tau=\tau_\rho$ as above.
Let $D' = \disc \setminus \overline{B(1,\rho)}$. Since $\Phi$ and $\Phi^{-1}$ are uniformly continuous
on $D'$ and $\Phi(D')$, respectively, it is sufficient to prove the corresponding claim for $F_\disc$.
Furthermore it is sufficient to show that $|F_\disc(t,y) - \gamma(t)|\leq \psi(y)$ for $0<y \leq \delta$,
because $y \mapsto F_\disc(t,y)$, $y \in [\delta,1]$, is equicontinuous family by Koebe distortion theorem. 

Let $R_n>0$, $n \in \N$, be any sequence such that $R_n \searrow 0$ as $n \to \infty$.
By the previous proposition, we can choose a sequence $r_n$, $n \in \N$, 
such that $r_n < R_n$ and
\begin{equation}
\P(E(r_n,R_n)) \leq 2^{-n}
\end{equation}
for all $n \in \N$ and for all $\P \in \Sigma_\disc$.
Therefore the random variable $N \coloneq \max\{ n\in \N \,:\, \gamma \in E(r_n,R_n)\}$ is tight:
for each $\eps>0$ there exists $m \in \N$ such that
\begin{equation}\label{eq: six arms tight}
\P( N \leq m) \geq 1 - \eps  
\end{equation}
for all $\P \in \Sigma_\disc$. Fix now $\eps>0$ and let $m \in \N$ be such that \eqref{eq: six arms tight} holds.

Define $n_0(\delta)$ to be the maximal integer such that the inequality 
\begin{equation}
\frac{2\pi}{\sqrt{|\log \delta|}} \leq r_{n_0(\delta)}
\end{equation}
holds.
For given $0<\delta<1$, there is a $\delta' \in [\delta,\delta^{1/2}]$ which can depend on $t$ and $\gamma$ such that
the crosscut $C\coloneq\{ \Phi^{-1} \circ g_t^{-1}(W(t) + i \delta' e^{i \theta} ) \,:\, \theta \in (0,\pi) \}$ has
length less than $2\pi/\sqrt{|\log \delta|}$, see Proposition~2.2 in \cite{pommerenke-1992-}. 

Now if $N > n_0(\delta)$, then there must be a path from $w\coloneq\Phi^{-1} \circ g_t^{-1}(W(t) + i \delta' )$
to $\gamma(t)$ in $D_t$ that has diameter less than $R_{n_0(\delta)}$.
By Gehring-Hayman theorem (Theorem~4.20 in \cite{pommerenke-1992-}) 
the diameter of the hyperbolic geodesic $y \mapsto F_\disc( t,y )$, $0 \leq y \leq \delta'$,  
is of the same order as the smallest possible diameter of the curve which connects $w$ with $\gamma(t)$ in $D_t$. Consequently
there is a universal constant $c>0$ such that
\begin{equation}
\diam\{ F_\disc(t,y) \,:\, y \in [0,\delta]\} \leq c R_{n_0(\delta)}
\end{equation}
for all $t \in [0,\tau]$, for all $\delta>0$ such that $n_0(\delta)>m$ and for all $\gamma$ such that $N \leq m$.
\end{proof}

%
%

\subsection{Proof of the main theorem} \label{ssec: proof main}

\begin{proof}[Proof of Theorem~\ref{thm: main} {(}Main theorem{)}]
Fix $\eps>0$. We will first choose four events $E_k$, $k=1,2,3,4$, that have large probability, namely,
\begin{equation}\label{eq: pmt high probability Ek}
\P(E_k) \geq 1 - \eps/4
\end{equation}
for all $\P \in \Sigma_\disc$. Then those events have large probability occurring simultaneously
since
\begin{equation}
\P\left( \, \bigcap_{k=1}^4 E_k \, \right) \geq 1 - \eps .
\end{equation}
Once we have defined $E_k$, denote $E = \bigcap_{k=1}^4 E_k$.

We choose $E_1$ in such a way that the half-plane capacity of $\gamma[0,t]$ goes to infinity as $t \to \infty$
in a tight way on $\Sigma_\disc$. We use
Proposition~\ref{pr: contdrv} and choose $E_1$ the intersection of the events in 
the inequality~\eqref{ie: prob capacity goes to infinity} where $n = k^2$ runs from $k=m_1$ to $\infty$ where
$m_1$ is chosen so that \eqref{eq: pmt high probability Ek} holds. Then we choose $E_2$ and $E_3$ 
so that
$E_2$ is the set of simple curves  which are \emph{in some parametrization} H\"older continuous 
with a H\"older exponent $\alpha_c>0$ and a H\"older constant $K_c$ 
and 
$E_3$ is the set of simple curves which have \emph{in the capacity parametrization} H\"older continuous driving process
with a H\"older exponent $\alpha_d>0$  and
a finite H\"older norm $K_{d,T}$ for any $T>0$ when the process is restricted to the time interval $[0,T]$.
(Here $K_{d,T}$ is naturally increasing in $T$.)
Using Proposition~\ref{prop: cond implies aizbur} and Theorem~\ref{thm: holder exponent driving}
the constants are chosen so that the bound \eqref{eq: pmt high probability Ek} is satisfied. 
Finally using Theorem~\ref{thm: hyp geodesic tip} 
we set $E_4$ to be the set of simple curves that have function $\psi_\rho$ for each $\rho>0$ as in Theorem~\ref{thm: hyp geodesic tip}
and $\delta>0$ such that the geodesic to the tip is continuous with $|F(t,y) - F(t,y')| \leq \psi_\rho(|y-y'|)$ for $|y-y'| < \delta$
and $t \in [0,\tau_\rho]$.
Also here $\psi$ and $\delta>0$ are chosen so that \eqref{eq: pmt high probability Ek} holds.

Now by Lemma~\ref{lm: main lemma with convergence} of the appendix, the set $E$ is relatively
compact in the convergence in the path convergence and in the driving convergence (and
in the convergence of curves in the capacity parametrization)
and the closure of $E$ is the same in both topologies
as the following argument shows: for a sequence $\gamma_n \in E$ we can choose subsequence such
that $\gamma_n$ converges in $X$ and $W_n$ converges uniformly on compact subsets of $[0,\infty)$
and $F_n$ converges uniformly on compact subsets of $[0,\infty) \times [0,\infty)$.
Then by Lemma~\ref{lm: main lemma with convergence}, the limits agree in the sense that
if we parametrize $\lim_{n\to \infty} \gamma_n$ by the capacity forms a Loewner chain
that is driven by $\lim_{n\to \infty} \gamma_n$.

Since $E$ is precompact in the space of curves, we have shown that 
$\Sigma_\disc$ is a tight family of probability measures
on $X$, especially, and hence by Prohorov's theorem we can choose for any sequence $\P_n \in \Sigma_\disc$ a weakly convergent
subsequence. This shows the first claim.
The claims \ref{enui: main thm a}--\ref{enui: main thm e} 
of Theorem~\ref{thm: main}
follow from taking the closure of $E$ in any of the above topologies. Any subsequent weak limit $\P^*$ of $\P_n \in \Sigma_\disc$
satisfies $\P^*(\overline{E}) \geq \limsup_{n \to \infty} \P_n(\overline{E}) \geq 1 - \eps$. Hence these
claims holds $\P^*$ almost surely.

The last claim of Theorem~\ref{thm: main} on the exponential integrability of the driving process follows similarly from
the claim \ref{enui: pr contdrv 3} of Proposition~\ref{pr: contdrv}.
\end{proof}

\subsection{The proofs of the corollaries of the main theorem}\label{ssec: proof corollaries}

In this section we will prove Corollaries~\ref{cor: driving to curve} and \ref{cor: convergence in general domains}. 

\begin{proof}[Proof of Corollary~\ref{cor: driving to curve}]
If $\gamma^{(n)}$ satisfy Condition~\ref{def: b unf crossing} and its law is $\P_n$, then by (the proof of) Theorem~\ref{thm: main}
for each $\eps>0$ we can choose an event $E$ satisfying
$\inf_{n} \P_n (E) \geq 1 -\eps$ such that $E$ is relatively compact in all three topologies of the statement of 
Corollary~\ref{cor: driving to curve}. This fact follows from Lemma~\ref{lm: main lemma with convergence}
when for any sequence $\tilde{\gamma}_n \in E$ we pass to a subsequence where $\tilde{\gamma}_{n_k}$
converges in $X$, its driving term converges uniformly on compact intervals and the hyperbolic geodesic $F_n$
converges on compact sets. By Lemma~\ref{lm: main lemma with convergence}, we get the convergence in the 
capacity parametrization and, in addition, it holds that these limits agree in the sense that the limiting curve is
driven by the limiting driving term. Since $E$ is relatively compact, the sequence $\P_n$ is tight
in the same topology.

By this tightness, we see that if the sequence of random curves $\gamma^{(n)}$ or the sequence of driving
processes  $W^{(n)}$ converges in one of the three topologies, it converges also in the two other topologies.
The argument for this is essentially the same as above. We pass to a subsequence where the convergence
takes place also in the other topology. Then we notice that the sequence of the laws satisfies 
$\inf_{n} \P_n (E) \geq 1 -\eps$ hence the probability for the limiting objects to agree in the above sense
is at least $1-\eps$. Since this holds for any $\eps$, the law of the other limiting object is uniquely determined.
Therefore there is no need to pass to a subsequence, but the entire sequence converges.
\end{proof}

For the proof of Corollary~\ref{cor: convergence in general domains}
notice first that
by the proof of \ref{cond:quadexp}$\Rightarrow$\ref{def: b unf crossing} 
in Section~\ref{ssec: equiv cond} we have constants $C_1,C_2$ such that if 
$Q \subset U$ is a simply connected domain, whose boundary consists of a subset of $\partial U$ and some subsets of $U$
which are crosscuts
$S_0$ and $S_2^j$, $j = 1,2, \ldots$, (finite or infinite set), and if $Q$ has the property that it doesn't disconnect $a$ from $b$
and $S_0$ is the ``outermost'' of the crosscuts (disconnecting the others from $a$ and $b$),
then
\begin{equation}\label{ie: generalized quad crossing bound}
\P(\gamma \text{ crosses } Q) \leq C_1 \exp( - C_2 \elen(Q))
\end{equation}
where crossing means that $\gamma$ intersects one of the $S_2^j$'s and $\elen(Q)$ is the extremal length of the curve family
connecting $S_0$ to $\bigcup_j S_2^j$. Use the notation $S_0(Q)$ for the outermost crosscut and $\mathcal{S}_2(Q)$ for
the collection of $S_2^j$, $j = 1,2,\ldots$.

\begin{lemma}\label{lm: fjords}
Let $(U,a,b,\P)$ be a domain and a measure such that \eqref{ie: generalized quad crossing bound} with some 
$C_1$ and $C_2$ is satisfied for all $Q$
as above. Then for each $\eps>0$ and $R>0$ there is $\delta$ which only depends on $C_1,C_2,\eps,R$ and $\mathrm{area}(U)$
such that the following holds.
Let $Q_j$, $j \in I$ be a collection of quadrilaterals satisfying the conditions above
such that $\diam(S_0(Q_j)) < \delta$ for all $j$ and the length of the shortest path from $S_0(Q_j)$ to $\mathcal{S}_2(Q_j)$
is at least $R$. Then
\begin{equation}
\sum_{j \in I} \P( \gamma \text{ crosses } Q_j ) \leq \eps .
\end{equation}
\end{lemma}

\begin{proof}
Take any $\delta$-ball $B(z_j,\delta)$ that contains the crosscut $S_j \dd = S_0(Q_j)$.
The standard estimate of extremal length in Lemma~\ref{lm: elen annulus} gives that
\begin{equation}
\elen(Q_j) \geq \frac{\log\frac{R}{\delta}}{2 \pi} .
\end{equation}
We claim also that
\begin{equation}
\elen(Q_j) \geq \frac{(R-\delta)^2}{A_j} .
\end{equation}
To prove the second inequality fix $j$ for the time being. Let 
\begin{equation*}
\eta(r) = \{ z \in \C \,:\, |z-z_j|=r, z\in Q_j\}
\end{equation*}
Define a metric $\rho: \C \to \R_+$ by setting $\rho(z) = 1/\Lambda(\eta(r))$, 
if $z \in \eta(r)$, and $\rho(z)=0$, otherwise. Here $\Lambda$ is again the arc length.
Then for any crossing $\gamma$ of $Q_j$
\begin{align}
\mathrm{length}_\rho(\gamma) & \geq \int_\delta^R \frac{\de r}{\Lambda(\eta(r))}  \\
\mathrm{area}(\rho) & = \int_\delta^r \Lambda(\eta(r)) \frac{\de r}{\Lambda(\eta(r))^2} = \int_\delta^R \frac{\de r}{\Lambda(\eta(r))}
\end{align}
Now the claim follows from the Cauchy-Schwarz inequality
\begin{equation}
\int_\delta^R \frac{\de r}{\Lambda(\eta(r))} \, A_j \geq \int_\delta^R \frac{\de r}{\Lambda(\eta(r))} \int_\delta^R \Lambda(\eta(r))\de r
  \geq \left( \int_\delta^R \de r \right)^2 = (R - \delta)^2
\end{equation}
and the lower bound $\elen(Q_j) \geq \inf_\gamma \mathrm{length}_\rho(\gamma)^2 / \mathrm{area}(\rho)$.

Fix some $\eps>0$.
Let $I_1 \subset I$ be the set of all $j \in I$ such that $A_j \geq \delta^{\frac{C_2}{4 \pi}}$.
Then since $Q_j$ are disjoint, the number of elements in $I_1$ is at most $\mathrm{area}(U) \delta^{-\frac{C_2}{4 \pi}}$
\begin{align*}
\sum_{j \in I_1} \P( \gamma \text{ crosses } Q_j ) & \leq C_1 \sum_{j \in I_1} \exp\left( - C_2 \frac{\log\frac{R}{\delta}}{2 \pi} \right) \\
& = C_1 \, \mathrm{area}(U) \, R^{-\frac{C_2}{2 \pi}} \, \delta^{\frac{C_2}{4 \pi}} \leq \frac{\eps}{2}
\end{align*}
when $\delta$ is small, more precisely, when $0 < \delta < \delta_1$ where $\delta_1$ depends on
$C_1$, $C_2$, $\mathrm{area}(U)$, $R$ and $\eps$ only.

On the other hand, on $I \setminus I_1$, $A_j < \delta^{\frac{C_2}{4 \pi}}$ and therefore
\begin{align}
\sum_{j \in I \setminus I_1} \P( \gamma \text{ crosses } Q_j ) 
  & \leq C_1 \sum_{j \in I \setminus I_1} \exp\left( - C_2 \frac{(R-\delta)^2}{A_j} \right) 
  \leq C_1 \sum_{j \in I \setminus I_1} A_j^2 \\
  & \leq C_1 \,  \delta^{\frac{C_2}{4 \pi}} \, \sum_{j \in I \setminus I_1} A_j
  \leq C_1 \, \mathrm{area}(U) \, \delta^{\frac{C_2}{4 \pi}} \leq \frac{\eps}{2}
\end{align}
for $0 < \delta < \delta_2$ where $\delta_2 = \delta_2 (C_1,C_2,R,\mathrm{area}(U),\eps)$.
Here we used that $\exp(-\tilde{C}/x) < x^2$ when $0 < x < x_0(\tilde{C})$.
\end{proof}

Suppose now that $(U_n,a_n,b_n)$ converges in the Carath\'eodory sense to $(U,a,b)$.
We call a subset $V$ of $U_n$ a $(\delta,R)$-\emph{fjord} if it is a connected component of $U_n \setminus S$ for some crosscut $S$
of $U_n$ such that $\diam(S) \leq \delta$, $S$ disconnects $V$ from $a_n$ and $b_n$ and the set of points $z \in V$ such that
$\dist_{U_n}(z,S) \geq R$ is non-empty,  
where $\dist_{U_n}$ is the distance inside $U_n$, i.e., the length of the shortest path connecting the two sets.
The crosscut $S$ is called the \emph{mouth} of the fjord.

\begin{proof}[Proof of Corollary~\ref{cor: convergence in general domains}]
By the assumptions $U_n \subset B(0,M)$, for some $M>0$.

The precompactness of the family of measures $(\P_n)_{n \in \N}$ when restricted outside of neighborhoods of $a_n$ and $b_n$
follows from the results of Section~\ref{ssec: aizburch}.
So it is sufficient to establish that the subsequential measures are supported on the curves of $U$ (when restricted outside of the
neighborhoods of $a$ and $b$).

Fix $0 < \delta_1 < 1/2$. For $\delta>0$ small enough and for all $n$ there is a (unique) connected component of the open set
\begin{equation}
\phi_n^{-1} \left( \disc \cap \left( B(-1,\delta_1) \cup B(1,\delta_1) \right) \right) \cup \{ z \,:\, \dist(z,\partial U_n) > \delta \}
\end{equation}
which contains the corresponding neighborhoods of $a_n$ and $b_n$. Call it $\hat{U}_n^\delta$. For $R>0$
define
\begin{equation}\label{eq: cor gen dom event}
P(R,\delta,n) = \P( \exists t \in [0,1] \text{ s.t. } \dist_{U_n}(\gamma(t),\hat{U}_n^\delta) \geq 2R ) .
\end{equation}

Suppose now that the event in \eqref{eq: cor gen dom event} happens then $\gamma$ has to enter one of the
$(3 \delta,R)$-fjords in depth $R$ at least. By approximating the mouths of the fjords from outside by curves in $3\delta$-grid
(either real or imaginary part of the point on the curve belongs to $3\delta\Z$) and by exchanging some parts of curves if they intersect,
we now define a finite collection of fjords with mouths $S_j$ on the grid
which are pair-wise disjoint. And the event in \eqref{eq: cor gen dom event} implies that
$\gamma$ enters one of these fjords to depth $R$ at least. Denote the set of points in the fjord of $S_j$ that are at most at distance $R$ to
$S_j$ by $Q_j$.

Now by Lemma~\ref{lm: fjords}, for each $\eps>0$ and $R>0$, there exists $\delta_0$ which is independent of $n$ such that
for each $0 < \delta < \delta_0$, 
\begin{equation}
P(R,\delta,n) \leq \sum_j \P_n(\gamma \text{ crosses } Q_j) \leq \eps
\end{equation}
Choose sequences $\eps_m = 2^{-m}$, $R_m = 2^{-m}$ and $\delta_m \searrow 0$ such that this estimate is satisfied.
Then we see that the sum $\sum_{m=1}^\infty P(R_m,\delta_m,n)$ is uniformly convergent for all $n$. Hence by the Borel--Cantelli lemma
for any subsequent limit measure $\P^*$, the curve $\gamma$ restricted outside $\delta_1$ neighborhoods of $a$ and $b$ stay in 
the closure of
\begin{equation}
\bigcup_{\delta>0} \lim_{n \to \infty} \hat{U}_n^\delta 
  \setminus \phi_n^{-1} \left( \disc \cap \left( B(-1,\delta_1) \cup B(1,\delta_1) \right) \right)
\end{equation}
which gives the claim.
\end{proof}

%% file: random_planar_curves_sec4.tex


\section{Interfaces in statistical physics and 
Condition~\ref{def: b unf crossing}} \label{sec: checking condition} 

In this section, we prove (or in some cases survey the proof)
that the interfaces in the following models satisfy Condition~\ref{def: b unf crossing}:
\begin{itemize}
\item Fortuin--Kasteleyn model with the parameter value $q=2$, a.k.a. FK Ising, at criticality on the square lattice or
on a isoradial graph,
\item Fortuin--Kasteleyn model with a general parameter value $q \geq 1$,
this result holds \emph{conditionally} on a bound for 
the probability of a certain crossing event in a quadrilateral,
\item Ising model at criticality on the square lattice or
on a isoradial graph,
\item Site percolation at criticality on the triangular lattice,
\item Harmonic explorer on the hexagonal lattice,
\item Loop-erased random walk on the square lattice.
\end{itemize}
We also comment why Condition~\ref{def: b unf crossing} fails for uniform spanning tree.


\subsection{Fortuin--Kasteleyn model} \label{ssec: fk model}  

In Section~\ref{sssec: fk on general graph} we define the FK model, also known as random cluster model, on a general graph
and state the FKG inequality which is needed when verifying Condition~\ref{def: b unf crossing}.
Then in Sections~\ref{sssec: medial lattice}--\ref{sssec: fk model} we define carefully the model
on the square lattice. As a consequence it is possible to define the interface as a \emph{simple} curve
and the set of domains is \emph{stable} under growing the curve. Neither of these properties is absolutely necessary
but the former was a part of the standard setup that we chose to work in and the latter 
makes the verification of Condition~\ref{def: b unf crossing} slightly easier.
Finally, in Section~\ref{sssec: cond for fk ising}
we prove that Condition~\ref{def: b unf crossing} holds for the critical FK Ising model on the square lattice.

\subsubsection{FK Model on a general graph} \label{sssec: fk on general graph}

Suppose that $G=(V(G),E(G))$ is a finite graph, which is allowed to be a multigraph, that is,
more than one edge can connect a pair of vertices. 
For any $q > 0$ and $p \in (0,1)$, define a probability
measure on $\{0,1\}^{E(G)}$ by
\begin{equation} \label{eq: random cluster}
\mu_{G}^{p,q} (\omega)= \frac{1}{Z_G^{q,p}} \left( \frac{p}{1-p} \right)^{|\omega|} q^{k(\omega)}
\end{equation}
where $|\omega| = \sum_{e \in E(G)} \omega(e)$, $k(\omega)$ is the number of connected components in
the graph $(V(G),\omega)$ and $Z_G^{p,q}$ is the normalizing constant making the measure a probability measure.
This random edge configuration is called the \emph{Fortuin-Kasteleyn model} (FK) or the \emph{random cluster model}.

Suppose that there is a given set $E_W \subset E(G)$ which is called the set of \emph{wired edges}.
Write $E_W = \bigcup_{i=1}^n E_W^{(i)}$ where $(E_W^{(i)})_{i=1,2,\ldots,n}$ are the connected components of $E_W$.
Let $P$ be a partition of $\{1,2,\ldots,n\}$. In the set
\begin{equation}
\Omega_{E_W} = \{ \{0,1\}^{E(G)} \,:\, \omega(e) = 1 \textnormal{ for any }  e\in E_W \}
\end{equation}
define a function $k_P(\omega)$ to be the number of connected components in $(V(G),\omega)$ counted in a way that
for any $\pi \in P$ all the connected components $E_W^{(i)}$, $i \in \pi$, are counted to be in the same connected component.
The reader can think that for each $\pi \in P$ we add a new vertex $v_\pi$ to $V(G)$ and connect $v_\pi$ to a vertex in every
$E_W^{(i)}$, $i \in \pi$, by an edge which we then add also to $E_W$ and hence in the new graph there are
exactly $|P|$ connected components of the wired edges and each of those components contain exactly one $v_\pi$.
Call this new graph $\hat{G}$ and 
the new set of wired edges $\hat{E}_W$, which are defined once we give the triplet $(G,E_w,P)$.
Now the random-cluster measure with wired edges is defined on $\Omega_{E_W}$ to be
\begin{equation} \label{eq: random cluster wired}
\mu_{G,E_W,P}^{p,q} (\omega)= \frac{1}{Z_{G,E_W,P}^{p,q}} \left( \frac{p}{1-p} \right)^{|\omega|} q^{k_P (\omega)} 
\end{equation}
where we use the partition dependent $k_P$ which was defined above.
It is easy to check that if $\hat{G} / \hat{E}_W$ is defined to be the graph obtained when each component
of $\hat{E}_W$ is contracted to a single vertex $v_\pi$ (all the other edges going out of that set are kept and now have
$v_\pi$ as one of their ends)
then we have the identity
\begin{equation} \label{eq: random cluster wired suppressed}
\mu_{G,E_W,P}^{p,q} (\omega)=\mu_{\hat{G} / \hat{E}_W }^{p,q} (  \omega' )
\end{equation}
where $\omega'$ is the restriction of $\omega$ to $E(G) \setminus E_W$.
Therefore the more complicated measure \eqref{eq: random cluster wired} with wired edges can always be 
returned to the simpler one \eqref{eq: random cluster}. If $E_W$ is connected, then there is only one partition
and we can use the notation $\mu_{G,E_W}^{p,q}$. Sometimes we omit some of the subscripts if they are otherwise known.

A function $f: \{0,1\}^{E(G)} \to \R$ is said to be \emph{increasing} if $f(\omega) \leq f(\omega')$ whenever 
$\omega(e) \leq \omega'(e)$ for each $e \in E(G)$.
A function $f$ is \emph{decreasing} if $-f$ is increasing. An event $F \subset \{0,1\}^{E(G)}$ is increasing or decreasing
if its indicator function $\ind_F$ is increasing or decreasing, respectively.

A fundamental property of the FK models is the following inequality.
\begin{theorem}[FKG inequality] \label{thm: fkg}
Let $q \geq 1$ and $p \in (0,1)$ and let $G=(V(G),E(G))$ be a graph. 
If $f$ and $g$ are increasing functions on $\{0,1\}^{E(G)}$ then
\begin{equation}\label{ie: fkg}
\E (fg) \geq \E (f) \E (g) 
\end{equation}
where $\E$ is expected value with respect to $\mu_{G}^{p,q}$.
\end{theorem}

\begin{remark}
As explained above the measure $\mu_{G}^{p,q}$ can be replaced by any measure
conditioned to have wired edges.
\end{remark}

For the proof see Theorem~3.8 in \cite{grimmett-2006-}. 
The edges where $\omega(e)=1$ are called \emph{open} and 
the edges where $\omega(e)=0$ are called \emph{closed}. The property~\eqref{ie: fkg} is called positive association and
it means essentially that knowing that certain edges are open increases the probability for the other edges to be open. 

It is well known that the FK model with parameter $q$ is connected to the Potts model with parameter $q$. Here we are interested
in the model connected to the Ising model and hence we mainly focus to the case $q=2$ which is called \emph{FK Ising} (model).

\subsubsection{Modified medial lattice}\label{sssec: medial lattice}

Consider the planar graph $(\Z^2)_{\mathrm{even}}$ formed by the set of vertices
$\{ (i,j) \in \Z^2 : i+j \textnormal{ even}\}$ and the set of edges so that $(i,j)$ and $(k,l)$ are
connected by an edge if and only if $|i-k|=1=|j-l|$. Similarly define $(\Z^2)_{\mathrm{odd}}$ which
can be seen as a translation of $(\Z^2)_{\mathrm{even}}$ by the vector $(1,0)$, say.
Both $(\Z^2)_{\mathrm{even}}$ and $(\Z^2)_{\mathrm{odd}}$ are square lattices.
Figure~\ref{sfig: bathrooma} shows a chessboard coloring of the plane. In that figure, the vertices of
$(\Z^2)_{\mathrm{even}}$ are the centers of the blue squares, say, and 
the vertices of $(\Z^2)_{\mathrm{odd}}$ are the centers of the red squares, and 
two vertices (of the same color) are connected by an edge if the corresponding squares touch by corners. 
Note also that $(\Z^2)_{\mathrm{even}}$ and $(\Z^2)_{\mathrm{odd}}$ are the dual graphs of each other.

\begin{figure}[tb!]
\centering
\subfigure[%
The chessboard coloring holds within three square lattices:
$(\Z^2)_{\mathrm{even}}$ (blue dots and lines), $(\Z^2)_{\mathrm{odd}}$ (red dots and lines) and
the medial lattice $\hat{L}$ (black dots and lines).] 
{
	\label{sfig: bathrooma}
	\includegraphics[scale=1]
{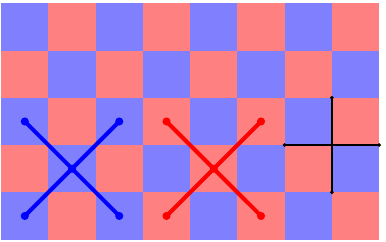}
} 
\hspace{0.5cm}
\subfigure[The modified medial lattice $L$ is formed when every vertex of $\hat{L}$
is replaced by a square. The dual lattice of $L$ is called bathroom tiling for obvious reasons.
] 
{
	\label{sfig: bathroomb}
	\includegraphics[scale=1]
{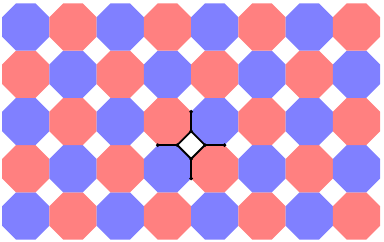}
}
\subfigure[An admissible domain: here $c_1$ and $c_2$ agree on the beginning and end and they
are otherwise avoiding each other and the domain they cut from the bathroom tiling has boundary
consisting of two monochromatic arcs.]
{
	\label{sfig: bathroomc}
	\includegraphics[scale=1]
{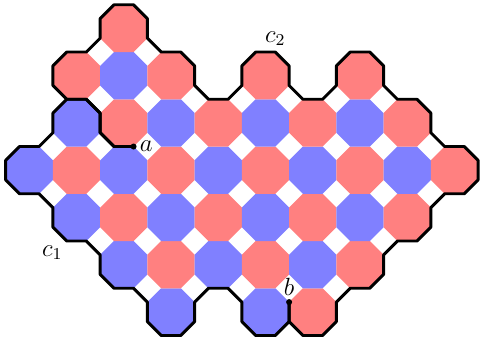}
}
\caption{Modified medial lattice and its admissible domain.} 
\label{fig: bathroom}
\end{figure}

Let $\hat{L}=(\Z + 1/2)^2$, i.e., the graph formed by the vertices and the edges of the colored squares
in the chessboard coloring. It is called
the \emph{medial lattice} of $(\Z^2)_{\mathrm{even}}$ and its dual $(\Z^2)_{\mathrm{odd}}$.
Note that vertices of $\hat{L}$ are exactly those points
where an edge of $(\Z^2)_{\mathrm{even}}$ and an edge of $(\Z^2)_{\mathrm{odd}}$ intersect.

It is useful to modify the medial lattice slightly. At each vertex of $\hat{L}$ position a white square so that
the corners are lying on the edges of $\hat{L}$. The size of the square can be chosen so that the resulting blue and red octagons
are regular. See Figure~\ref{sfig: bathroomb}. Denote the graph formed by the vertices and the edges of the octagons
by $L$ and call it \emph{modified medial lattice} of $(\Z^2)_{\mathrm{even}}$ (or $(\Z^2)_{\mathrm{odd}}$).
The dual of $L$, i.e. the blue and red octagons and the white squares
(or rather their centers), is called the \emph{bathroom tiling}.

Similarly, it is possible to define the modified medial lattice of a general planar graph $G$.
For each middle point of an edge put a vertex of $\hat{L}$. 
Go around each vertex of $G$ and connect any vertex of $\hat{L}$ to its successor by
an edge. The resulting graph is the medial graph. Notice that each vertex has degree four and hence it is possible to replace each
vertex by an quadrilateral. The result is the modified medial lattice.

\subsubsection{Admissible domains}\label{sssec: admissible}

Suppose that we are given two paths 
$(c_j(k))$, $j=1,2$, on the modified medial lattice and $k$ runs over the values $0,1,\ldots,n_j$,
that satisfy the following properties:
\begin{itemize}
\item Each $c_j$ is simple and has only blue and white faces of the bathroom tiling on its one side and red 
and white faces on the other side.
\item The first (directed) edges $(c_j(0),c_j(1))$ coincide and the edge is between a blue and a red face. 
Denote by $a$ the common starting point of $c_j$, $j=1,2$.
\item The last edges $(c_j(n_j-1),c_j(n_j))$ coincide and the edge is between a blue and a red face. 
Denote by $b$ the common ending point of $c_j$, $j=1,2$.
\item The paths $c_j$ may have arbitrarily long common beginning and end parts, but otherwise they are avoiding each other.
\item The unique connected component of $\C \setminus \bigcup \hat{c}_j$ which is bounded, has $a$ and $b$ on its boundary,
where $\hat{c}_j$ is the locus of the polygonal line corresponding to vertices $c_j(k)$, $0 \leq k \leq n_j$. 
Denote this component by $U=U(c_1,c_2)$.
\end{itemize}
Let's call a pair $(c_1,c_2)$ satisfying these properties an \emph{admissible boundary}
and $U=U(c_1,c_2)$ is called \emph{admissible domain}. Let's use a shortened notation that $U$ contains the information
how $c_1$ and $c_2$ or $a$ and $b$ are chosen.

Suppose that $(\gamma(k))_{0 \leq k \leq l}$ is a path on the modified medial lattice that starts from $a$ and possibly ends at $b$ 
but is otherwise avoiding $c_1$ and $c_2$. Suppose also that $\gamma$ has the property that 
it has only blue and white faces on one side
and only red and white faces on the other side. Call this kind of path \emph{admissible path}.
If $\gamma(k)$, $0 \leq k \leq 2m <l$ and $c_j$ are concatenated in a natural way (they have only one common point $a$)
as a curve from $\gamma(2m)$ to $b$ and this curve is
denoted as $c_{j,2m}$, then the pair $(c_{1,2m},c_{2,2m})$ is an admissible boundary. 

Later it will be useful to consider the following object.
Define \emph{generalized admissible domain with $2n$ marked boundary points} or simply \emph{$2n$-admissible domain} as
the $U(c_1,c_2,\ldots,c_{2n})$ as the bounded component of $\C \setminus (\hat{c}_1 \cup \ldots \cup \hat{c}_{2n})$
where $c_j$ are simple paths on $L$ so that $c_{2k-1}$ and $c_{2k}$ agree on the beginning and $c_{2k}$ and $c_{2k+1}$
on the end (here use cyclic order so that $c_{2n+1}=c_1$)
and otherwise as above. 
For example, we require the marked points $c_1(0), c_1(n_1), c_3(0), c_3(n_3),\ldots$
to be on the boundary of $U(c_1,c_2,\ldots,c_{2n})$. 

\subsubsection{Advantages of the definitions}\label{sssec: advantages of admissible}

Now the advantages of the above definitions are the following:
\begin{itemize}
\item It is easier to deal with simple curves on the discrete level. This is the primary motivation of 
considering the \emph{modified}
medial lattice.
\item As noted above, if we start from an admissible boundary and follow an admissible curve, then
the pair $(c_{1,2m},c_{2,2m})$ stays as an admissible boundary.
It is practical to have a \emph{stable} class of domains in that sense.
\item Let $(\pi(t)_{0 \leq t \leq l})$ be the polygonal curve corresponding to $(\gamma(k))$ so that
$\pi(k)=\gamma(k)$ and the parametrization is linear on the intervals $[k,k+1]$. Then the points of $\pi$
are bounded away from the boundary $\partial U =\hat{c}_1 \cup \hat{c}_2$ except near the end points, that is,
\begin{equation*}
\de( \pi(t), \partial U) \geq 2\eta \quad \text{when } 1 \leq t \leq l-1
\end{equation*}
and similarly the points on 
\begin{equation*}
\de( \pi(t), \pi(s)) \geq 2\eta \quad \text{when } |t-s| \geq 1
\end{equation*}
here $\eta>0$ is a constant depending on the lattice. Later, we can use this to deal with the scales smaller than $\eta$
when checking the condition.
\end{itemize}

\subsubsection{FK model on the square lattice}\label{sssec: fk model}

Let $U=U(c_1,c_2)$ be an admissible domain and assume that the octagons along $c_1$ (inside $U$ away from the common
part with $c_2$) are blue.
Denote by $G=G(c_1,c_2) \subset (\Z^2)_{\mathrm{even}}$ 
the graph formed by the centers of the blue octagons inside $U$ and
by $E_W$ the blue edges along $c_1$. $E_W$ is connected and it will be the set of wired edges.
Let $G'$ be the planar dual of $G$, that is, the graph formed by the centers of the red octagons inside $U$.
Let $G_L$  be the subgraph of $L$ formed by the vertices in $U \cup \{a,b\}$ and edges which stay inside $U$. 

For each $0<p<1$, $q>0$, define a probability measure on 
$\Omega=\Omega(c_1,c_2) = \{ \omega \in \{0,1\}^{E(G)} : \omega = 1 \textnormal{ on } E_W \}$
by
\begin{equation} \label{eq: fk measure}
\mu_U^{p,q}  = \mu_{G,E_W}^{p,q}
\end{equation}
The setup is illustrated in Figure~\ref{fig: bathroom2}.
There is a natural dual $\omega'$ of $\omega$ defined on $E(G')$ 
such that for each white square in $U$ the edge going through that square is open in $\omega'$ if and only if 
the edge of $E(G)$ going through that square is closed in $\omega$.
The duality between $\omega$ and $\omega'$ is shown in 
Figures~\ref{sfig: bathroom2a} and \ref{sfig: bathroom2b}.
Which of the two edges intersecting in a white square is open in
$\omega$ or $\omega'$ can be represented by coloring the square with that color. The picture then looks like 
Figure~\ref{sfig: bathroom2c}.
The essential information of that picture is encoded in the set of interfaces, that is, one interface starting from and ending 
to the boundary, because of the boundary conditions, and several \emph{loops}.
These interfaces are separating open cluster of $\omega$ from open cluster of $\omega'$.
Moreover, there is one-to-one correspondence between  $\omega$, $\omega'$ and the interface picture.
The random curve connecting $a$ and $b$ in the interface picture is denoted by $\gamma$ and its law by
$\P_U$ when the values of $p$ and $q$ are otherwise known.

\begin{figure}[tb!]
\centering
\subfigure[A configuration of open edges on $G$ satisfying the wired boundary condition along $c_1$.]
{
	\label{sfig: bathroom2a}
	\includegraphics[scale=.8]
{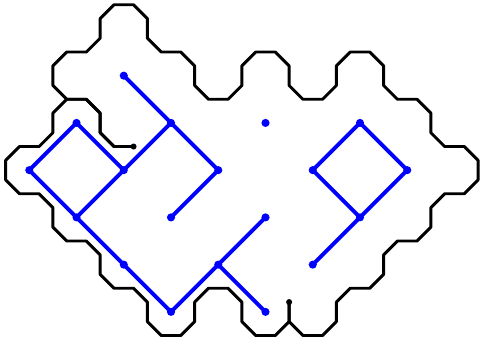}
}
\hspace{0.5cm}
\subfigure[The corresponding dual configuration of open edges on $G'$. Note that it is wired along $c_2$.]
{
	\label{sfig: bathroom2b}
	\includegraphics[scale=.8]
{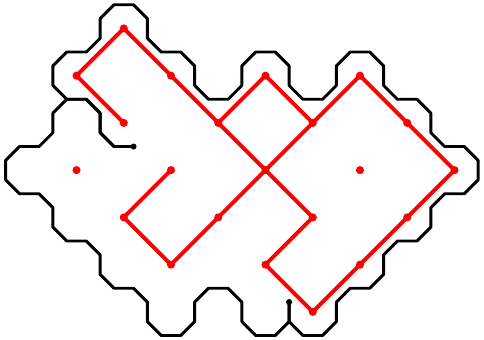}
}
\subfigure[Coloring of the squares with blue and red enables to define the collection of interfaces which
separate the blue and red regions.]
{
	\label{sfig: bathroom2c}
	\includegraphics[scale=1]
{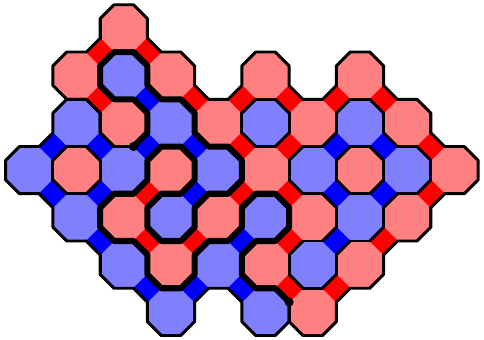}
} 
\caption{The correspondence between the configuration on $G$ (a), the configuration on $G'$ (b) 
and the interfaces and the coloring of the squares (c).} 
\label{fig: bathroom2}
\end{figure}

It is generally known that the probability measure $\mu_U^{p,q}$ can be written in the form
\begin{equation} \label{eq: fk measure var}
 \mu_U^{p,q} (\omega) = \frac{1}{Z'} 
 \left( \frac{p}{(1-p) \sqrt{q}} \right)^{|\omega|} (\sqrt{q})^{\textrm{number of loops}} .
\end{equation}
From this it follows that
\begin{equation}
\psd(q) = \frac{\sqrt{q}}{1+ \sqrt{q}}
\end{equation}
is a self-dual value of $p$. When $p = \psd$, the quantity inside the first brackets is equal to $1$, and 
it doesn't make difference whether the model was originally defined in $G$ or $G'$.
Both give the same probability for the configuration of Figure~\ref{sfig: bathroom2c}.
It turns out that the self-dual value $p = \psd$ is also the critical value at least for $q \geq 1$, 
see \cite{beffara-duminil-2012-}.

\subsubsection{Verifying Condition~\ref{def: b unf crossing} for the critical FK Ising}\label{sssec: cond for fk ising}

For each admissible domain $U$ 
(and for each choice of $a$ and $b$) define a conformal and onto map $\phi_U : U \to \disc$
such that $\phi_U(a)=-1$ and $\phi_U(b)=1$. 
In this subsection, the following result will be proven.

\begin{proposition} \label{prop: fkIsing}
Let $\P_U$ be the law of the critical FK Ising interface in $U$, i.e., $\P_U$ is the law of $\gamma$
under $\mu_U^{\psd,2}$.
Then the collection
\begin{equation} \label{eq: sigma fk}
\Sigma_{\textnormal{FK Ising}} = \left\{ \left(\phi_U, \P_U \right) \;:\;  U \textnormal{ admissible domain}
   \right\}
\end{equation}
satisfies Condition \ref{def: b unf crossing}.
\end{proposition}

\begin{remark}
In a typical application, a sequence $U_n$ of admissible domains and a sequence of positive numbers $h_n$ are chosen.
Then the family
\begin{equation}
\Sigma = \left\{ \delta_{n,*} \left( \phi_{U_n}, \P_{U_n} \right) \;:\;  n\in \N
   \right\}
\end{equation}
also satisfies Condition~\ref{def: b unf crossing}, where $\delta_{n,*}$ is the pushforward map of the scaling
$z \mapsto h_n \, z$. 
The scaling factors $h_n$ play no role
in checking Condition~\ref{def: b unf crossing}.
\end{remark}

We postpone the proof of Proposition~\ref{prop: fkIsing} until the required tools have been presented.

Consider a $4$-admissible domain $U=U(c_1,c_2,c_3,c_4)$ such that $c_1$ and $c_3$ are wired arcs.
Let the marked points be $a_j$, $j=1,2,3,4$, in counterclockwise direction and assume that
$a_1$ and $a_2$ lie on $c_1$ and $a_3$ and $a_4$ lie on $c_3$.
Then there is a unique conformal mapping $\phi_U$ from $U$ to $\half$ such that $b_j = \phi(a_j) \in \R$
satisfy
\begin{equation*}
b_1 < b_2 < b_3 < b_4, \quad
b_2 - b_1  = b_4 - b_3, \quad
 b_2  = -1 \quad \text{and}\quad
 b_3  = 1 .
\end{equation*}
A sequence of domains $U_{n}$ is said to \emph{converge in the Carath\'eodory sense}
if the mappings $\phi_{U_n}^{-1}$ converge uniformly in the compact subsets of $\half$.

An \emph{open path} is a path of open edges of $\omega$.
For any 4-admissible domain $U$ we say that $U$ is \emph{crossed by an open path} if there is an open path
which connects the wired arcs.
Denote this event by $O(U)$. 
More generally on any graph we can talk about \emph{open crossing} of a set vertices with two specified
subsets, which we call ``sides''. This means that in the configuration $\omega$ there is an open path
connecting the sides within this set. Open crossings and crossings by the interface are different but 
in some cases related events --- this fact will be used below. 

\begin{proposition} \label{prop: fkising basic crossing estimate}
Let $U_n = h_n \hat{U}_n$ be a sequence of domains such that the sequence of reals $h_n \searrow 0$ 
and $\hat{U}_n$ is a sequence of 4-admissible domains.
If the sequence $U_n$ converges to a quadrilateral $(U,a,b,c,d)$ in the Carath\'eodory sense as $n \to \infty$,
then 
$\P_n[O (\hat{U}_n)]$ converges to a value
$s \in [0,1]$. If $(U,a,b,c,d)$ is non-degenerate then $0<s<1$.
Here $\P_n$ is the probability measure $\mu_{\hat{U}_n,P}^{\psd,2}$ where $P$ is a fixed partitioning of the set $\{1,2\}$.
\end{proposition}

This proposition is proved in \cite{chelkak-smirnov-2009-} for general isoradial graphs
with an exact formula based on discrete holomorphicity. 
The following is a direct consequence of Proposition~\ref{prop: fkising basic crossing estimate}.

\begin{corollary} \label{cor: fk crossing}
If $(U,a,b,c,d)$ is non-degenerate then there are $\eps >0$ and $n_0>0$ so that
$ \eps < \P_n [O(\hat{U}_n)] < 1-\eps$ for any $n>n_0$. 
\end{corollary}

Finally before giving the proof of Proposition~\ref{prop: fkIsing}, we state the following conditional theorem
which is likely to be useful for FK models with $q \neq 2$. The proof is exactly the same as for 
Proposition~\ref{prop: fkIsing}. In fact, Condition \ref{def: b unf crossing} is verified for $1 \leq q \leq 4$
in \cite{duminil-sidoravicius-tassion-2014-} based on this type of estimates.

\begin{proposition} \label{prop: fk model qneq2}
Let $\P_U$ be the law of the critical FK model interface in $U$, i.e., $\P_U$ is the law of $\gamma$
under $\mu_U^{\psd,q}$ for $q \geq 1$.
If the statement of Corollary~\ref{cor: fk crossing} holds for the critical FK model with the parameter $q$,
then the collection
\begin{equation} 
\Sigma_{\textnormal{FK(q)}} = \left\{ \left(\phi_U, \P_U \right) \;:\;  U \textnormal{ admissible domain}
   \right\}
\end{equation}
satisfies Condition \ref{def: b unf crossing}.
\end{proposition}

We establish below one of the geometric conditions and not directly one of the conformally invariant conditions.
The reason for this is that we want to apply Corollary~\ref{cor: fk crossing} only a finite number of times.
To verify the conformally invariant condition directly, we would have to use some compactness property for the family
of quadrilaterals. This might be more or less equivalent to the proof below.

\begin{proof}[Proof of Proposition~\ref{prop: fkIsing}]
Let's use the continuous time parametrization of $\gamma$ with constant speed so that $\gamma(n)$ is a
lattice point of $L$ if and only if $n$ is an integer (between $0$ and $l$).
Notice that on even time instances $\gamma(2n)$ is in a crossing ``arriving'' to a white square and
it chooses left or right turn depending on the color revealed on the square, in the sense of Figure~\ref{sfig: bathroom2c}.
The filtration generated $\F_t$ by $\gamma(s)$, $s \leq t$, (or the finer filtration made by adding an
infinitesimal peek to the future) remains constant on $t \in (2n,2n+2)$ where $n$ is an integer.
Hence we can restrict to the case $t=2n$, $n$ an integer. Then $U_t  = U \setminus \gamma[0,t]$
is an admissible domain. 

This simplifies the proof a lot. Instead of considering all $t=2n$ we will consider $t=0$
and all admissible domains. In other word, we don't have to consider any stopping times below, but instead
we have to consider all possible admissible domains. But luckily this was the set of domains we used
in the definition of $\Sigma_{\textnormal{FK Ising}}$.

We can also assume that $r > \eta$ where $\eta$ is as in the section~\ref{sssec: advantages of admissible}. 
In the complementary case, we notice that no disc of the form $B=B(z_0,r)$, where $0< r \leq \eta$ 
and which intersects the boundary of the domain, can contain any lattice points of $L$ which are in the interior of the domain.
Also choose $C>0$ such that there are no trivialities such as an edge crossing $A(z_0,r,R)$ for some $z_0$ and
$r>\eta$ and $R>Cr$.

Let $U$ be an admissible domain and 
$G(U) \subset (\Z^2)_{\textnormal{even}},G'(U)\subset (\Z^2)_{\textnormal{odd}},G_L(U)\subset L$
the corresponding graphs, and let $c_1$ and $c_2$ be the two marked boundary arcs as in Section~\ref{sssec: admissible}.
Let $A=A(z_0,r,R)$ be an annulus such that $r>\eta$. Write $\mu_1=\mu_U^{\psd,2}$.
Write the set $A^u$, which is defined as in \eqref{eq: definition Au},
as a disjoint union $A^u = U_1 \cup U_2$ where 
$U_k$ is next to $c_k$, that is, if we approach the boundary of the domain by a sequence in $U_k$, we hit $c_k$ 
(say, in the conformal
sense after mapping to a reference domain with simple boundary). 
See also Figure~\ref{fig: path p}.

\begin{figure}[tb!]
\centering
	\includegraphics[scale=1]
{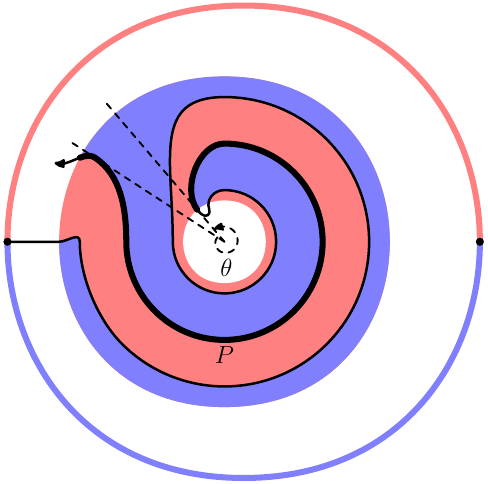}
\caption{The boundary of the domain $U$ is the black solid line and
the boundary arcs $c_1$ and $c_2$ have the solid blue and red lines, respectively, next to them. 
The parts $U_1$ and $U_2$ of $A^u$ are colored with slanted lines.
They are the regions colored with blue and red, respectively.
The boundary of $U_1$ is wired and the boundary of $U_2$ is dual wired.
As one boundary arc $P$ is fixed the components are in the $2\pi$ sector starting from that curve.
The angle $\theta$ is the difference of the maximum and minimum value of the angle variable 
defined continuously on $P$.} 
\label{fig: path p}
\end{figure}

Since $\gamma$ is the interface which separates the cluster of open edges connected to $c_1$ from the cluster 
of dual open edges connected to $c_2$
\begin{align*}
&\P_U ( \gamma \textnormal{ crosses } A^u)\\
& \leq \mu_1 ( \textnormal{open crossing of }          U_1
              \textnormal{ or dual open crossing of } U_2 ) \\
& \leq \mu_1 ( \textnormal{open crossing of }      U_1 ) +
       \mu_1 ( \textnormal{dual open crossing of } U_2 ) \\
& \leq 2 K
\end{align*}
where $K$ is the maximum of the two terms on the preceding line. Therefore we have to prove that $K< 1/4$.
By symmetry, it is enough to prove that 
\begin{equation} \label{ie: fk crossing 1}
\mu_1 ( \textnormal{open crossing of } U_1 ) < 1/4 .
\end{equation}

Let's define two discrete versions of the annulus $A$ on the modified medial lattice. 
The subset $A_{\text{blue}} \subset A $ is a doubly connected domain in $\C$.
We require that the boundary of $A_{\text{blue}} $ is a path in the modified medial lattice and that the faces
of the modified medial lattice inside $A_{\text{blue}} $ next to its boundary are blue or white.
We also require that $A_{\text{blue}} $ is maximal such domain with respect to taking unions.
Similarly, let
$A_{\text{red}}$ be the maximal subdomain of the annulus $A$ that has red and white boundary on the modified medial lattice. 
In other words, $A_{\text{blue}} $ and $A_{\text{red}}$ are discrete approximations of $A$, with correct type of boundary.
Let $V_-$ and $V_+$ be the connected components of the boundary vertices on $A_{\text{blue}}$ and denote by $V_- \leftrightarrow V_+$
the event that there is an open path between $V_-$ and $V_+$ in the given graph.
Let 
$G_2 \subset G$ be the subgraph corresponding to the domain $U_1 \cap A_{\text{blue}}$. 
Let $E_2 \subset E(G_2)$ be the set of blue edges along the boundary. Let $\mu_2$ be the random cluster
measure on $G_2$ such that the edges in $E_2$ are wired and all the components of $E_2$ are counted to be separate. 
Then by considering $f = \ind_{E_2 \subset \omega}$ in the FKG inequality we have that
\begin{equation*}
\mu_1 ( \textnormal{open crossing of } U_1 )
  \leq \mu_2 \left( V_- \leftrightarrow V_+ \right) .
\end{equation*}
Similarly, it is enough to prove there is a constant $s<1$ such that
\begin{equation} \label{ie: fk crossing 2}
\mu_2 \left( V_- \leftrightarrow V_+ \right)  
\leq s
\end{equation}
for a fixed ratio $R/r$ since using this in several concentric annuli 
we get \eqref{ie: fk crossing 1} for a larger annulus. 
Yet another similar argument shows that we can consider only annuli $A(z_0,r,R)$ where $r>C' \eta$ for any fixed $C' \geq 1$.
Namely, if \eqref{ie: fk crossing 1} holds for $r > C' \eta$ then for $\eta < r \leq C' \eta$
we can ignore the part below the scale $C' \eta$ and only consider crossing between $R$ and $C'  \eta$
and then we notice that $R \geq C r > (C/C') \cdot (C' \eta)$ and therefore by modifying the value $C$
we get \eqref{ie: fk crossing 1} for the whole range of $r$. 
Therefore we will prove \eqref{ie: fk crossing 2} when $r > C' \eta$ when $C'$ is suitably chosen and $R / r$ fixed.

\begin{figure}[t!]
\centering
\subfigure[The domain $U_1$.]
{
	\label{sfig: fk crossings le 4pi a}
	\includegraphics[scale=.6]
{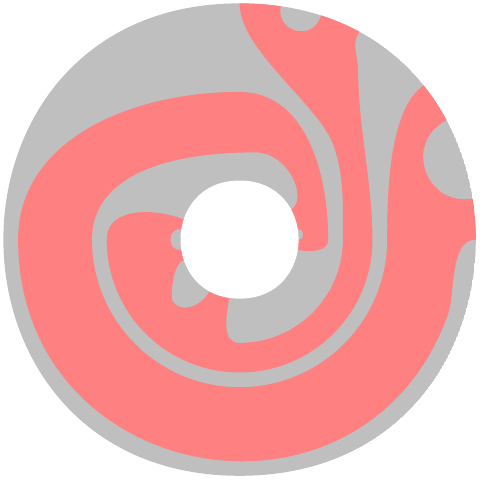}
} 
\hspace{0.25cm}
\subfigure[In the measure $\mu_2$ the edges along the boundaries of the annulus are wired.]
{
	\label{sfig: fk crossings le 4pi b}
	\includegraphics[scale=.6]
{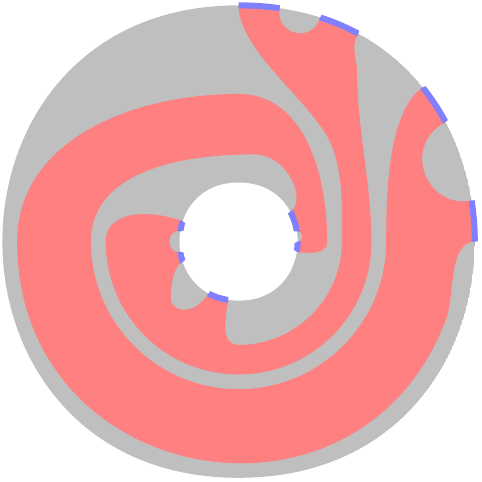}
}
\subfigure[In the case $\theta \leq 4\pi$, the final domain to be considered is a $6\pi$ opening
in the annulus with wired edges along the boundaries if the annulus. The blue color indicates the wired edges
and the red color the dual wired edges. The crossing is between the two blue boundary components.]
{
	\label{sfig: fk crossings le 4pi c}
	\includegraphics[scale=.6]
{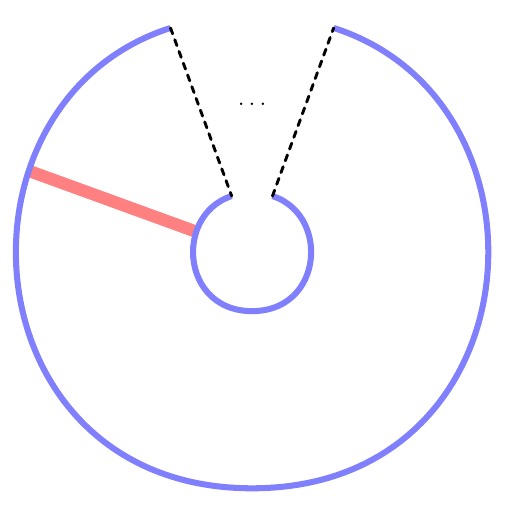}
}
\hspace{0.25cm}
\subfigure[A schematic illustration how to cover the annulus with infinitely many layers of the lattice.]
{
	\label{sfig: fk crossings le 4pi d}
	\includegraphics[scale=.9]
{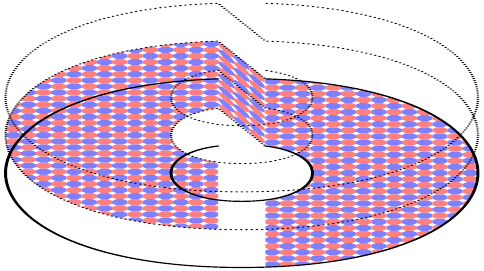}
}
\caption{Illustration how the FKG inequality is applied here in general and especially in the case $\theta \leq 4\pi$.} 
\label{fig: fk crossings le 4pi}
\end{figure}

Let $P$ be one of the boundary arcs of $U_1$ which cross $A$.
Write the points $z \in P$ in polar coordinates $z=z_0 + \rho e^{i \xi}$ so that $\xi$ is continuous along $P$.
Denote by $\theta$ the difference between the maximum and the minimum value of $\xi$ along $P$ and by
$\alpha$ the minimum value of $\xi$. The value of $\alpha$ is determined only up to additive multiple of $2\pi$
but $\theta$ is unique. Now $\xi$ spans the interval $[\alpha,\alpha+\theta]$ along $P$. 
The rest of the proof is divided into two cases: $\theta \leq 4\pi$ and $\theta > 4\pi$.

\emph{Case $\theta \leq 4\pi$:}
Consider the right half-plane $\half_1 = \{ (\rho,\xi) : \rho>0 , \xi \in \R \}$ as an infinite covering surface of 
$\C \setminus \{z_0\}$ such that
$(\rho,\xi) \in \half_1$ gets projected on $z_0 + \rho \, e^{i\xi} \in \C \setminus \{z_0\}$. 
Lift the lattice $L$ to $\half_1$ using this mapping locally in neighborhoods where it is a bijection and
define $S_{\text{blue}}$ as the lift of $A_{\text{blue}}$, that is, 
as the maximal subdomain of 
$S=S(r,R)=\{ (\rho,\xi) : 0<\rho<R , \xi \in \R \}$ 
such that the boundary is on the medial lattice and it is a blue boundary. 
Let $G_3$ be the subgraph of the lifted $(\Z^2)_{\text{even}}$ corresponding to the domain 
$S_{\text{blue}} \cap (r,R)\times(\alpha,\alpha+6\pi)$
and denote the edges along the vertical boundary as $E_3$. Let
$\mu_3$ be the random-cluster measure on $G_3$ where $E_3$ is wired and the components of $E_3$ are counted to be separate. 
Now $G_2$ can be seen as a subgraph of $G_3$. If the wired edges of the dual of $G_2$ are denoted by $E_2'$,
then applying the FKG inequality for the decreasing event $\{\omega \,:\, E_2' \subset \omega' \}$ and for the measure $\mu_3$
shows that
\begin{equation*}
\mu_2 \left( V_- \leftrightarrow V_+ \right) 
\leq \mu_3 \left( V_- \leftrightarrow V_+ 
\right) .
\end{equation*}
At this point we have reduced the problem to the event of an open crossing of fixed shape 
(independent of the domain we started with), but with variable size. 
Figure~\ref{sfig: fk crossings le 4pi c} illustrates how this shape looks like.
Since this domain is a 4-admissible domain, we can use
Proposition~\ref{prop: fkising basic crossing estimate}
to show that there are constants $C_1 \geq 1$ and 
$s_1 < 1$ such that the right-hand side of the previous
inequality is less than $s_1$ uniformly for any $r \geq C_1 \eta$ and $R = 3 r$.

The previous statement follows if we are able to show that probability of an open crossing remains bounded away
from one when we consider larger and larger $r$.
This follows claim from the following argument.
Make a counter assumption that there exists a sequence $\alpha_n \in [0,2\pi]$ and $r_n \to \infty$ such that
if we set $m_n = \mu_3 ( V_- \leftrightarrow V_+) $ in the domain $U_{3,n}$ with $r=r_n$, $R=3 r_n$ and $\alpha=\alpha_n$
and $\theta = 6 \pi$, then $m_n \to 1$ as $n \to \infty$. By choosing a subsequence we can assume that $\alpha_n$ converges.
Now the sequence of domains $r_n^{-1} U_{3,n}$ converges in the same sense 
as in Proposition~\ref{prop: fkising basic crossing estimate}
and hence $\lim_{n \to \infty} m_n < 1$.

\begin{figure}[th!]
\centering
\subfigure[$\mu_2$]
{
	\label{sfig: fk crossings g 4pi a}
	\includegraphics[scale=.7]
{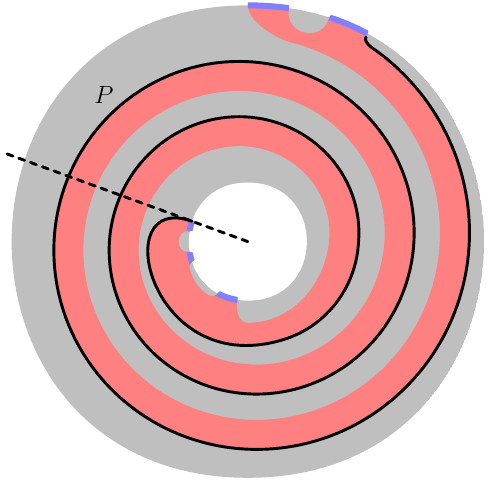}
} 
\hspace{0.25cm}
\subfigure[$\mu_4$]
{
	\label{sfig: fk crossings g 4pi b}
	\includegraphics[scale=.7]
{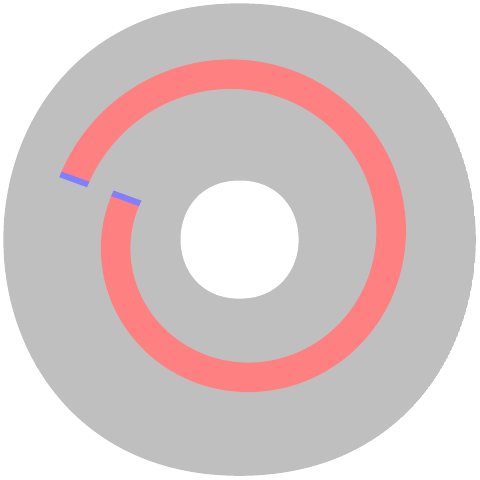}
}
\subfigure[$\mu_5$]
{
	\label{sfig: fk crossings g 4pi c}
	\includegraphics[scale=.7]
{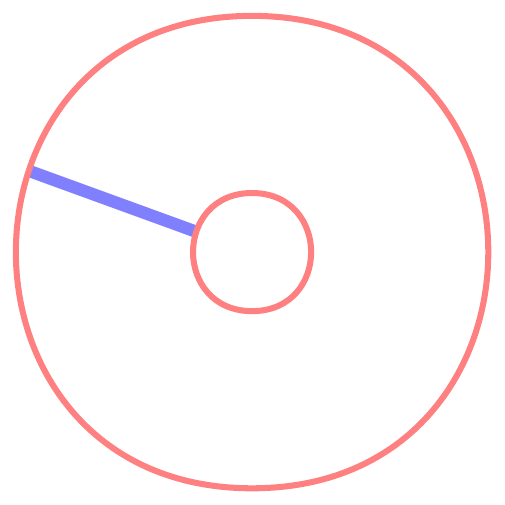}
}
\caption{Illustration how the FKG inequality is applied in the case $\theta > 4\pi$.} 
\label{fig: fk crossings g 4pi}
\end{figure}

\emph{Case $\theta > 4\pi$:} 
Similarly as in the other case
define $S_{\text{red}}$ to be the lift of $A_{\text{red}}$ to $S$.
Now note that any component of $G_2$ (view as lifted to $S$) intersects the radials $\alpha + 2 \pi$ and $\alpha + 4 \pi$
and any open crossing has to intersect those radial. Hence in the same way as above we can add blue
boundary and blue wired edges to those radials and ignore the part outside $(r,R)\times(\alpha+2\pi,\alpha+4\pi)$.
Denote the resulting graph by $G_4$ and the measure by $\mu_4$ and denote the vertices of the lifted $(\Z^2)_{\text{even}}$
along those two radials by $V_{2\pi}$ and $V_{4\pi}$. Denote the dual wired edges
of $\mu_4$ by $E_4'$. Finally if $G_5$ is the graph
corresponding to the domain $S_{\text{red}} \cap (r,R)\times(\alpha+2\pi,\alpha+4\pi)$
and $E_5$ are the boundary edges along the radials, then let $\mu_5$ be the random-cluster measure
on $G_5$ with wired edges $E_5$. In the same way as above, we can apply the FKG inequality for $\mu_5$ and
for the decreasing event $\{ \omega \,:\, E_4' \subset \omega' \}$ to get the second inequality in 
\begin{equation*}
\mu_2 \left( V_- \leftrightarrow V_+ \right) 
\leq \mu_4 \left( V_{2\pi} \leftrightarrow V_{4\pi} \right) 
\leq \mu_5 \left( V_{2\pi} \leftrightarrow V_{4\pi} \right) .
\end{equation*}
Again at this point we have reduced the problem to the event of an open crossing of fixed shape 
which is now illustrated in 
Figure~\ref{sfig: fk crossings le 4pi c}.
Use again 
Proposition~\ref{prop: fkising basic crossing estimate}
to show that there are constants $C_2 \geq 1$ and $s_2 < 1$ 
such that the right-hand side of the previous
inequality is less than $s_2$ uniformly for any $r \geq C_2 \eta$ and $R = 3 r$. 

When we combine these separate cases, 
the claim follows for $s = \max\{s_1,s_2\}$ and $C' = \max\{ C_1, C_2 \}$.
\end{proof}


\subsection{%
  Spin Ising model}\label{ssec: spin Ising}

Consider the spin Ising model at criticality on the square lattice (or
any isoradial graph)
with Dobrushin boundary conditions.

In \cite{chelkak-smirnov-2009-} a discrete holomorphic observable
$f_\delta(z)=f_\delta^{U_\delta,a_\delta,b_\delta}(z)$ is constructed
with a martingale property.
It is shown in Theorem~5.6 that, as the mesh $\delta$ of the lattice
tends to zero, the discrete domains $U_\delta$
with marked points $a_\delta, b_\delta$ tend to a continuum domain $(U,a,b)$,
the observable converges to its continuous counterpart $f(z)=f^{U,a,b}(z)$.
The latter is given by a solution of the Riemann-Hilbert boundary value problem,
and can be written as $f=\sqrt{\Psi'}$, where $\Psi$ is the conformal
map of $U$ to the upper half-plane
with $\Psi(a)=\infty$, $\Psi(b)=0$, appropriately normalized at $b$.
The convergence is uniform inside the domain $U$, and on straight
pieces of the boundary common to $U$ and $U_\delta$.

Consider the interface (domain wall) joining the points $a_\delta$ and
$b_\delta$ inside $U_\delta$.
The results of  \cite{chelkak-smirnov-2009-} immediately imply
convergence of the interface to SLE$_3$ in the sense of
the Loewner driving functions convergence.
In \cite{chelkak-duminil-hongler-kemppainen-smirnov-2013-} it was shown how
to use the results of the present article together with the crossing estimates of the article
\cite{chelkak-duminil-copin-hongler-2013-}
to deduce the strong interface convergence (see Theorem~\ref{thm: cd-chks} of the present article), by verifying the
conditions of present article.

Below we sketch an alternative way to check Condition~\ref{def: b unf crossing},
using only the observable results of
\cite{chelkak-smirnov-2009-}.

\subsubsection{Fermionic observable of spin Ising model}\label{sssec: observable spin Ising}

The spin Ising model is defined on any finite graph $G=(V,E)$. 
The random configuration $\usigma$ takes values in $\{-1,+1\}^V$
and its distribution
is given by
\begin{equation}
\P_\beta[\usigma=\underline{s}] = \frac{1}{Z} \exp \left( \beta \sum_{(v_1,v_2) \in E} s_{v_1}s_{v_2} \right).
\end{equation}
for any $\underline{s} \in \{-1,+1\}^V$. Here the partition function $Z$ is the constant normalizing the measure to be a probability measure.
The quantity $s_v$ is called the spin at $v$.
The parameter $\beta>0$ is interpreted as the inverse temperature. 

Consider the critical Ising model, $\beta = \beta_c$, on a finite, connected subgraph of the square lattice with mesh $\delta>0$. Then
\begin{equation}
\P_{\beta_c}[\usigma=\underline{s}] = \frac{1}{\tilde{Z}} x_c^{n(\underline{s})}
\end{equation}
where $n(\underline{s})=\# \{ (v_1,v_2) \in E \,:\, s_{v_1} \neq s_{v_2} \}$ and  $x_c = \sqrt{2} -1$.
If we fix the spin at one vertex, 
there is a one-to-one correspondence between spin configurations and even subgraphs of the dual graph of $G$,
given by the interfaces (domain walls) separating $+1$ and $-1$ spins.

Let $\mathcal{S}$ be the collection of even subgraphs of $G^*$,
and $\mathcal{S}_{a,b}$ be the collection of subgraphs which are even everywhere except at $a$ and $b$, where they are odd.
Any element in $S\in\mathcal{S}_{a,b}$ can be written as a pair $S=(\gamma,\Gamma)$ such that
$\gamma$ and $\Gamma$ are edge-disjoint,
$\gamma$ is a non-self-intersecting path from $a$ to $b$ and $\Gamma$ is an even graph.
The representation is not unique, but we will fix it uniquely by taking $\gamma$ to be the \emph{left-most} such path.

For $(\gamma,\Gamma) \in \mathcal{S}_{a,z}$ denote by $W(z,\gamma)$ the winding of $\gamma$ from $a$ to $z$.
Define an observable
\begin{equation}
f_\delta(z) = f_\delta^{U,a,b}(z) 
  =\nu \frac{\sum_{S= (\gamma,\Gamma) \in \mathcal{S}_{a,z}} x_c^{\# S} e^{-\frac{i}{2} W(z,\gamma)} }{
   \sum_{S=  (\gamma,\Gamma) \in \mathcal{S}_{a,b}} x_c^{\# S} e^{-\frac{i}{2} W(b,\gamma)} }
  =\nu \frac{\sum_{S=  (\gamma,\Gamma) \in \mathcal{S}_{a,z}} \; x_c^{\# S} e^{-\frac{i}{2} W(z,\gamma)} }{
   Z_{a,b} \, e^{-\frac{i}{2} W(b)} }   
\end{equation}
where in the last equation we use that when $a$ and $b$ are boundary points, the winding to $b$ doesn't depend on $\gamma$.
Here $Z_{a,b}$ is the partition function $\sum_{S\in \mathcal{S}_{a,b}} x_c^{\# S}$, 
the quantity $\# S$ is the number of edges in $S$ and $\nu$ is a constant.
The constant $\nu$ depends on local shape of the domain $U$ near $b$, but
it takes a fixed value over the class of subdomains of $U$ that we will consider.

\begin{theorem}[Chelkak--Smirnov \cite{chelkak-smirnov-2009-}]\label{thm: cs2009}
Let $U$ be a simply connected domain and let $b$ be a boundary point of $U$.
Assume that the boundary of $U$ is straight near $b$ and $U$ contains a rectangular neighborhood
$R$ of $b$. 
Then
for an appropriate choice of the constant $\nu$, $f$ has both of the following properties:
\begin{enumerate}\enustyii
\item\label{enui: cs thm i} For any $z$, $f_\delta(z)$ is a martingale with respect to the growing interface. 
\item\label{enui: cs thm ii} As $\delta \to 0$,
$f_\delta$ converges to $\lambda \sqrt{\Psi'}$, uniformly on any compact subset of $U$
and uniformly in any straight part of $\partial U$, where
$\Psi$ is a conformal map from $U$ to $\half$ with $\Psi(b) \in \R$ and $|\Psi'(b)|=1$
and $\lambda$ is fixed constant with unit modulus.
\end{enumerate}
Moreover once $b$ and $R$ are fixed, the convergence of $f_\delta^{U,a,b}$ 
in \ref{enui: cs thm ii} is uniform over all domains $U$ as well as
points $a$ and $z$, as long as $a$ is at a finite distance from $R$ and $z$ is inside $R$.
\end{theorem}

\begin{remark}
Although $\Psi$ is only unique up to an additive constant, $\Psi'$ is uniquely determined.
The branch of the square root and the constant $\lambda$ are chosen so that  
$\lambda\sqrt{\Psi'}$ is positive at $b$. Denote by $f=f^{U,a,b}$ the function
\begin{equation}\label{eq: def f cont spin ising}
f^{U,a,b} = \lambda \sqrt{\Psi'} .
\end{equation}
\end{remark}

\subsubsection{Using monotonicity and the martingale observable}

Consider a triplet $(U,a,b)$ and an annulus $A=A(z_0,r,R)$. 
We aim to verify Condition~\ref{def: b unf crossing} for spin Ising model on the domain $(U,a,b)$ with respect to the annulus $A$.
To that effect we consider the random curve $\gamma$ which is the interface
separating the macroscopic $+1$ and $-1$ clusters in the domain $U$ with Dobrushin boundary conditions
which change at $a$ and $b$.

Remember that $A^u$ is defined as in \eqref{eq: definition Au}.
Let $V_k$, $k=1,2,\ldots,n$, be the connected components of $A^u$ which can be crossed by the curve
without first crossing some other connected component of $A^u$. We
can assume that all $V_k$ have boundary conditions $-1$ and that any crossing of $V_k$ has to first go 
from the outer circle to the inner circle of $\partial A$. 
The $+1$ boundary components or the components that go from inside to outside could be
dealt with in identical manner.

Next we observe in the same way as in the case of FK Ising model that 
one of the two cases occurs:
either all $V_k$ can be lifted simultaneously to the universal cover 
\begin{equation}\label{eq: cover of A}
F=\{ (\rho,\theta) \in [r,R] \times \R\}
\end{equation}
of $A$ so that they are in a sector of $6 \pi$ opening in $F$, or that each of $V_k$
crosses a $2\pi$ sector in the angular direction.  
Basically, this division is possible since when we fix an arc of $\partial U$ that crosses $A$,
its winding around $z_0$ is either less than $4 \pi$ or greater than $4 \pi$. In the first
case, when we take the radial line through the point of the arc with smallest angle, then all $V_k$
lie in the $6 \pi$ sector from it. In the second case, we can similarly find a $2 \pi$ sector that
all $V_k$ cross.
We will deal with the first case here explicitly. The other case is similar.

\begin{figure}[th!]
\centering
\subfigure[The domain in ``logarithmic coordinates'', that is, in order to get back to the original domain we need to apply the covering map 
$w \mapsto z_0 + e^w$. 
The annulus is the vertical strip drawn with dots. We wish to give an upper bound to 
the probability of the event of  a connected path of $+1$ spins crossing the annulus 
(or more accurately any of its components with purely $-1$ boundary), indicated by the dashed red arrow.]
{
	\label{sfig: spin ising monot a}
	\hspace{0.25cm}\includegraphics[scale=.7]
{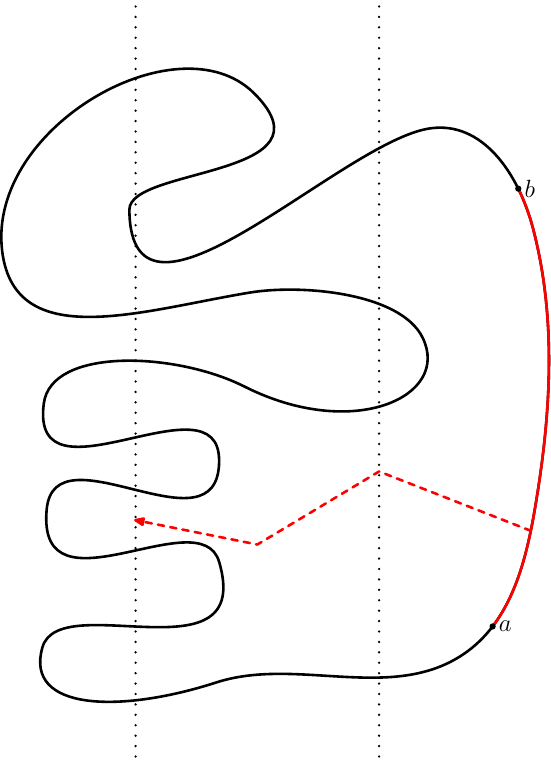}\hspace{0.25cm}
} 
\hspace{0.75cm}
\subfigure[The domain after the transformation. The horizontal dotted line is the middle radial line of the sector. By monotonicity,
the upper bound is given as the probability of having a connected path of $+1$ spins crossing the lower half of the sector.
The fact that the crossing is bounded away from $b$ is useful for technical purposes when using Theorem~\ref{thm: cs2009}.]
{
	\label{sfig: spin ising monot b}
	\hspace{0.5cm}\includegraphics[scale=.7]
{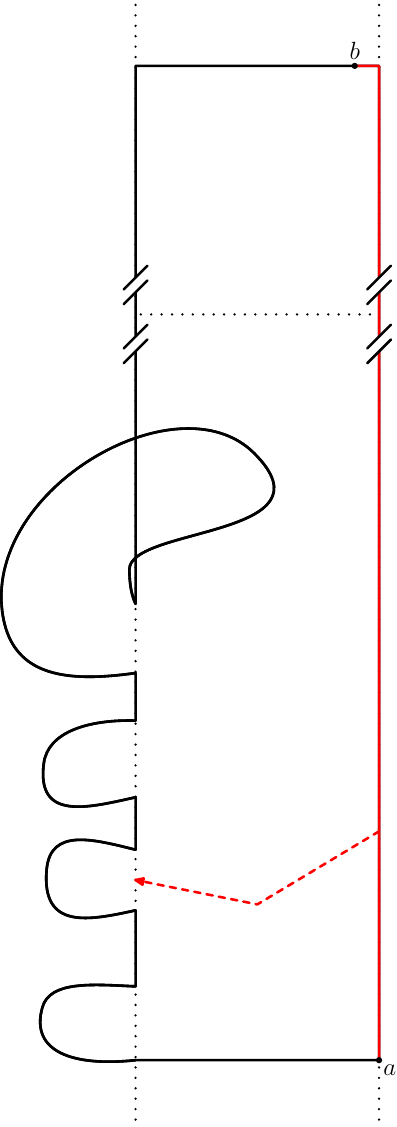}\hspace{0.75cm}
}
\caption{We apply a transformation to the domain which is consistent with the monotonicity of the Ising model.
The transformed domain is a simply connected subdomain of an appropriately chosen Riemann surface.
Thus any of the hanging parts to the left of the annulus in \protect\subref{sfig: spin ising monot b}
which seem to overlap with the part of the domain in the annulus should be considered to be on a different sheet of the cover
than the annular part. Notice that in \protect\subref{sfig: spin ising monot b} we require that the crossing is in the lower half.} 
\label{fig: spin ising monot}
\end{figure}

We apply the transformation illustrated in Figure~\ref{fig: spin ising monot} to the domain $(U,a,b)$.
We'll give the details in next paragraphs.

Let $I_j$, $j=1,2,\ldots,m$, be the boundary arcs of $V_k$ that lie on the circle of radius~$r$
and centered at $z_0$, that is, on the inner boundary of $A$. And for each $j=1,2,\ldots,m$,
let $U_j$ be the connected component of $U \setminus I_j$ which is disconnected from $a$ and $b$
in $U$ by $I_j$.

Suppose that the interface in $(U,a,b)$ makes an unforced crossing of $A$.
Then in particular, there is a crossing of $+1$ spins from the $+1$ boundary arc to
one of the components $U_j$. By monotonicity of the Ising model, the probability
of such an event increases if we pull the ``$+1$ boundary'' closer and push the ``$-1$ boundary'' away.

Consider the $12 \pi$ sector on the universal cover to which we can lift all $V_k$
so that they are lifted to the ``lower'' half (opening of $6\pi$) of the $12 \pi$ sector.
Below we will always consider crossings that stay in the ``lower'' half.
Let $\conti{U}_0 \subset F$ be the $12 \pi$ sector and let
$\discr{U}_0 \subset F$ be its discrete approximation, say, let $\discr{U}_0 \subset F$ be the union of all the faces
of the lifted square lattice that are contained in $\conti{U}_0$.
Denote by $\conti{U}$ and $\discr{U}$ the domain and its discrete approximation which we get by
gluing each $U_j$ along the lifted arcs $I_j$ to $\conti{U}_0$ and $\discr{U}_0$, respectively,
on appropriate Riemann surfaces so that $U_j$ remain disjoint from $\conti{U}_0$ and $\discr{U}_0$.
See Figure~\ref{fig: spin ising monot} for illustration.

Let $\discr{a}$ be one of the boundary points of the dual lattice near one of the corners of $\discr{U}_0$
corresponding to $\rho=R$ and let $\discr{b}$ be close to the other corner, but next to a point with $\rho=9R/10$.
Here $\rho$ refers to the coordinates \eqref{eq: cover of A}.
Suppose that the radial angle of $\discr{a}$ is smaller than the radial angle of $\discr{b}$.
If the boundary condition change at $\discr{a}$ and $\discr{b}$ from $-1$ to $+1$ and back in this new setup, 
then the probability of the $+1$ crossing from $\discr{a} \discr{b}$ to any of $U_j$,
and which stays in ``lower half'',
gives an upper bound to the probability of an unforced crossing of $A$ in $(U,a,b)$ by monotonicity 
of the Ising model (FKG inequality).

This means that the interface in the new setup will make a crossing staying in the 
``lower half'' of $\discr{U}_0$ to some $U_j$. 
The probability of this can be estimated using
the martingale property of $f$.

Let $\conti{a}$ and $\conti{b}$ be the boundary points of $\conti{U}$ that correspond to $\discr{a}$ and $\discr{b}$.

Let $\discr{c}$ be on the same radial boundary segment of $\discr{U}$ as $\discr{b}$
and let $\conti{c}$ be the corresponding boundary point of $\conti{U}$.
Since $\lambda\sqrt{\Psi'(\conti{b})}$ is positive and the boundary near $b$ is straight, 
also $\lambda\sqrt{\Psi'(\conti{c})}$ is positive
as well as any $f_\delta^{\discr{U}_\tau,\discr{a}_\tau,\discr{b}}(\discr{c})$ which we consider below.
We define $\discr{a}_\tau$ to be the tip of $\gamma$ at the random time $\tau$ and
$\discr{U}_\tau$ to be $\discr{U} \setminus \gamma[0,\tau]$.
The domain $\conti{U}_\tau$ is $\conti{U} \setminus ([\conti{a},\discr{a}] \cup \gamma[0,\tau])$
where $[\conti{a},\discr{a}]$ is the line segment from $\conti{a}$ to $\discr{a}$ in the plane.
Let $\tau$ be the hitting time of $\bigcup_{j > 0} U_j$ by $\gamma$ and
$\discrui{U}{-}_0$ be the ``lower half'' of $\discr{U}_0$.
By the martingale property
\begin{equation}\label{ie: martingale prop spin Ising observable}
 f_\delta^{\discr{U},\discr{a},\discr{b}}(\discr{c}) \geq 
  \E\left[ \ind_{\tau < \infty, \; \gamma[0,\tau] \subset \discrui{U}{-}_0}
   \; f_\delta^{\discr{U}_\tau,\discr{a}_\tau,\discr{b}}(\discr{c}) \right] .
\end{equation}
Therefore to get an estimate for $\P[ \tau < \infty, \; \gamma[0,\tau] \subset \discrui{U}{-}_0 ]$
we have to estimate the ratio
\begin{equation*}
\frac{f_\delta^{\discr{U}_\tau,\discr{a}_\tau,\discr{b}}(\discr{c})}{f_\delta^{\discr{U},\discr{a},\discr{b}}(\discr{c})}
\end{equation*}

By the convergence of $f_\delta^{\discr{U},\discr{a},\discr{b}}(\discr{c})$
to $f^{\conti{U},\conti{a},\conti{b}}(\conti{c})$,
we can choose  a constant $\delta_1>0$ such that
\begin{equation*}
f_\delta^{\discr{U},\discr{a},\discr{b}}(\discr{c}) \leq 2 f^{\conti{U},\conti{a},\conti{b}}(\conti{c})
\end{equation*}
for all $0<\delta<\delta_1$. By the uniform convergence of 
$f_\delta^{\discr{U}_\tau,\discr{a}_\tau,\discr{b}}(\discr{c})$
to $f^{\conti{U}_\tau,\discr{a}_\tau,\conti{b}}(\conti{c})$ and by the fact that
$f^{\conti{U}_\tau,\discr{a}_\tau,\conti{b}}(\conti{c})$ is uniformly bounded below by a positive constant,
we can choose a constant $\delta_2>0$ such that
\begin{equation*}
f_\delta^{\discr{U}_\tau,\discr{a}_\tau,\discr{b}}(\discr{c}) \geq \frac{1}{2} f^{\conti{U}_\tau,\discr{a}_\tau,\conti{b}}(\conti{c})
\end{equation*}
for all $0<\delta<\delta_2$. Notice that we continue to denote the tip of the curve by $\discr{a}_\tau$ since
it's the tip of the discrete path $\gamma$ and in fact, $\conti{U}_\tau$ is the continuum domain slitted by the discrete path.

Set $\delta_0 = \min \{\delta_1,\delta_2\}$.
Then 
\begin{equation*}
\frac{f_\delta^{\discr{U}_\tau,\discr{a}_\tau,\discr{b}}(\discr{c})}{f_\delta^{\discr{U},\discr{a},\discr{b}}(\discr{c})}
\geq \frac{1}{4} \frac{f^{\conti{U}_\tau,\discr{a}_\tau,\conti{b}}(\conti{c})}{f^{\conti{U},\conti{a},\conti{b}}(\conti{c})} 
\end{equation*}
for $0< \delta < \delta_0$.

Let $\psi$ be a conformal map that sends $\conti{U}$ to $\half$ and $\conti{b}$ to $\infty$ and
such that its derivative at $\conti{b}$ has modulus equal to $1$ in an appropriate sense.
Then the function $\Psi$ in Theorem~\ref{thm: cs2009} and in \eqref{eq: def f cont spin ising} can be written as
$\Psi = \eta_{\psi(\conti{a})} \circ \psi$ where $\eta_\alpha : z \mapsto -(z-\alpha)^{-1}$.
Notice also that $f^{\conti{U}_\tau,\discr{a}_\tau,\conti{b}}(\conti{c})$ is a constant
times the square root of the conformal map $\eta_{g(\psi(\discr{a}_\tau))} \circ g \circ \psi$
where $g$ is the conformal map
sending $\half \setminus \psi([\conti{a},\discr{a}] \cup \gamma[0,\tau])$ onto $\half$ and normalized hydrodynamically
at $\infty$.
Hence
\begin{align}
\frac{f^{\conti{U}_\tau,\discr{a}_\tau,\conti{b}}(\conti{c})}{f^{\conti{U},\conti{a},\conti{b}}(\conti{c})} 
&= \sqrt{\frac{g'(\psi(\conti{c})) \; \eta_{g(\psi(\discr{a}_\tau))}'(g(\psi(\conti{c})))}{\eta_{\psi(\conti{a})}'(\psi(\conti{c}))}} \nonumber \\
&= \sqrt{\frac{g'(\psi(\conti{c})) \; (\psi(\conti{c})-\psi(\conti{a}))^2}{(g(\psi(\conti{c}))-g(\psi(\discr{a}_\tau)))^2}} . 
\label{eq: ratio spin Ising observable}
\end{align}
It remains to estimate the quantity inside the square root. We will do this in the next section.

\subsubsection{Auxiliary results on conformal maps}

We will slightly simplify the notation of the previous subsection.

Let $0<r<R$ and $\theta_1 < \theta_2 \leq \theta_1 + 2\pi k$, where $k \in \N$ is fixed.
Let $U_0$ be an annular sector $U_0 = \{(\rho,\theta) \in (r,R) \times (\theta_1,\theta_2)\}$,
which we consider as a covering space of $A=A(z_0,r,R)$ through the map $(\rho,\theta) \mapsto z_0 + \rho e^{i\theta}$.
Consider a domain $U$ which is simply connected and is obtained from $U_0$
by gluing (disjointly) to it along the radius $r$ boundary a finite number of $U_j$'s in an appropriate covering space
as in the previous subsection.

Consider some numbers $0< r < r_d < r_c < R_b < R$.
Set 
\begin{gather*}
a=z_0 + R e^{i \theta_1}, \qquad b=z_0 + R_b e^{i \theta_2}, \\ 
c=z_0 + r_c e^{i \theta_2}, \qquad d=z_0 + r_d e^{i \theta_2},
\end{gather*}
where the proportionalities of $R$, $r_c$ and $r_d$ to $r$ are specified later, but we consider $R/R_b$ to be a fixed number close to $1$,
say, equal to $9/10$.
That is, $a,b,c,d$ follow  the counterclockwise order on the boundary of the annular sector
and $a$ is an outer corner of the sector, $b$ is close to the other corner and $b,c,d$ lie on the same ray.
See Figure~\ref{fig: spin ising lemma}.

\begin{figure}[tbh]
\centering
	\includegraphics[scale=.9]
{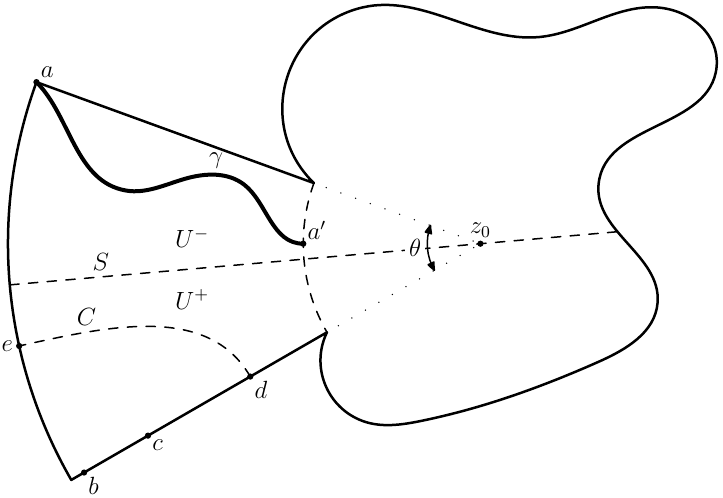}
\caption{The setup for Lemma~\ref{lem: spin ising lemma lemma} and Proposition~\ref{prop: spin ising lemma proposition}.} \label{fig: spin ising lemma}
\end{figure}

Let $S$ be the component of the intersection of $U$ and the line
$\{z_0 + t e^{i (\theta_1 + \theta_2)/2} \,:\, t \in \R \}$ which meets the
annular sector and hence ends at a point on the boundary of the annular section.

Let $\psi: U \to \half$ be conformal and onto such that $b$ is mapped to $\infty$.
Let $C_\half$ be the half-circle in the upper-half plane centered at $\psi(a)$ and 
with the left-end point equal to $\psi(d)$. Set $C= \psi^{-1} (C_\half)$.
Denote $\psi(x)$ by $x_\half$ where $x=a,b,c,d$.

Define further disjoint connected domains $U^\pm$ such that $U = U^- \cup S \cup U^+$
and $U^-$ is next to $a$ and $U^+$ is next to $b$.
Set also $U_0^\pm = U_0 \cap U^\pm$.

\begin{lemma}\label{lem: spin ising lemma lemma}
There is a universal constant $m>1$, such that $C \subset U^+$ if $r_d > m r$.
\end{lemma}

\begin{proof}
Denote the other endpoint of $C$ by $e$. By considering the modulus of the quad
$(U,b,d,a,e)$ and using the conformal invariance of the modulus as well as
the definition of $C_\half$, we see that $e \in \partial U^+$ for large enough $m$. 

Now the path $C$ is characterized by the property that the boundary arcs $ed$ and $de$ of $U$
have both harmonic measure seen form any point of $C$ equal to $1/2$. 
We claim that 
on $S$ the harmonic measure of $de$ is strictly larger than that of $ed$.
This is true in the domain obtained by gluing to $U_0$ along each $I_j$ the cover of $\C$ that is formed from infinite number of sheets
glued along $I_j$. Notice that this is a space where $U$ can be embedded naturally. Here $I_j = \overline{U}_0 \cap \overline{U}_j$. 
Since the harmonic measure of $de$ only increases when the domain is decreased, the claim holds for the general domain $U$.
Thus the lemma follows.
\end{proof}

\begin{proposition}\label{prop: spin ising lemma proposition}
Consider a simple path $\gamma$ from $a$ to some $a'$ with $|a'-z_0|=r$. 
Denote $\psi(\gamma)$ by $\gamma_\half$ and $\psi(a')$ by $a'_\half$.
Suppose that $\gamma$
is contained in $\overline{U_0^-}$. Let $H$ be $\half \setminus \gamma_\half$.
Let $g: H \to \half$ be the Loewner map associated to $\gamma_\half$.
Then for any $M>0$ there exists $N>0$ such that if $R/r \geq N^3$, $r_c = N^2 r$, $r_d = N r$, then
\begin{equation}
g'(c_\half) \;  \frac{(c_\half-a_\half)^2}{(g(c_\half)-g(a'_\half))^2} \geq M .
\end{equation}
\end{proposition}

\begin{proof} 
Set $l=a_\half-d_\half$, which is the radius of the semicircle $C_\half$.
Then $a_\half - c_\half > a_\half - d_\half=l$.

By considering the extremal length of suitable curve families, we can show that
$0< g(a'_\half)-g(d_\half) < g(d_\half)-g(c_\half)$ and that $0 < d_\half - c_\half < \eps l$ for
$\eps>0$ which can be arbitrarily small when $N$ is chosen to be large.
In the former inequality the curve family is the family connecting $a'b$ to $cd$ in $U \setminus \gamma$ and
in the latter it is the family connecting $bc$ to $da$ in $U$, which both have small extremal length.
By translating, we can assume that $a_\half=0$. 
Then $d_\half=-l$ and
$c_\half = -l -l\tilde{\eps}$ with $0 < \tilde{\eps} < \eps$.

We observe that by the properties of Loewner flows  $g(d_\half)-g(c_\half) < d_\half - c_\half = l\tilde{\eps}$.
That is, we can consider $g$ as a Loewner chain at the time corresponding to the value of its half-plane capacity
and use the fact that $t \mapsto |g_t(x)-g_t(x')|$ is decreasing for any real $x,x'$ that lie on the same component of
$\R \setminus \{W_0\}$, where $W_0$ is the Loewner driving term at time $0$. 

The second observation coming from Loewner flows is that $g$ satisfies for all $\eps>0$
\begin{equation*}
g'(c_\half) \geq \tilde{g}'(c_\half) \geq \tilde{g}'(-l-l\eps) = \frac{(1+\eps)^2-1}{(1+\eps)^2} 
  \geq 2 \eps - 3 \eps^2 ,
\end{equation*}
where $\tilde{g} (z)= z + l^2/z$. The first inequality follows from the fact that $g$ and $\tilde{g}$ can be seen
as two time instances of the same Loewner chain and $t \mapsto g'_t(x)$ for real $x$ is decreasing.
The half-plane capacity of $g$ is less than that of $\tilde{g}$, which gives the order of the corresponding
time instances.

Combining these estimates gives
\begin{equation*}
g'(c_\half) \;  \frac{(c_\half-a_\half)^2}{(g(c_\half)-g(a'_\half))^2} 
  \geq g'(c_\half) \; \frac{l^2}{(2l\tilde{\eps})^2}
  \geq (2 \eps - 3 \eps^2) \; \frac{1}{4\eps^2}  = \frac{1}{2\eps} - \frac{3}{4} .
\end{equation*}
And since $\eps>0$ becomes arbitrarily small when $N$ is increased, the claim follows.
\end{proof}

Combining \eqref{ie: martingale prop spin Ising observable} and \eqref{eq: ratio spin Ising observable}
with Proposition~\ref{prop: spin ising lemma proposition} gives Condition~\ref{def: b unf crossing}
for the spin Ising model as we will state in the next proposition. Before that we give the required definitions.

In the following $U$ is an \emph{admissible domain} if all the following conditions hold
\begin{itemize}
\item The domain $U$ is assumed to be a union of full plaquettes of a square lattice (with some mesh size).
\item The domain $U$ is assumed to be cut from the square lattice
by paths $c_1$ and $c_2$ (on the dual lattice), that is,
$U$ is a bounded connected component of $\C \setminus (c_1 \cup c_2)$.
The paths $c_1$ and $c_2$ are assumed to be edge-simple and non-self-crossing, but they are allowed to contain counterclockwise loops
(hence they are not necessarily vertex-simple),
and they assumed to be  mutually disjoint except that they share common parts in both ends (we can assume that
the common parts are at least one edge long in both ends, because we interpret that the boundary conditions change at the end points
from $-1$ to $+1$ and we can always explore the interface by one step).
\item Let $a$ and $b$ be the common end points of $c_1$ and $c_2$.
Then $a$ and $b$ are assumed to be on the boundary of $U$ 
and  $c_1$ and $c_2$ assumed to trace the boundary in clockwise and counterclockwise direction, respectively.
\end{itemize}
This set of conditions is consistent with growing the leftmost path $\gamma$
in the domain wall configuration in the spin Ising model --- to recall the definition see the beginning of
Section~\ref{sssec: observable spin Ising}.
The conformal map $\phi_U$ is defined to be a conformal map from $U$ onto $\disc$ such that
$\phi_U(a)=-1$ and $\phi_U(b)=1$. 

The probability measure $\P_U$ is the law of the leftmost path $\gamma$
in the domain wall configuration in the spin Ising model on the graph corresponding to $U$ with boundary conditions
equal to $+1$ on each vertex on the right of $c_1$ (including the vertices that are inner corners of $\partial U$)
and to $-1$ on each vertex on the left of $c_2$ excluding the vertices are inner corners of $\partial U$.

The martingale property works well with the definition of admissible domains and $\P_U$ and the exploration of $\gamma$.
The following result follows readily from the estimates of this subsection and the previous one.

\begin{proposition}
The collection of the laws of the interface of spin Ising model at criticality on square lattice
(or on isoradial graphs) 
\begin{equation} \label{eq: spin Ising collection}
\Sigma_{\textnormal{spin Ising}} = \{ (U,\phi_U,\P_U) \;:\; U \textnormal{ an admissible domain}  \} 
\end{equation}
satisfies Condition~\ref{def: b unf crossing}. 
\end{proposition}


\subsection{Percolation}

Here we verify 
that the interface of site percolation on the triangular lattice at criticality satisfies Condition~\ref{def: b unf crossing}.
More generally we could work on any graph dual to a planar \emph{trivalent graph}.
The triangular lattice is denoted by $\mathbb{T}$ and it consists of the set of vertices
$\{ x_1 e_1 + x_2 e_2 : x_k \in \Z \}$ where $e_1=1$ and $e_2 = \exp ( i\,\pi/3)$ and
the set of edges such that vertices $v_1,v_2$ are connected by an edge if and only if $|v_1 - v_2|=1$.
The dual lattice of the triangular lattice is the hexagonal lattice $\mathbb{T}'$
consisting of vertices $\{ z_\pm + x_1 e_1 + x_2 e_2 : x_k \in \Z \}$ where 
$z_\pm = (1/\sqrt{3}) \exp( \pm i\,\pi/6)$ and two vertices $v_1,v_2$ are neighbors if $|v_1 - v_2|=1/\sqrt{3}$.

The percolation measure on the whole triangular lattice  with a parameter $p \in [0,1]$ is the probability
measure $\mu_\tl^p$ on $\{ \text{open}, \text{closed} \}^{\tl}$ such that independently each vertex is open with probability $p$
and closed with probability $1-p$. The independence property of the percolation measure gives a consistent way
to define the measure on any subset of $\tl$ by restricting the measure to that set.
The well-known critical value of $p$ is $p_c = 1/2$.

In the case of triangular lattice define the set of \emph{admissible domains} containing any domain $U$
with boundary $\partial U = c_1 \cup c_2$ where $c_1$ and $c_2$ are
\begin{itemize}
\item simple paths on the the hexagonal lattice (write them as $(c_k(n))_{n=0,1,\ldots.N_k}$)
\item mutually avoiding except that they have common beginning and end part:
$c_1(k)=c_2(k)$, $k=0,1,\ldots,l_1$, and $c_1(N_1 - k)=c_2(N_2 - k)$, $k=0,1,\ldots,l_2$, 
where $l_1,l_2 > 0$ 
\item such that 
$a=c_1(0)=c_2(0)$ and $b=c_1(N_1)=c_2(N_2)$ are contained on the boundary of the bounded component
of $\C \setminus (c_1 \cup c_2)$ and furthermore there is at least one path from $a$ to $b$ staying in $U \cap \hl$.
\end{itemize}
The last condition is needed to guarantee that $a$ and $b$ are boundary points of the bounded domain and
that the subgraph containing all the vertices reach from either $a$ or $b$ is connected. Note that the graph
is in fact simply connected.

On an admissible domain $U$ with boundary arcs $c_1$ and $c_2$, denote by $V$ the set of vertices of $\tl$ inside $U$, 
denote by $V_1$ the set of vertices of $\tl$ next to $c_1$ 
and by $V_2$ the set of vertices next to $c_2$. 
Define a probability measure $\mu_{U}^{p}$, $p \in [0,1]$, on the set $\{ \text{open}, \text{closed} \}^{V}$
such that each vertex is independently chosen to be open with the probability $p$ 
and closed with the probability $1-p$ 
and such that it satisfies 
the boundary conditions: the vertices are open on $V_1$ and closed on $V_2$.
Now there are interfaces on $\hl$ separating clusters of open vertices from clusters of closed vertices.
Define $\P_U$ be the law of the unique interface connecting $a$
to $b$ under the critical percolation measure $\mu_{U}^{p_c}$

The proof of the fact that the collection $(\P_U \, : \, U \text{ admissible} )$
satisfies Condition~\ref{def: b unf crossing} couldn't be easier
to give since we have the Russo--Seymour--Welsh theory (RSW).
Let $B_n$ be the set of points in the triangular lattice that are at graph-distance $n$ or less from $0$ and
let $A_n = B_{3n} \setminus B_{n}$ and let $O_n$ be the event that there is a open path inside $A_n$ separating
$0$ from $\infty$. Then there exists $q>0$ such that for any $n$
\begin{equation}
\mu_{\tl}^{p_c} ( O_n ) \geq q .
\end{equation}
Denote by $O'_n$ the event that there is a closed path inside $A_n$ separating
$0$ from $\infty$. By symmetry the same estimate holds for $O'_n$.

Let now $\tilde{A}_n = B_{9n} \setminus B_{n}$, i.e. $\tilde{A}_n$ is the union of the disjoint sets $A_n$ and $A_{3n}$.
Now probabilities that $A_n$ contains an open path and $A_{3n}$ contains a closed path (both separating $0$ from $\infty$)
are independent and hence the corresponding joint event has positive probability
\begin{equation} \label{ie: percolation choking surface}
\mu_{\tl}^{p_c} ( O_n \cap O'_{3n} ) \geq q^2 .
\end{equation}

\begin{proposition}
The collection of the laws of the interface of site percolation at criticality on triangular lattice 
\begin{equation} \label{eq: perc collection}
\Sigma_{\textnormal{Percolation}} = \{ (U,\phi_U,\P_U) \;:\; U \textnormal{ an admissible domain}  \} .
\end{equation}
satisfies Condition~\ref{def: b unf crossing}.
\end{proposition}

\begin{remark}
Exactly the same proof as for FK Ising works for percolation. However RSW provides a simpler way to prove the proposition.
\end{remark}

\begin{proof}
As in the case of FK Ising, we don't have to consider the stopping times at all. The reason for this is that
if $\gamma:[0,N] \to U \cup \{a,b\}$ is the interface parametrized such that $\gamma(k)$, $k=0,1,2,\ldots, N$,
are the vertices along the path, then $U \setminus \gamma(0,k]$ is admissible for any $k=0,1,2,\ldots, N$
and no information is added during $(k,k+1)$. Hence after stopping we stay within the family~\eqref{eq: perc collection}.
Here we also need that the law of percolation conditioned to the vertices explored up to time $n$ is the percolation measure
in the domain where $\gamma(k)$, $k=1,2,\ldots, n$, are erased.

For any $U$, we can apply a translation and consider annuli around the origin. Consider the annular region
$B_{9^N n} \setminus B_n$ for any $n,N \in \N$. By the inequality \eqref{ie: percolation choking surface}
the probability that $\gamma$ makes an unforced crossing is at most $(1-q^2)^N \leq 1/2$, for large enough $N$.
\end{proof}

\subsection{Harmonic explorer}\label{ssec: he}

The result that the harmonic explorer (HE) satisfies Condition~\ref{def: b unf crossing} 
appears already in \cite{schramm-sheffield-2005-}.
We will here just recall the definitions and state the auxiliary result needed. For all the details we refer to
\cite{schramm-sheffield-2005-}.

In this section and also in Sections~\ref{ssec: lerw} and \ref{ssec: apriori ust}
the models are directly related to simple random walk. The next basic estimate is needed for
bounds like in Conditions~\ref{cond: geom power-law} and \ref{cond: conf power-law}.

\begin{lemma}[Weak Beurling estimate of simple random walk]\label{lm: beurling}
Let $L=\Z^2$ or $L=\mathbb{T}$ and let
$(X_t)_{t =0,1,2,\ldots}$ be a simple random walk on $L$ with the law $P_x$ such that $P_x(X_0 = x )=1$
and let $\tau_B$ be the hitting time of a set $B$.
For an annulus $A=A(z_0,r,R)$, denote by $E(A)$ the event that a simple random walk starting at $x \in A \cap L$
makes a non-trivial loop around $z_0$ before exiting $A$, that is,
there exists $0 \leq s < t \leq \tau_{\C \setminus A}$ s.t. $X|_{[s,t]}$ is not nullhomotopic in $A$.
Then there exists $K>0$ and $\Delta>0$ such that 
\begin{equation*}
P_x ( E(A(z_0,r,R))) ) \geq 1 - K \left(\frac{r}{R}\right)^\Delta
\end{equation*}
for any annulus $A(z_0,r,R)$ with $1 \leq r \leq R$ and for any $x \in A(z_0,r,R) \cap L$ such that 
$\sqrt{r R} - 1 < |x-z_0| < \sqrt{r R} + 1$.
\end{lemma}

\begin{proof}[Sketch of proof] 
Either use the similar property of Brownian motion and the convergence of simple random walk to Brownian motion
or construct the event $E(A(z_0,r,4r)$ for $|x - z_0| \approx 2 r$ from elementary events which, for $L=\Z^2$,
are of the type that
a random walk started from $(n,n) \in \Z^2$ will exit the rectangle $R_n=[0,\lfloor a n \rfloor] \times [0,2 n]$
through the side $\{ \lfloor a n \rfloor \} \times [0,2 n]$.
That elementary event for given $a>1$ has positive probability uniformly over all $n$.
\end{proof}

We use here the same definition as in the case of percolation for admissible domains, for $c_k$, for $V_k$ etc.
In the same way as above, the random curve $\gamma$ will be defined on $\hl$. We describe here
how to take the first step in the harmonic explorer. 
Let $U$ be an admissible domain
and choose $a$ and $b$ in some way. 
Suppose for concreteness that $c_1$ follows the boundary clockwise from $a$ to $b$
and therefore $c_1$ lies to the ``left'' from $a$ and $c_2$ lies to the ``right''. Denote by $H_U : U \cap \tl \to [0,1]$
the discrete harmonic function on $U \cap \tl $ that has boundary values $1$ on $V_1$ and $0$ on $V_2$.

Now $\gamma(0)=a$ has either one or two
neighbor vertices in $U$. If it has only one, then set $\gamma(1)$ equal to that vertex. If it has two neighbors, say,
$w_L$ and $w_R$ (defined such that $w_L-a, w_R - a, c_1(1) - a $ are in the clockwise order)
calculate the value of $p_0=H_U(v_0)$ at the center $v_0$ of the hexagon that is lying next to all these three vertices.
Then flip a biased coin and set $\gamma(1) = w_R$ with probability $p_0$ and $\gamma(1) = w_L$ with probability $1-p_0$.
Note that the rule followed when there is only one neighbor can be seen as a special case of the second rule.

Extend $\gamma$ linearly between $\gamma(0)$ and $\gamma(1)$ and set now $U_1 = U \setminus \gamma(0,1]$ which is an admissible
domain. Repeat the same procedure for $U_1$ to define $\gamma(2)$ using a biased coin independent from the first one
so that the curve turns right with probability $p_1=H_{U_1}(v_1)$ and left with probability $1-p_1$ where $v_1$
is the center of the hexagon next to $\gamma(1)$ and its neighbors except for $\gamma(0)$.
Then define $U_2 = U_1 \setminus \gamma(1,2] = U \setminus \gamma(0,2]$ and continue the construction in the same manner.
This repeated procedure defines a random curve $\gamma(k)$, $k=0,1,2, \ldots, N$, such that $\gamma(0) = a$, $\gamma(N)=b$,
$\gamma$ is simple and stays in $U$.

A special property of this model is that the values of the harmonic functions $M_n = H_{U_n}(v)$ for fixed $v \in U \cap \tl$
but for randomly varying $U_n$ defined as above will be a martingale with respect to the $\sigma$-algebra generated by
the coin flips or equivalently by the curve or the domains $(U_n)$.

It turns out that in this case, the harmonic ``observables'' $(H_{U_n}(v))_{v \in U \cap \tl, n = 0,1,\ldots, N}$,
provide also a method to verify the Condition~\ref{def: b unf crossing}. 
This is done in Proposition~6.3 of the article \cite{schramm-sheffield-2005-}.
We only sketch the proof here.
Let $U$ be an admissible domain and $A = A(z_0,r,R)$ an annulus.
Let $V_-$ be the set of
vertices in $V_1 \cap B(z_0,3r)$ that are disconnected from $b$ by $A^u$ and let the corresponding part of $A^u$ be $A^u_-$.
Let
$\tilde{M}_n = \sum_{x \in V_-} \tilde{H}_{U_n} (x)$, where $\tilde{H}_{U}(x)$, $x \in V_1$ is defined to be the harmonic measure
of $V_2$ seen from $x$ and can be expressed in terms of $H_U$ as the average value $H_U$ among the neighbors of $x$.
Now the key observation in the above proof is that $(\tilde{M}_n)$ is a martingale with 
$\tilde{M}_0 = \OO( (r/R)^\Delta )$ for some $\Delta>0$ (following from Beurling estimate of simple random walk)
and on the event of crossing one of $A^u_-$ it increases to $\OO(1)$. A martingale stopping
argument tells that the probability of the crossing event is then $\OO( (r/R)^\Delta )$.

\begin{proposition}[Schramm--Sheffield]
The family of harmonic explorers satisfies Condition~\ref{def: b unf crossing}.
\end{proposition}

%
%

\subsection{Chordal loop-erased random walk}\label{ssec: lerw}

The loop-erased random walk is one of the random curves proved to be conformally invariant. In \cite{lawler-schramm-werner-2004-},
the radial loop-erased random walk between an interior point and a boundary point was considered.
We'll treat here the chordal loop-erased random walk between two boundary points.
Condition~\ref{def: b unf crossing} is slightly harder to verify in this case. 
Namely, the natural extension of Condition~\ref{def: b unf crossing} to the radial case can be verified in the same way, except that
Proposition~\ref{prop: srw crosses quad} is not necessary, and
it is done in \cite{lawler-schramm-werner-2004-}.
For another approach, yet similar, see \cite{zhan-2008b-}.

Let $(X_t)_{t =0,1,\ldots}$ be a simple random walk (SRW) on the lattice $\Z^2$ and $P_x$ its law so that $P_x(X_0=x)=1$.
Consider a bounded, simply connected domain $U \subset \C$ whose boundary $\partial U$ is a path in $\Z^2$.
Call the corresponding graph $G$, i.e., $G$ consists of vertices $\overline{U} \cap \Z^2$ and the edges
which stay in $U$ (except that the end points may be in $\partial U$).
Let $V$ be the set of vertices and $\partial V \dd= V \cap \partial U$.
When $X_0=x \in \partial V$ condition SRW on $X_1 \in U$. For any $X_0=x \in V$
define $T$ to be the hitting time of the boundary, i.e., $T = \inf\{ t \geq 1 \,:\, X_t \in \partial V\}$.

Denote by $\tau_A$ the hitting time of the set $A$ by the simple random walk $(X_t)_{t=0,1,\ldots}$
or $(X_t)_{t=0,1,\ldots, T}$. Let $\omega_U( x,A) = P_x^U(X_T \in A) = P_x^U(\tau_A \leq T)$
which is the \emph{discrete harmonic measure} of $A$ in $U$ as seen from $x$.
The quantity $\omega_U( x,A)$ is discrete harmonic in $x$ and satisfies the properties of a measure with respect to $A$.

For $a \in V$ and $b \in \partial V$ define $P_{a \to b}=P_{a \to b}^U$ to be the law of $(X_t)_{t=0,1,2,\ldots,T}$ with $X_0=a$
conditioned on $X_T=b$.
If $(X_t)_{t=0,1,2,\ldots,T}$ distributed according to $P_{a \to b}^U$ then the process $(Y_t)_{t=0,1,2,\ldots,T'}$,
which is obtained from $(X_t)$ by erasing all loops in chronological order,
is called \emph{loop-erased random walk} (LERW) from $a$ to $b$ in $U$.
Denote its law by $\P^{U,a,b}$. 
We will show that the collection $\{\P^{U,a,b} \,:\, (U,a,b) \}$ of chordal LERWs satisfies Condition~\ref{cond:quad}, where
$U$ runs over all simply connected domains as above and $\{a,b\} \subset \partial U$.

\begin{proposition}\label{prop: srw crosses quad}
There exists $\eps_0>0$ such that
for any $c>0$ there exists $L_0>0$ such that the following holds. Let $U$ be a discrete domain 
($\partial U$ is a path in $\Z^2$) and 
let $Q$ be a topological quadrilateral with ``sides'' $S_0, S_1, S_2, S_3$ and which lies on the boundary 
in the sense that $S_1, S_3 \subset \partial U$. Let $A \subset V \setminus Q$ be a set of vertices such that
$S_0$ disconnects $S_2$ from $A$. If $\elen(Q) \geq L_0$, then there exists $u \in Q$ and $r>0$ such that 
\begin{enumerate}\enustyii
\item $B \dd= V \cap B(u,r) \subset Q$, \label{eni: lerw prop 1}
\item $\min_{x \in B} \omega_U(x,A) \geq c \, \max_{x \in S_2} \omega_U(x,A)$ and \label{eni: lerw prop 2}
\item $P_{x \to y}^Q ( X[0,T] \cap B \neq \emptyset ) \geq \eps_0$ for any $x \in S_0$ and $y \in S_2$. \label{eni: lerw prop 3}
\end{enumerate}
\end{proposition}

\begin{proof}
Cut $Q$ into three quads (topological quadrilaterals) by transversal paths 
$p_1$ and $p_2$
and call these quads $Q_k$, $k=1,2,3$. The sides of $Q_k$ are denoted by $S_j^k$, $j=0,1,2,3$, and we assume
that $S_0^1=S_0$, $S_2^1=p_1=S_0^2$, $S_2^2=p_2 =S_0^3$ and $S_2^3=S_2$. 

We assume that $\elen (Q_1) = \elen(Q_2) = l$ and $\elen(Q_3)=L - 2l$ where $L=\elen(Q)$.
Using the Beurling estimate, Lemma~\ref{lm: beurling}, it is possible to fix $l$ 
so large that $\omega_{Q_1 \cup Q_2} (z, S_0^1 \cup S_2^2) \leq 1/100$ for any $z$ on the discrete path closest
to $S_2^1=S_0^2$.

Since the harmonic measure $z \mapsto \omega_{Q_1 \cup Q_2} (z, S_1^1 \cup S_1^2)$ changes at most by a constant factor
between neighboring sites, we can find $u$ along the discrete path closest
to $S_2^1=S_0^2$ in such a way that 
$\omega_{Q_1 \cup Q_2} (u, S_1^1 \cup S_1^2), \omega_{Q_1 \cup Q_2} (u, S_3^1 \cup S_3^2) \geq 1/6$. Let $r$ be equal
to half of the inradius of $Q_1 \cup Q_2$ at $u$. Then $B \dd= V \cap B(u,r)$ satisfies \ref{eni: lerw prop 1} by definition
and \ref{eni: lerw prop 3} for some $\eps_0>0$ follows from Proposition 3.1 of \cite{chelkak-2012-}.

Let $H(x)= \omega_U( x,A)$.
Let $c'>0$ be such that $H(x) \geq c' \, H(y)$ for any $x,y \in B$. The constant $c'$ can be chosen to be universal
by Harnack's lemma. Let $M= \max_{x \in S_2^2} H(x)$ and let $x^*$ be the point where the maximum is attained.
By the maximum principle there is a path $\pi$ from $x^*$ to $A$ such that $H \geq M$ on $\pi$. Now
$H(u) \geq M/6$ and hence $\min_{x \in B} H(x) \geq M c'/6$. Finally by the Beurling estimate 
$\max_{x \in S_2^2} H(x) \geq \exp (  \alpha \elen(Q_3) ) \max_{x \in S_2} H(x)$ for some universal constant $\alpha>0$.
And hence we can choose $L_0$ so large that \ref{eni: lerw prop 2} holds for any $L \geq L_0$.
\end{proof}

\begin{theorem}
Condition~\ref{cond:quad} holds for LERW.
\end{theorem}

\begin{proof}
Let $L_0>0$ and $\eps_0>0$ be as in Proposition~\ref{prop: srw crosses quad} for $c=2$.
Consider a quad $Q$ with $L=\elen(Q) \geq L_0>0$ as in 
Proposition~\ref{prop: srw crosses quad} for $A=\{b\}$. 
We will show that there is uniformly positive
probability that $(X_t)_{t=0,1,\ldots,T}$ conditioned on $X_T=b$ doesn't cross $Q$. By iterating
that estimate $n \in \N$ times (for large enough $n$) 
we get that the probability of crossing is at most $1/2$ for $L\geq n\,L_0$.

We can assume $P_a ( \tau_{S_2} < T \,|\, X_T=b) \geq 1/2$, otherwise there wouldn't be anything to prove.
By the previous proposition
\begin{equation*}
P_a(\tau_B < (\tau_{S_2} \wedge T) \,|\, X_T=b) \geq \eps_0 P_a( \tau_{S_2} < T \,|\, X_T=b) \geq \frac{\eps_0}{2} .
\end{equation*}
Now since $\max_{x \in S_2} P_x(X_T=b) \leq (1/2) \cdot \min_{y \in B} P_y(X_T=b)$ by assumption, 
\begin{align*}
P_y(\tau_{S_2} < T \,|\, X_T=b) &= \frac{P_y(\tau_{S_2} < T ,\, X_T=b)}{P_y(X_T=b)} \nonumber \\
&\leq \frac{\max_{x \in S_2} P_x(X_T=b)}{P_y(X_T=b)} \leq \frac{1}{2} 
\end{align*}
for any $y \in B$.
Combine these estimates to show that
\begin{equation*}
P_a ( \tau_{S_2} < T \,|\, X_T=b) \leq 1 - P_a(\tau_B < T < \tau_{S_2}) \,|\, X_T=b) \leq 1 - \frac{\eps_0}{4}
\end{equation*}
from which the claim follows.
\end{proof}

%
%

\subsection{Condition~\ref{def: b unf crossing} fails for uniform spanning tree}\label{ssec: apriori ust}%

For a given connected graph $G$, a spanning tree is a subgraph $T$ of $G$ such that $T$ is a tree, i.e.,
connected and without any cycles, and $T$ is spanning, i.e., $V(T)=V(G)$. A \emph{uniform spanning tree} (UST) of $G$ is 
a spanning tree sampled uniformly at random from the set of all spanning trees of $G$. More precisely,
if $T$ is a uniform spanning tree and $t$ is any spanning tree of $G$ then
\begin{equation}
\P(T=t) = \frac{1}{N(G)}
\end{equation}
where $N(G)$ is the number of spanning trees of $G$. The UST model can be analyzed via simple random walks and electrical networks,
see \cite{grimmett-2010-} and references therein. 
The conformal invariance of UST on planar graphs was shown in \cite{lawler-schramm-werner-2004-} where Lawler, Schramm and Werner proved
that the UST Peano curve (see below) converges to SLE(8). Their work partly relies on Aizenman--Burchard theorem and
\cite{aizenman-etall-1999-} where the relevant
crossing estimate was established. 

Concerning the current work, the UST Peano curve gives a counterexample: it is a curve otherwise
eligible but it fails to satisfy Condition~\ref{def: b unf crossing}. 
For discrete random curves converging to some SLE$(\kappa)$ this
shows that Condition~\ref{def: b unf crossing} is only relevant to the case $0 \leq \kappa < 8$. In some sense $0 \leq \kappa \leq 8$
is the physically relevant case. For instance, the reversibility property holds only in this range of $\kappa$.
Therefore it is interesting to extend the methods of this paper to the spacial case of UST Peano curve.
This shown to be possible by the first author of this paper in \cite{kemppainen-2013-}.

Consider a finite subgraph $G_\delta \subset \delta \Z^2$, $\delta>0$, which is simply connected, i.e., it is a union of entire faces 
of $\delta \Z^2$
such that the corresponding domain is Jordan domain of $\C$.
A boundary edge of $G_\delta$ is an edge $e$ in $G_\delta$ such that there is a face in $\delta \Z^2$ which contains $e$
but which doesn't belong entirely to $G$.
Take a non-empty connected set $E_W$ of boundary edges not equal to the entire set of boundary edges. 
Then $E_W$ will be a path which we call $\emph{wired boundary}$. Call its end points in the counterclockwise 
direction as $\tilde{a}_\delta$ and $\tilde{b}_\delta$. 

Let $T$ be a uniform spanning tree on $G_\delta$ conditioned on $E_W \subset T$. Then $T$ can be seen as an unconditioned UST of
the contracted graph $G_\delta / E_W$. The UST Peano curve is defined to be the simple cycle $\gamma$ on $\delta (1/4 + \Z/2)^2$
which is clockwise oriented and follows $T$ as close as possible, i.e., for each $k$, there is either a vertex of $G_\delta$ 
on the right-hand side of $(\gamma(k),\gamma(k+1))$ or there is a edge of $T$. We restrict this path to a part which goes from a point
next to $\tilde{a}_\delta$ to a point next to $\tilde{b}_\delta$. With an appropriate choice of the domain $U_\delta$, $\gamma$
is a simple curve in $U_\delta$ connecting boundary points $a_\delta$ and $b_\delta$ and it is also a space-filling curve, i.e.,
$\gamma$ visits all the vertices $U_\delta \cap \delta (1/4 + \Z/2)^2$.

It is easy to see that $\gamma$ doesn't satisfy Condition~\ref{def: b unf crossing}: since it is space filling it will make 
an unforced crossing of $A(z_0,r,R)$ with probability $1$ if there are any sites which are disconnected from $a_\delta$ and $b_\delta$
by a component of $A(z_0,r,R)$. However the probability of having more than $2$ crossings in such a component is small. This approach is taken
in \cite{kemppainen-2013-}, where it is sufficient to consider the following event: let $Q \subset U_\delta$ be a topological quadrilateral such that
$\partial_0 Q $ and $\partial_2 Q$ are subsets of $U_\delta$ and $\partial_1 Q $ and $\partial_3 Q$ are subsets of wired part of
$\partial U_\delta$. Then $Q$ has the property that it doesn't disconnect $a_\delta$ from $b_\delta$.
We call a set of edges $E \subset E(G_\delta) \cap Q$ a \emph{confining layer} if $E \cup E_W$ is a tree and
the vertex set of $E$ contains a crossing of $Q$ from $\partial_1 Q $ to $\partial_3 Q$. This is equivalent to the fact that
there can be only two vertex disjoint crossings of $Q$ from $\partial_0 Q $ to $\partial_2 Q$ by paths in
$U_\delta \cap \delta (1/4 + \Z/2)^2$.
Based on the connection to the simple random walk by Wilson's method, it is possible to establish that there exists
universal constant $c>0$ such that
\begin{equation}
\P( T \text{ contains a confining layer in } Q ) \geq 1 - \frac{1}{c} \exp( -c \elen(Q,\partial_0 Q, \partial_2 Q) ) .
\end{equation}
This will be sufficient for the treatment of UST Peano curve in the similar manner as in this article.

%% file: random_planar_curves_appendix.tex

\section{Appendixes}

\subsection{Schramm--Loewner evolution} \label{ssec: sle}

We will be interested in describing random curves in simply connected
domains with boundary in the complex plane 
by Loewner evolutions with random driving functions.
Since the setup for Loewner evolutions is conformally invariant,
we can define them in some fixed domain.
A standard choice is the upper half-plane $\half := \{ z \in \C : \imag(z)>0 \}$. 
Another choice could be the unit disc 
$\disc := \{ z \in \C : |z| < 1 \}$. 

Consider a simple curve $\gamma:[0,T] \to \C$
such that $\gamma(0) \in \R$ and $\gamma(t) \in \half$ for any $t>0$.
Let $K_t = \gamma [0,t]$ and $H_t = \half \setminus K_t$. Note that $K_t$ is compact and 
$H_t$ is simply connected. 

There is a unique conformal mapping $g_t : H_t \to \half$ satisfying the normalization 
$g_t(\infty)=\infty$ and $\lim_{z \to \infty} [g_t(z) - z] = 0$. 
This is called the \emph{hydrodynamical normalization} and then around the infinity
\begin{equation} \label{eq: hydrodn}
g_t(z) = z + \frac{a_1(t)}{z}+\frac{a_2(t)}{z^2}+\ldots
\end{equation}
The coefficient $a_1(t)=\hcap(K_t)$ is called the \emph{half-plane capacity} of $K_t$ or shorter the \emph{capacity}.
Quite obviously, $a_1(0)=0$, and it can be shown that $t \mapsto a_1(t)$ is strictly increasing and continuous.
The curve can be reparameterized (which also changes the value of $T$) such that $a_1(t)=2t$ for each $t$.

Assuming the above normalization and parameterization, 
the family of mappings $(g_t)_{t \in [0,T]}$ satisfies the upper half-plane version of the \emph{Loewner differential equation}, that is
\begin{equation} \label{eq: loewner eq}
 \frac{\partial g_t}{\partial t} (z) = \frac{2}{g_t(z) -W_t}
\end{equation}
for any $t \in [0,T]$,
where the ``\emph{driving function}'' $t \mapsto W_t$ is continuous and real-valued. 
It can be proven that $g_t$ extends continuously to the point $\gamma(t)$ and $W_t = g_t(\gamma(t))$.
For the proofs of these facts see Chapter 4 of \cite{lawler-2005-}. An illustration of the construction is 
in Figure \ref{fig: basic le}. The equation or rather its version on the unit disc was introduced
by Loewner in 1923 in his study of the Bieberbach conjecture \cite{loewner-1923-}. 

\begin{figure}[tb!]
\centering
	\label{fig: sle}
	\includegraphics[scale=1.3]%
{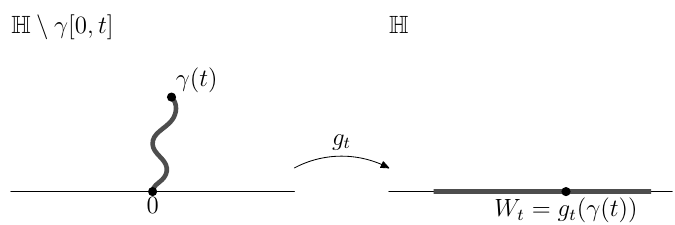}
\caption{The mapping $g_t$ maps the complement of $\gamma[0,t]$ onto the upper half-plane. The tip $\gamma(t)$ is mapped to a point
$W_t$ on the real line.} \label{fig: basic le}
\end{figure}

Consider more general families of growing sets.
Call a compact subset $K$ of $\overline{\half}$ such that $\half \setminus K$ is simply connected, as a \emph{hull} .
The sets $K_t$ given by a simple curve, as above, are hulls. Also other families of hulls 
can be described by the Loewner equation with a continuous driving function. 
The necessary and sufficient condition is given in the following proposition.
Also some facts about the capacity are collected there.

\begin{proposition} \label{prop: sle basic}
Let $T >0$ and $(K_t)_{t \in [0,T]}$ a family of hulls s.t. $K_s \subset K_t$, for any $s < t$, 
and let $H_t = \half \setminus K_t$.
\begin{itemize}
\item If $(K_t \setminus K_s) \cap \half \neq \emptyset$ for all $s < t$, then $t \mapsto \hcap(K_t)$ is strictly increasing
\item If $t \mapsto H_t$ is continuous in Carath\'eodory kernel convergence, then $t \mapsto \hcap(K_t)$ is continuous. 
\item Assume that $\hcap (K_t) = 2t$
  (under the first two assumptions there is always such a time reparameterization).
   Then there is a continuous driving function $W_t$ so that
  $g_t$ satisfies Loewner equation \eqref{eq: loewner eq} 
  if and only if
  for each $\delta >0$ there exists $\eps>0$ so that for any $0\leq s < t \leq T$, $|t-s|<\delta$, 
  a connected set $C\subset H_s$ can be chosen such that $\diam(C)<\eps$ and $C$ separates $K_t \setminus K_s$ from infinity.
\end{itemize}
\end{proposition}

Two first statements are relatively simple. The second claim is
almost self-evident: Carath\'eodory kernel convergence means that $g_s \to g_t$ as $s \to t$ in
the compact subsets of $H_t$ and then we have to use the fact that $\hcap(K)$ can be expressed as an integral 
$\frac{1}{2 \pi } \int_0^\pi \real \big(R e^{i\theta} g_{K} (R e^{i\theta})\big) \de \theta$ for $R$ large enough.
The third claim is proved in \cite{lawler-schramm-werner-2001-}. 

By the third claim not all continuous $W_t$ correspond to a simple
curve. One important class of $(K_t)_{t \in [0,T]}$ are the ones generated by a curve in the following sense:
For a curve $\gamma:[0,T] \to \overline{\half}$, $\gamma(0) \in \R$, that is not necessarily simple, define
$H_t$ to be the unbounded component of $\half \setminus \gamma[0,t]$ and $K_t = \overline{ \half \setminus H_t}$.
For each $t$, $K_t$ is a hull and
the collection of hulls $(K_t)_{t \in [0,T]}$ is said to be \emph{generated by the curve} $\gamma$.
But even this class is not general enough:
a counterexample is a spiral that winds infinitely many times around a circle in the upper half-plane and then unwinds,
see for example the discussion in the article
by Lind, Marshall and Rohde \cite{lind-etall-2010-}.

A \emph{Schramm--Loewner evolution}, SLE$_\kappa$, $\kappa > 0$, is a random $(K_t)_{t \geq 0}$ corresponding to a random driving function
$W_t = \sqrt{\kappa} B_t$ where $(B_t)_{t \geq 0}$ is a standard one-dimensional Brownian motion. SLE was introduced by
Schramm \cite{schramm-2000-} in 1999. An important result about them is that they are curves in the following sense:
\begin{itemize}
\item $0<\kappa\leq 4$: $K_t$ is a simple curve 
\item $4<\kappa < 8$: $K_t$ is generated by a curve 
\item $\kappa \geq 8$: $K_t$ is a space filling curve 
\end{itemize}
This is proven  $\kappa\neq 8$ in \cite{rohde-schramm-2005-}. For $\kappa =8$, it follows since SLE$_8$ is a scaling limit
of a random planar curve in the sense explained in the current paper, see \cite{lawler-schramm-werner-2004-}.
So based on this result, the above definition can be reformulated: 
a Schramm--Loewner evolution is a random curve in the upper half-plane whose Loewner evolution is driven by a Brownian motion.

In fact, Schramm--Loewner evolutions are characterized by the conformal Markov property 
\cite{schramm-2000-}, see \cite{smirnov-2006-} for an extended discussion.
For this reason, if the scaling limit of a random planar curve is conformally invariant in an appropriate sense, then
it has to be SLE$_\kappa$, for some $\kappa>0$.

\subsection{Equicontinuity of Loewner chains}

In this section, we prove simple statements about equicontinuity of general Loewner chains.
For $g_t$ as in the previous section, denote its inverse by $f_t = g_t^{-1}$, which
satisfies the corresponding Loewner equation
\begin{equation}
\partial_t f_t(z) = - f_t'(z) \frac{2}{z-W_t}
\end{equation}
together with the initial condition $f_0(z)=z$. We call any of the equivalent objects
$(g_t)_{t \in [0,T]}$, $(f_t)_{t \in [0,T]}$ and $(K_t)_{t \in [0,T]}$ as a \emph{Loewner chain}
(with the driving term $(W_t)_{t \in [0,T]}$).

Let $V_{T,\delta} =[0,T] \times \{ z \in \C \,:\, \imag z \geq \delta \}$

\begin{lemma}\label{lm: equicontinuity of loewner chains}
For any $T,\delta>0$ the family
\begin{equation}\label{eq: def loewner chain family}
\left\{ \tilde{F}: V_{T,\delta} \to \C \,:\, 
  \begin{gathered}
  \text{there is a Loewner chain $(f_t)_{t \in \R_+}$ s.t. } \\
  \tilde{F}(t,z)= f_t(z), \, \forall (t,z) \in V_{T,\delta} 
  \end{gathered}
  \right\}
\end{equation}
is equicontinuous and
\begin{equation}\label{eq: lmeclc nabla bound}
\left| \partial_t \tilde{F} (t,z) \right|  \leq \frac{2}{\delta} e^{8 \frac{t}{\delta^2}} ,\qquad
\left| \partial_z \tilde{F} (t,z) \right|  \leq  e^{8 \frac{t}{\delta^2}}
\end{equation}
for any $\tilde{F}$ in the set \eqref{eq: def loewner chain family} and for any $(t,z) \in V_{T,\delta}$.
\end{lemma}

\begin{proof}
Since $V_{T,\delta}$ is convex, it is sufficient to show \eqref{eq: lmeclc nabla bound}.
The equicontinuity follows from that bound by integrating along a line segment in $V_{T,\delta}$

Let $\Phi_w (z) = i ( \imag w ) \frac{1+z}{1-z} + \real w$ and $f : \half \to \C$ be any conformal map.
Then 
\begin{equation*}
z \mapsto \frac{f \circ \Phi_w (z) - f \circ \Phi_w (0)}{\Phi_w'(0)f'(w)}
\end{equation*}
belongs to the class $S$ of univalent functions, see Chapter~2 of \cite{duren-1983-}, and therefore by
Bieberbach's theorem
\begin{equation*}
(\imag w) \left| \frac{f''(w)}{f'(w)} \right| \leq 3 .
\end{equation*}
If we apply this bound to the Loewner equation of $f'_t$ we find that
\begin{equation*}
|\partial_t f_t'(z)| \leq \frac{8}{(\imag z)^2} |f_t'(z)|
\end{equation*}
and therefore
\begin{equation*}
|f_t'(z)| \leq e^{8\frac{t}{(\imag z)^2}} \leq e^{8\frac{T}{\delta^2}} .
\end{equation*}
Furthermore, plugging this estimate in the Loewner equation gives
\begin{equation*}
|\partial_t f_t(z)| \leq \frac{2}{\imag z} e^{8\frac{T}{\delta^2}} \leq \frac{2}{\delta} e^{8\frac{T}{\delta^2}} .
\end{equation*}
\end{proof}

For $T,\delta>0$ and family of hulls $(K_t)_{t \in [0,T]}$, let
\begin{equation}
S_K(T,\delta) = \{ (t,z) \in [0,T] \times \overline{\half} \,:\, \dist(z,K_t) \geq \delta\}
\end{equation}

\begin{lemma}\label{lm: curve convergence implies caratheodory}
Let $\gamma_n$ be a sequence of curves in $\half$
and let $\gamma$ be a curve in $\half$ all parametrized with the interval $[0,T]$, $T>0$,
and let $g_{n,t}$ and $g_t$ be the normalized conformal maps related to
the hulls $K_{n,t}$ and $K_t$ of $\gamma_n[0,t]$ and $\gamma[0,t]$, respectively.
If $\gamma_n \to \gamma$ uniformly, then $g_{n,t} \to g_t$ uniformly on $S_K(T,\delta)$. Especially $\hcap \gamma_n [0,\cdot] \to \hcap \gamma [0,\cdot]$
uniformly.
\end{lemma}

\begin{proof}
The lemma follows from the Carath\'eodory convergence theorem (Theorem~3.1 of \cite{duren-1983-} 
and Theorem~1.8 of \cite{pommerenke-1992-}). Convergence is uniform in $t$ since the interval $[0,T]$ is compact.
\end{proof}

\begin{lemma}\label{lm: driving term convergence implies caratheodory}
Let $W_n$ be a sequence of continuous functions on $[0,T]$
and let $W$ be a continuous functions on $[0,T]$
and let $g_{n,t}$ and $g_t$ be the solutions of Loewner equation with the driving terms
$W_{n,t}$ and $W_t$, respectively, and let $K_t$ be the hull of $g_t$.
If $W_n \to W$ uniformly, then $g_{n,t} \to g_t$ uniformly on $S_K(T,\delta)$ and $g_t$ and $W$ satisfy the Loewner equation.
\end{lemma}

\begin{proof}
This lemma follows from the basic properties of ordinary differential equations.
\end{proof}

\subsubsection{Main lemma}

Consider a sequence $\tilde{\gamma}_n \in \xs(\disc)$ with $\tilde{\gamma}_n(0)=-1$ and $\tilde{\gamma}_n(1)=+1$
which converges to some curve $\tilde{\gamma} \in X$. After choosing a parametrization and using 
the chosen conformal transformation from $\disc$ to $\half$, it is natural to consider for some $T>0$
a sequence of one-to-one continuous functions $\gamma_n : [0,T] \to \C$ with $\gamma_n(0) \in \R$ and 
$\gamma(0,T] \subset \half$ such that $\gamma_n$ converges
uniformly to a continuous function $\gamma:[0,T] \to \C$ which is not constant on any subinterval of $[0,T]$.
In this section we present practical conditions under which $\gamma$ is a Loewner chain, that is,
$\gamma$ can be reparametrized with the half-plane capacity.

Let
\begin{equation}
\upsilon_n(t) = \frac{1}{2} \hcap( \gamma_n[0,t]), \qquad
\upsilon(t)   = \frac{1}{2} \hcap( \gamma[0,t]) .
\end{equation}
Then $t \mapsto \upsilon_n(t)$ and $t \mapsto \upsilon(t)$ are continuous 
and $\upsilon_n \to \upsilon$ uniformly as $n \to \infty$.
In particular, $\lim_n \upsilon_n(T) = \upsilon(T)$ and by the assumptions $\upsilon(1)>0$.
Furthermore, $t \mapsto \upsilon(t)$ is non-decreasing.
Let $(W_n(t))_{t \in [0,\upsilon_n(T)]}$ be the driving term of $\gamma_n$ which exists since $\gamma_n$ is simple.

When is it true that $\gamma$ has a continuous driving term? It is a fact that if $\upsilon$ is strictly increasing
then $\gamma$ has a driving term $W((t))_{t \in [0,\upsilon(1)]}$ and that $W_n \to W$ uniformly on $[0,\upsilon(1))$.
However we won't prove this auxiliary result, instead we prove a weaker result which gives a practical conditions to 
be verified.

\begin{lemma}\label{lm: main lemma with convergence}
Let $T>0$ and for each $n \in \N$, let $\gamma_n : [0,T] \to \C$ be injective continuous function such that
$\gamma_n(0) \in \R$ and $\gamma_n (0,T] \subset \half$. Suppose that
\begin{enumerate}
\item $\gamma_n \to \gamma$ uniformly on $[0,T]$ and $\gamma$ is not constant on any subinterval of $[0,T]$
\item $W_n \to W$ uniformly on $[0,\upsilon(T)]$.  
\item $F_n \to F$ uniformly on $[0,T] \times [0,1]$, where 
\begin{equation}
F_n(t,y) = g_{\gamma_n[0,t]}^{-1} \big(\, W_n(\upsilon_n(t)) + iy \,\big) .
\end{equation}
\end{enumerate}
Then $t \mapsto \upsilon(t)$ is strictly increasing and $g_t \coloneq g_{\gamma\circ\upsilon^{-1}[0,t]}$
satisfies the Loewner equation with the driving term $W$. Furthermore, $\gamma_n \circ \upsilon_n^{-1}$, which is
the sequence of curves in the capacity parametrization, converges uniformly to $\gamma \circ \upsilon^{-1}$ and
the sequence of mappings
$(t,z) \mapsto g_{\gamma_n\circ\upsilon_n^{-1}[0,t]}(z)$ converges to $g_t$ uniformly on
\begin{equation}\label{eq: lmmain domain conv g}
S_K(T,\delta) = \{(t,z) \in [0,T] \times \overline{\half} \,:\, \dist(z, K_t) \geq \delta \}
\end{equation}
for any $\delta>0$. Here $K_t$ is the hull of $\gamma[0,t]$.
\end{lemma}

\begin{remark}
By applying a scaling and corresponding time change, it's enough that there exists $\eps>0$ such that
$F_n \to F$ uniformly on $[0,T] \times [0,\eps]$.
\end{remark}

\begin{proof}
By Lemma~\ref{lm: curve convergence implies caratheodory}, 
$\gamma_n \to \gamma$ implies that $\upsilon_n \to \upsilon$ uniformly as $n \to \infty$.
Let $f_{n,t} = g_{\gamma_n [0,t]}^{-1}$.
Since
\begin{equation}\label{eq: def hyp geod}
F_n(t,y) = f_{n,t}( W_n \circ \upsilon_n(t) + i\, y)
\end{equation}
and since by Lemma~\ref{lm: equicontinuity of loewner chains}
\begin{equation}
|f_{n,t}(z) - f_{n,t'} (z')| \leq C(\delta,\upsilon(T)) \, 
   \left( |\upsilon_n(t)-\upsilon_n(t')| + |z - z'| \right) ,
\end{equation}
it follows directly from the assumptions that if for some $s < t$, $\upsilon(s) = \upsilon(t)$, then
$F(s,y)=F(u,y)=F(t,y)$ for all $u \in [s,t]$ and $y \geq 0$. Especially
$\gamma(u)=F(u,0)$ is constant on the interval $u \in [s,t]$ which contradicts with the assumptions
of the lemma. Hence $\upsilon$ is strictly increasing. An application of Helly's selection theorem
gives that $\upsilon_n^{-1}$ converges uniformly to $\upsilon^{-1}$.
Therefore $\gamma_n \circ \upsilon_n^{-1}$ converges uniformly to $\gamma \circ \upsilon^{-1}$
and hence for any $\delta>0$
$(t,z) \mapsto g_{\gamma_n\circ\upsilon_n^{-1}[0,t]}$ converges to $g_t$ uniformly on the set 
\eqref{eq: lmmain domain conv g}. The convergence of $W_n$ together with standard results about
ODE's imply that $g_{\gamma_n\circ\upsilon_n^{-1}[0,t]}$, which are the solutions of the Loewner equation
with the driving terms $W_n$, converge uniformly to the solution of the Loewner equation with the driving
term $W$, see Lemma~\ref{lm: driving term convergence implies caratheodory}.
Hence $g_t$ is generated by $\gamma$ and driven by $W$.
\end{proof}

\begin{lemma}\label{lm: main lemma with equicontinuity}
Let $\gamma: [0,T] \to \C$ be continuous and not constant on any subinterval of $[0,T]$.
Let $\gamma_n: [0,T] \to \C$ be a sequence of simple parametrized curves such that $\gamma_n(0) \in \R$ and
$\gamma_n( \,(0,T] \,) \subset \half$. Suppose that
$\gamma_n \to \gamma$ uniformly as $n \to \infty$. If
\begin{itemize}
\item $(W_n)_{n \in \N}$ is equicontinuous and
\item there exist increasing continuous $\psi:[0,\delta) \to \R_+$ such that $\psi(0)=0$ and
$|F_n(t,y) - \gamma_n(t)| \leq \psi(y)$ for all $0 < y < \delta$ and for all $n \in \N$
\end{itemize}
then $W_n$ converges to some continuous $W$, 
$\gamma$ can be continuously reparametrized with the half-plane capacity and $\gamma \circ \upsilon^{-1}$ is driven by $W$. 
\end{lemma}

\begin{proof}
It is clearly enough to show that $(F_n)_{n \in \N}$ is a equicontinuous family of functions on $[0,T] \times [0,1]$.
The claim then follows from the previous lemma after
choosing by
Arzel\`a-Ascoli theorem a subsequence $n_j$ such that $F_{n_j}$ and $W_{n_j}$ converge. 

Let $g_{n,t} = g_{\gamma_n [0,t]}$ and $f_{n,t} = g_{n,t}^{-1}$. 
Let $\eps>0$ and choose $\delta > 0$ such that
\begin{align}
|F_n(t,y) - \gamma_n(t)| &\leq \frac{\eps}{6} \label{eq: cont small y const t} \\
|\gamma_n(t') - \gamma_n(t)| &\leq \frac{\eps}{6} \label{eq: cont gamma t}
\end{align}
when $0 \leq y \leq \delta$ and $t,t' \in [0,T]$ with $|t-t'| \leq \delta$.
Then by the triangle inequality
\begin{equation}\label{eq: lmeqc cont F}
|F_n(t',y') - F_n(t,y)| \leq \frac{\eps}{2}
\end{equation}
for all $0 \leq y,y' \leq \delta$, for all $t,t' \in [0,T]$ with $|t-t'| \leq \delta$
and for all $n \in \N$.

By \eqref{eq: def hyp geod}
and Lemma~\ref{lm: equicontinuity of loewner chains}, the family of mappings $(F_n |_{[0,T] \times [\delta,1]})_{n \in \N}$ is
equicontinuous. Hence we can choose $0 < \tilde{\delta} \leq \delta$ such that
\eqref{eq: lmeqc cont F}
for all $\delta \leq y,y' \leq 1$ with $|y - y'| \leq \tilde{\delta}$, 
for all $t,t' \in [0,T]$ with $|t-t'| \leq \tilde{\delta}$
and for all $n \in \N$.
Hence by the triangle inequality
\begin{equation}
|F_n(t',y') - F_n(t,y)| \leq \eps
\end{equation}
for all $0 \leq y,y' \leq 1$ with $|y - y'| \leq \tilde{\delta}$, 
for all $t,t' \in [0,T]$ with $|t-t'| \leq \tilde{\delta}$
and for all $n \in \N$.
\end{proof}


\subsection{Some facts about conformal mappings}

%
%

In this section, a collection of simple lemmas about normalized conformal mappings is presented. Only elementary methods are used,
and therefore it is advantageous to present the proofs here, even though they appear elsewhere in the literature.

Denote the inverse mapping of $g_K$ by $f_K$ and by $I \subset \R$ the image of $\partial K$ under $g_K$, 
i.e. $I = \overline{\{ x \in \R \,:\, \imag f_K(x)>0 \}}$. Now $f_K$ can be given by integral with Poisson kernel of upper half-plane as
\begin{equation} \label{eq: poisson}
f_K(z) = z + \frac{1}{\pi} \int_I \frac{\imag f_K(x)}{x-z} \de x .
\end{equation}
This gives a nice proof of the following fact.
\begin{lemma} \label{lm: expand}
Denote $u_+ = \max I$ and $u_- = \min I$ and $x_\pm = f_K(u_\pm)$. Assume $\half \cap K \neq \emptyset$. Then 
\begin{equation} \label{eq: fcontr}
f_K(x) < x \textrm{ when } x \geq u_+ \textrm{ and } f_K(x) > x \textrm{ when } x \leq u_-
\end{equation}
and
\begin{equation} \label{eq: gexpan}
g_K(x) > x \textrm{ when } x \geq x_+ \textrm{ and } g_K(x) < x \textrm{ when } x \leq x_- .
\end{equation}
\end{lemma}

\begin{proof}
Note that $\imag f_K(x)$ is non-negative everywhere. It is positive in a set of non-zero Lebesgue measure, otherwise
the equation \eqref{eq: poisson} would imply that $f_K$ is an identity which is a contradiction. Now the equation
\eqref{eq: poisson} implies directly the equation \eqref{eq: fcontr}. Apply $g_K$ on both sides to get the equation
\eqref{eq: gexpan}.
\end{proof}

The lemma can be used, for example, in the following way.
\begin{lemma} \label{lm: cmapinc}
Let $K \subset K'$ be two hulls. Let $x \in \R$ s.t. 
$g_K(\overline{K' \setminus K}) \cap (x,\infty) = \emptyset$,
and let $z = f_K(x)$. Then
$g_K(z) \leq g_{K'}(z)$.
\end{lemma}

\begin{proof}
Apply Lemma \ref{lm: expand} to hull $J = g_K( \overline{K' \setminus K})$ and $u = g_J(x)$.
\end{proof}

Let's introduce still one more concept. Consider now $K=[-l,l]\times[0,h]$ where $l,h>0$. The domain $\half \setminus K$ 
can be thought as a polygon
with the vertices $w_1 = -l$, $w_2= -l + ih$, $w_3= l + ih$, $w_4 = l$ and $w_5 = \infty$. The interior angles at these vertices are
$\alpha_1 = \frac{\pi}{2}$, $\alpha_2 = \frac{3\pi}{2}$, $\alpha_3 = \frac{3\pi}{2}$, $\alpha_4 = \frac{\pi}{2}$ and $\alpha_5 = 0$,
respectively.
For the last one, this has to be thought on the Riemann sphere.

Mappings from $\half$ to polygons are well-known. They are Schwarz-Christoffel mappings. In this case, when $f_K(\infty)=w_5=\infty$
all such mappings can be written in the form
\begin{equation}
f_K(z) = A + C \int^z \frac{ \sqrt{\zeta - z_2} \sqrt{\zeta - z_3} }{ \sqrt{\zeta - z_1} \sqrt{\zeta - z_4}  } \de \zeta .
\end{equation}
Here $f_K(z_i)=w_i$, $i=1,2,3,4$. So in a sense $\real A, \imag A, C$ and $z_i$ are parameters in the problem. 
Two of them can be chosen freely and the rest are determined from them. 
The branches of the square roots are chosen so that far on the positive real axis the square root
is positive and then analytic continuation is used.

In our case $f_K$ is normalized at the infinity. This fixes $C=1$ and $\real A$ so that it cancels the constant term in the expansion of the integral.
But if we are only interested in differences $f_K(z) -f_K(z')$ we don't have to care about $A$.

\begin{lemma} \label{lm: scmapest}
Let $K=[-l/2,l/2]\times[0,h]$, $h,l>0$, and let $z_i$ be as above. 
Then
\begin{equation}
z_3      =-z_2      =\frac{1}{2} l  \, \big(1+\oo(1)\big) \textrm{ and } 
z_4 - z_3= z_2 - z_1=\frac{2}{\pi}h \, \big(1+\oo(1)\big) \textrm{ as } \frac{h}{l} \to 0 .
\end{equation}
\end{lemma}

\begin{proof}
Note first that by symmetry $z_1=-z_4$ and $z_2 =-z_4$.

Denote $\lambda = z_3 - z_2$ and $\theta = z_4 - z_3$. We would like to estimate
$\lambda$ and $\theta$ in terms of $l$ and $h$.

Calculate $h=\imag (w_3-w_4)$ as an integral along the real axis
\begin{equation*}
h = \int_{z_3}^{z_4} 
\sqrt{ \frac{\zeta - z_2}{\zeta - z_1} }
\sqrt{ \frac{\zeta - z_3}{z_4 - \zeta} } \de \zeta .
\end{equation*}
Since the first factor of the integrand is a decreasing function $\zeta$, it can be bounded with the values at
the end points $z_3$ and $z_4$. After couple of variable changes, the integral of the second factor
is
\begin{equation*}
\int_{z_3}^{z_4} 
\sqrt{ \frac{\zeta - z_3}{z_4 - \zeta} } \de \zeta = \frac{\pi}{2} (z_4 -z_3)
\end{equation*}
and therefore
\begin{equation} \label{ie: hbound}
\frac{\pi}{2} \sqrt{ \frac{1}{1+\frac{\theta}{\lambda} } } \, \theta
\leq h \leq
\frac{\pi}{2} \sqrt{ \frac{1+\frac{\theta}{\lambda}}{1+2 \frac{\theta}{\lambda} } } \, \theta .
\end{equation}

Calculate $l=w_3-w_2$ as
\begin{equation*}
l = \int_{z_2}^{z_3} 
\sqrt{ \frac{(\zeta - z_2)(z_3 - \zeta)}{(\zeta - z_1)(z_4 - \zeta)} } \de \zeta 
=  2 \int_{0}^{z_3} 
\sqrt{ \frac{z_3^2 - \zeta^2}{z_4^2 - \zeta^2} } \de \zeta
\end{equation*}
The integrand is always less or equal then one. So $l \leq \lambda$. For the lower bound, note that the integrand
is a decreasing function of $\zeta$. Therefore
\begin{equation*}
\int_{0}^{z_3} \sqrt{ \frac{z_3^2 - \zeta^2}{z_4^2 - \zeta^2} } \de \zeta
\geq \zeta_0 \sqrt{ \frac{z_3^2 - \zeta_0^2}{z_4^2 - \zeta_0^2} }
\end{equation*}
Maximize this with respect to $\zeta_0 \in [0,z_3]$ to get
\begin{equation*}
\int_{0}^{z_3} \sqrt{ \frac{z_3^2 - \zeta^2}{z_4^2 - \zeta^2} } \de \zeta
\geq \left( z_4 - \sqrt{z_4^2 - z_3^2} \right)
\end{equation*}
To conclude this
\begin{equation} \label{ie: lbound}
\left( 1 + 2 \frac{\theta}{\lambda} -2 \sqrt{\frac{\theta}{\lambda}} \sqrt{ 1+\frac{\theta}{\lambda} } \right)
\, \lambda \leq l \leq \lambda
\end{equation}

The inequalities \eqref{ie: hbound} and \eqref{ie: lbound} can be combined to conclude that $\frac{\theta}{\lambda}$ is
small when $\frac{h}{l}$ is small. And in this case $\theta \approx \frac{2}{\pi} h$ and $\lambda \approx l$. And all the
claims follow.
\end{proof}

\begin{lemma} \label{lm: exitpointest}
Let hull $K$ be a subset of a rectangle $[-l,l] \times [0,h]$, $l,h>0$. If $K \cap \, (\{l\} \times [0,h])\neq \emptyset$ then
uniformly for any $z$ in this set $g_K(z) = l\big( 1+ \oo(1) \big)$ as  $\frac{h}{l} \to 0$.
\end{lemma}

\begin{proof}
Assume that $K \cap \half \neq \emptyset$. Otherwise the statement is trivial since $z = l$ and $g_K$ is identity.

Let $K' = [-l,l] \times [0,h]$. Then $K \subset K'$. Take any $z \in K \cap \, (\{l\} \times [0,h])$.
Let $x_+ = l$, $u_+=g_K(x_+)$ and $v_+=g_{K'}(x_+)$.

By Lemma \ref{lm: cmapinc} and Lemma \ref{lm: scmapest} $l \leq u_+ \leq v_+ = l \big( 1+ \oo(1) \big)$. 
And by an length area principle, for example, Wolff's lemma (Proposition~2.2 of \cite{pommerenke-1992-}),  
$0 \leq u_+ - g_K(z) = \oo(1)  \, l$.%
\end{proof}

\begin{lemma} \label{lm: hcaplb}
If $i \in K$ then $\hcap (K) \geq \frac{1}{4}$.
\end{lemma}

\begin{proof}
First of all note that this is sharp. It is attained by a vertical slit extending from $0$
to $i$.

Now assume that there is $K$ containing $i$ s.t. $\hcap (K) < \frac{1}{4}$. It is possible to choose
$\tilde{K}$ containing $K$ s.t. the capacities are arbitrarily close and the boundary of $\tilde{K}$
is a smooth curve. This can be done by choosing a smooth, simple curve $\gamma$ that separates an interval
containing $g_K(K)$ from $\infty$ in $\half$. Then $\tilde{ K}$ is the hull that has $f_K(\gamma)$ as the boundary.
Therefore there exists now $\tilde{K}$ s.t. it contains $i$, $\hcap (\tilde{K}) < \frac{1}{4}$ and
the boundary is a curve.

Therefore, there exists a simple curve $\gamma(t)$, $t \in [0,T]$, parameterized by the capacity so that $0<T<\frac{1}{4}$ and
$\gamma$ contains some point lying on the line $i + \R$. Now take any point $z$ s.t. $\imag z > 4 T$, and let
$Z_t = X_t + i Y_t = g_t(z)$. Then by Loewner equation
\begin{equation*}
\frac{\de Y_t}{\de t} = - \frac{2 Y_t}{(X_t - U_t)^2 + Y_t^2} \geq - \frac{2}{Y_t} .
\end{equation*}
Therefore
\begin{equation*}
Y_t \geq \sqrt{ \big( \imag  z \big)^2 - 4 t} >0 .
\end{equation*}
Hence $z \notin \gamma[0,T]$. This leads to a contradiction: $\gamma$ doesn't intersect the line $i + \R$.
\end{proof}

\begin{lemma} \label{lm: hcap} Let $K$ be a hull.
If $K \cap \big( \R \times \{h i\} \big) \neq \emptyset$
then $\hcap (K) \geq \frac{1}{4} h^2$.
If $K \subset [-l,l]\times[0,h]$, then $\hcap (K) \leq \hcap\big( [-l,l]\times[0,h] \big)$ and 
$\hcap\big( [-l,l]\times[0,h] \big)= \frac{1}{2\pi} h l \big( 1+ \oo(1) \big)$ as $\frac{h}{l} \to 0$.
\end{lemma}

\begin{proof}
The lower bound follows from Lemma \ref{lm: hcaplb} and scaling.

For the upper bound let's use the Schwarz-Christoffel mapping. Write
\begin{align}
\hcap(K) &=  \frac{1}{8} \big( -z_1^2 - z_4^2 + z_2^2 + z_3^2  \big) = \frac{1}{4} \big( z_4 - z_3 \big)\big( z_4 + z_3 \big) \nonumber \\
  & = \frac{1}{2\pi} h l \big( 1+ \oo(1) \big)
\end{align}
This gives the desired upper bound.
\end{proof}